\documentclass[final,1p,times,numbers, xcolor=dvipsnames]{elsarticle}

\usepackage{booktabs}
\usepackage{siunitx} 
\usepackage{graphicx}
\usepackage{comment}
\usepackage{verbatimbox}
\usepackage[dvipsnames]{xcolor}
\usepackage{amssymb}
\usepackage{placeins}
\usepackage{amsthm}
\newtheorem{theorem}{Theorem}[section]
\newtheorem{lemma}[theorem]{Lemma}

\newtheorem{proposition}[theorem]{Proposition}
\usepackage{pdflscape}\usepackage{amsmath}
\newtheorem{definition}{Definition}[section]

\usepackage{url}
\usepackage{hyperref}

\begin{document}

\begin{frontmatter}
 \title{Ghost tasking for parametrized Gaussian Processes solving linear differential equations}

 \author[label1]{Johanna Moser\corref{cor1}}
\ead{j.moser@tugraz.at}

\author[label1]{Christopher Albert}

\author[label2,label3]{Sascha Ranftl\fnref{label4}}

\cortext[cor1]{Corresponding author}
\fntext[label4]{Current address: School of Mechanical Engineering, Purdue University, USA.}

\affiliation[label1]{organization={Institute of Theoretical \& Computational Physics, University of Technology},
            city={Graz},
            country={Austria}}

\affiliation[label2]{organization={Courant Institute, New York University},
            city={New York},
            country={USA}}

\affiliation[label3]{organization={Division of Applied Mathematics, Brown University},
            country={USA}}

\begin{abstract}
Physics-informed machine learning has gained significant attention in recent years. In regimes of limited data, parametrized Gaussian processes have become popular. Existing approaches, however, often face limitations, such as requiring parametrizable (also called controllable) systems or a large number of output tasks. In this work, we introduce a systematic procedure we call "ghost tasking", using auxiliary tasks to circumvent these limitations. We prove that such ghost tasks can render any non-parametrizable system effectively parametrizable, enabling algorithmic construction of parametrized Gaussian Processes while keeping the number of required tasks (i.e. output dimensions) and latent functions low. We find that ghost tasking performs especially well in an inverse problem setting, even with very few available data. We show the usage and power of ghost tasking in three experiments, providing systematic comparisons to the only other currently available method applicable to all experiments. We provide necessary syntax and explications for two computer algebra programs that compute parametrizations for systems with polynomial or rational coefficients. Our theoretical results extend to systems with 
meromorphic functions. 
\end{abstract}

\begin{graphicalabstract}
\begin{figure}
\centering
\includegraphics[width=0.8\textwidth]{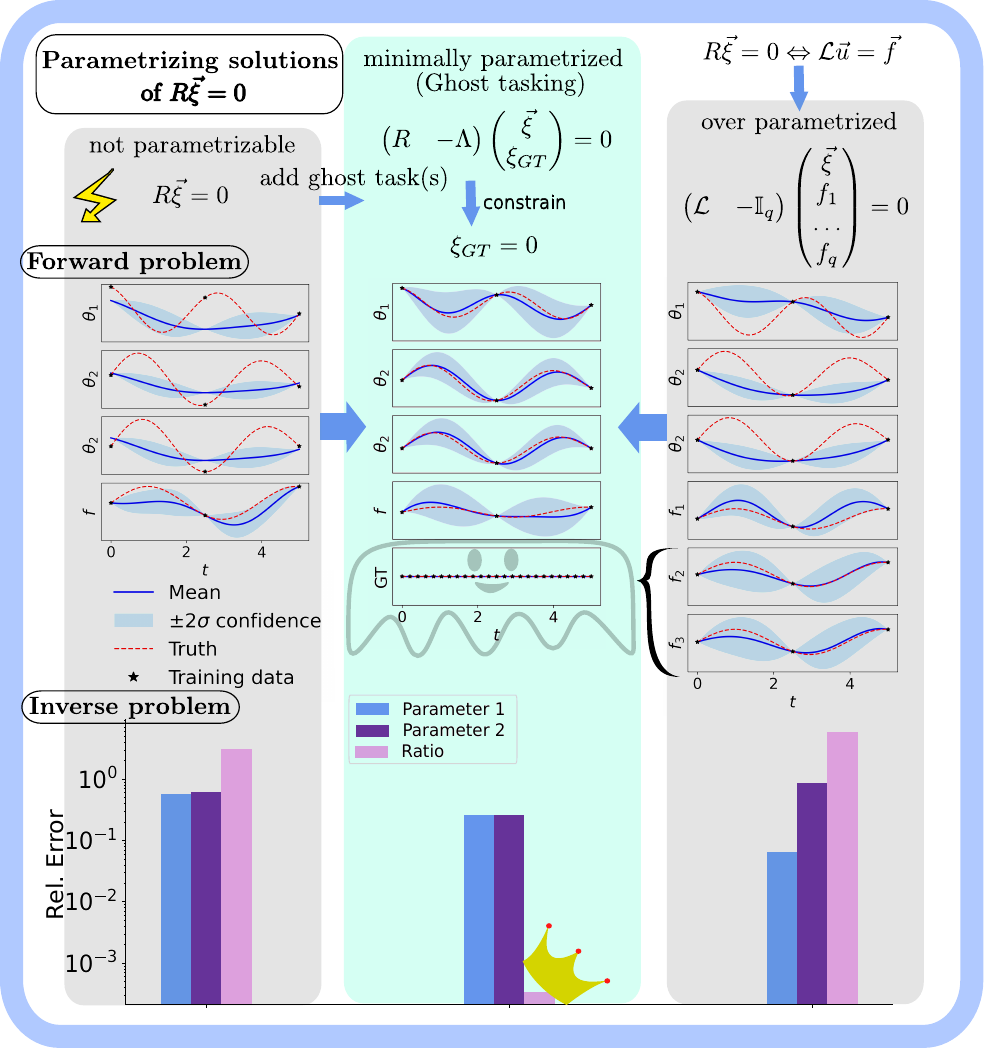}
\caption{Visualization of ghost tasking as the minimal amount of tasks necessary to depict any linear system of differential equations through parametrization. On top, we show the equations each system is solving, in the middle we show the forward fit of a physical system and below, the the mean absolute error of the two learned physical parameters and their ratio relative to the true values in an inverse problem.}
\label{fig:graphical_abstract}
\end{figure}
\end{graphicalabstract}

\begin{highlights}
\item Ghost tasking balances accuracy and applicability in physics-informed GPs
\item We connect the physics-informed GP methods of Raissi (2017) to Lange-Hegermann (2018)
\item We prove ghost tasking valid for full-row-rank systems with meromorphic coefficients
\item Ghost tasking improves accuracy over Raissi (2017) in low-data inverse problems
\item Computer algebra tutorials included for OreModules and Macaulay2
\end{highlights}


\end{frontmatter}

\section{Introduction}
\label{sec:intro}
Understanding complex systems from data typically relies on either data-driven models or physics-based formulations. Hybrid approaches have emerged under the name physics-informed machine learning, aiming to integrate physical structure directly into learning systems.
A central design choice in such hybrid models is how physical laws are enforced. Soft constraint approaches incorporate differential equations as penalties, whereas hard constraint methods enforce them exactly. In this work, we focus on hard constraints within the Gaussian process (GP) framework, which guarantees exact solutions of linear differential equations through the structure of the covariance function (also called kernel) while retaining principled uncertainty quantification capabilities.

\subsection{Related Work}
Research on incorporating physically motivated constraints into Gaussian processes has accelerated considerably in recent years; a survey of methods as of 2020 can be found in \cite{swiler2020survey}. We provide here a brief and non-exhaustive orientation.

The first class of methods capable of enforcing hard constraints are parametrized Gaussian Processes, which we will focus on in this work. They construct covariance functions by exploiting the closure of Gaussian processes under linear transformations, applying linear differential operators directly to base kernel functions. The earliest works in this direction are \cite{dong_kriging, dong_diss, narcowich} which have subsequently been applied in \cite{ wahlstrom0, wahlstrom2015, berlinghieri2023gaussian} and rendered algorithmic in \cite{jidling2017linearly, raissi_machine_2017, lange-hegermann_algorithmic_2018, besginow_constraining_2022, harkonen_gaussian_2023}.

Apart from the parametrization-based methods, another line of work that enforces linear differential equations exactly is represented by Mercer-type GPs whose kernels are directly built from fundamental solutions of the differential equations at hand. This approach has been proposed in \cite{boogaart} and has been re-discovered and applied to different equations repeatedly since \cite{alvarez_linear_2013, albert_gaussian_2019}. They however have the disadvantage that we need to know the fundamental solution of the specific problem at hand beforehand, making them, in contrast to parametrization-based GPs, non-algorithmic. 

Similar approaches have been proposed to enforce boundary conditions through specific kernel choices and/or mean-function constructions \cite{ding2019bdrygp, tan2018gaussian}. Harmonic-feature representations based on Laplacian eigenfunctions on irregular domains \cite{pmlr-v89-solin19a} have also been proposed with the additional advantage of providing computationally efficient low-rank GP representations. Fulfilling boundary conditions can also be combined with fulfilling differential equations inside of the domains and has been explored for example in \cite{pmlr-v130-lange-hegermann21a, gulian2022gaussian}.

Finally, especially for nonlinear PDEs, physics is often incorporated during conditioning as a soft constraint rather than through exact prior construction. Numerical Gaussian process methods using Euler integration schemes and collocation-based formulations enforce differential equations at selected points and have been applied to a broad range of nonlinear problems \cite{chen_nonlinear, raissi_nonlinear, hamelijnck2024physics}, while recent work has further generalized GP conditioning to flexible classes of constraints and observations \cite{moss2026conditioning}.

\subsection{Contribution}
The present work focuses on parametrization-based GP models that incorporate linear differential operators in the covariance function. Parametrization-based methods provide the strongest combination of exact constraint satisfaction and algorithmic construction, yet they remain subject to important limitations.
The method of Raissi et al. \cite{raissi_machine_2017} is applicable to a broad class of linear differential equations, but can lack robustness in low-data regimes \cite{mymaxent}. More efficient techniques however rely on restrictive assumptions such as parametrizability \cite{lange-hegermann_algorithmic_2018}, or the coefficients in the differential equations being constant \cite{besginow_constraining_2022, harkonen_gaussian_2023}.  We address these limitations by introducing \emph{ghost tasking}, a construction that generalizes parametrization-based GP models to traditionally not parametrizable systems, including those with non-constant coefficients, while preserving performance in forward and inverse problems. Therein, ghost tasks are an auxiliary (ghost) dimension in the GP output; i.e. the a non-parametrizable system is lifted to a higher dimension where it becomes parametrizable. After parametrizing the system in this higher dimension, we can conveniently convert the GP back to the original lower dimension by conditioning the ghost tasks to zero. 

A detailed conceptual comparison with existing methods is deferred to chapter~\ref{sec:discussion}.\\

The main idea of ghost tasking is to (1) embed an  non-parametrizable system represented by its operator matrix in a higher-dimensional space by adding ghost tasks, (2) compute the now fully parametrizable null space of the embedded operator matrix, (3) apply the null space as a linear operator to parametrize a base kernel of a GP and (4) re-constrain the embedded parametrized GP to the original solution space by conditioning the ghost tasks to be point wise zero. 
As discussed in \cite{lange-hegermann_algorithmic_2018}, the computation of operator null spaces can be performed very efficiently with computer algebra programs. Two open source packages that work well are \texttt{OreModules} \cite{oremodules} and \texttt{Macaulay2} \cite{macaulay2}.\\

Our contributions are the following:
\begin{enumerate}
\item We identify the work of Raissi et al. (2017) \cite{raissi_machine_2017} and Lange-Hegermann (2018)\cite{lange-hegermann_algorithmic_2018} regarding \lq\lq physics-informed GPs\rq\rq and GPs for differential equations more generally as belonging to the same methodological spectrum. 
\item We introduce ghost tasking as the optimum on said spectrum, bridging the gap between accuracy and applicability. 
\item We prove the theoretical validity of ghost tasking for all systems with full row-rank, and polynomial, rational, or generally meromorphic coefficients
\item We show improved accuracy, especially in low data settings, for inverse problems in 3 experiments and systematic comparisons to Raissi et al. \cite{raissi_machine_2017}
\item Our method is fully compatible with the already freely available python package \texttt{PCGP} \cite{mymaxent} and we include tutorials for using ghost tasking with the computer algebra programs \texttt{OreModules} \cite{oremodules} and \texttt{Macaulay2} \cite{macaulay2}.
\end{enumerate}

The remaining paper will be structured in the following way: in chapter \ref{sec:basics} we first introduce Gaussian Processes and subsequently discuss two state of the art methods for differential equations with non-constant coefficients to parametrize GPs \cite{raissi_machine_2017, lange-hegermann_algorithmic_2018} in chapter \ref{subsec:paramGP}. In chapter \ref{subsec:algebraintro}, we give an introduction to the algebraic concepts that we use in chapter \ref{subsec:params} to analyze the methods introduced in chapter \ref{subsec:paramGP} with regards to parametrizability. Our theoretical contributions are introduced in chapter \ref{sec:theorem}. In chapter \ref{subsec:mlhraissi}, we show how \cite{raissi_machine_2017} and \cite{ lange-hegermann_algorithmic_2018} are related and use these findings to develop ghost tasking in chapter \ref{sec:theorem}. We analyze ghost tasking theoretically in chapter \ref{subsec:ghostprops}. In chapter \ref{sec:experiments} we show three different experiments applying ghost tasking. Each of the experiments contains an algebraic part where we show the application of ghost tasking, followed by a numerical implementation and evaluation of the problem in a forward and inverse setting, and systematic performance comparison to \cite{raissi_machine_2017}, since this is the only available GP parametrization method next to ghost tasking that is general enough so that it can be applied to all proposed experiments. Finally, we discuss practical aspects of ghost tasking and place it within the context of the landscape of currently available methods in chapter \ref{sec:discussion}. In \ref{ap:macaulay} we provide and explain code snippets for the practical use of the computer algebra programs \texttt{OreModules} \cite{oremodules} and \texttt{Macaulay2} \cite{macaulay2} for ghost tasking. Ablation studies on the ghost tasking method can be found in \ref{ap:ablation}, and an example for commutative cases as discussed in chapter \ref{sec:discussion} is provided in \ref{app:commutative}.

\section{Background} 
\label{sec:basics}
\subsection{Gaussian Processes}
\newcommand{\vv}{\mathbf}
Gaussian Processes are a flexible framework for regression function approximation that naturally incorporates uncertainty. For a comprehensive introduction, see e.g. \cite{rasmussen_gaussian_2006}.
Given an input set ${\cal X} \subset \mathbb{R}^d$ with elements $\vv x, \vv x' \in {\cal X}$, a GP $g \sim {\cal GP}(m(\vv x), k^{\vartheta}(\vv x,\vv x'))$ is defined by a mean function $m(\vv x)$ and a kernel function $k^{\vartheta}(\vv x, \vv x')$ with hyperparameters $\vartheta$. For any finite set of inputs ${\cal X}^n \ni X = (\vv x_1, ..., \vv x_n)$, $g(X)$ follows a multivariate normal distribution with mean $\vec{\mu} = m(X)$ whose entries are $\mu_i = m(\vv x_i)$ and covariance matrix $\Sigma$ frequently written as $k^{\vartheta}(X, X)$ since the entries of the covariance matrix are defined through the kernel function as $\Sigma_{ij} = k^{\vartheta}(\vv x_i, \vv x_j)$. 
Given $n$ noisy observations $Y = (y_1,\cdots ,y_n)^T$ with $y_i = g(\vv x_i) + \varepsilon_i$ at the points $X$ with independently distributed measurement noise terms $\varepsilon_i \sim {\cal N}(0, \sigma^2)$, conditioning yields again a GP whose mean and covariance matrix at input new points $X_*$ are obtained by the standard multivariate Gaussian conditioning formulas
\begin{align}
	&m_{|X, Y}(X_*) = m(X_*) + k^{\vartheta}(X_*, X)\left(k^{\vartheta}(X, X) + \sigma^2\mathbb{I}\right)^{-1}(Y-m(X)) \nonumber\\
	&k_{|X,Y}(X_*, X_*) = k^{\vartheta}(X_*, X_*) - k^{\vartheta}(X_*, X)(k^{\vartheta}(X, X)+\sigma^2\mathbb{I})^{-1}k^{\vartheta}(X, X_*). 
\end{align}

The kernel encodes prior beliefs about our function such as smoothness, differentiability, and, through $\vartheta$, hyperparameters such as relevant lengthscales. 
The hyperparameters $\vartheta$ are generally chosen by minimizing the negative log marginal likelihood, i.e. for a zero-mean prior
\begin{equation}-\log\left(p(Y|\vartheta, X, \sigma)\right) = \frac{1}{2} Y^T (k^{\vartheta}(X, X)+\sigma^2\mathbb{I})^{-1}Y + \frac{1}{2} \log\left(\det (k^{\vartheta}(X, X)+\sigma^2\mathbb{I})\right) + \frac{n}{2}\log(2\pi)
\label{eq:mll}
\end{equation}
consisting of a data-fit term ($Y^T (k^{\vartheta}(X, X)+\sigma^2\mathbb{I})^{-1}Y$), a complexity penalty ($\log( \det (k^{\vartheta}(X, X)+\sigma^2\mathbb{I}))$), and a constant term to normalize the likelihood.
If we have prior beliefs about the hyperparameters $p(\vartheta)$ we can simply add $-\log(p(\vartheta))$ to \eqref{eq:mll} to obtain the maximum a posteriori (MAP) estimates of $\vartheta$.
Just like their finite dimensional analogues, multivariate Gaussian distributions, Gaussian Processes are also closed under linear transformations $T$, i.e. linearly transforming the output of a GP corresponds to sampling from another GP whose mean and kernel are transformed, formally 
\begin{equation}
    Tg \sim {\cal GP}(Tm(\vv x), Tk^{\vartheta}(\vv x, \vv x') T')\,
\end{equation} 
where $T$ acts on $\vv x$ and $T'$ acts on $\vv x'$.

For multidimensional outputs, we can define the so-called "multitask GP" by stacking the outputs (also called "tasks") in the mean vector and covariance matrix. This multitask structure and GPs' closedness under linear transforms gives us the flexibility to build GPs whose outputs are coupled by a given linear system of differential equations, rather than treating each output as an independent GP. This allows us to encode physical constraints directly into the kernel of the model. 

\subsection{Parametrized Gaussian Processes}
\label{subsec:paramGP}
We now discuss two particular state-of-the-art methods for incorporating linear differential equations into GPs. To our knowledge, they are the only two methods in the literature that are both algorithmic and capable of parametrizing GPs for learning solutions of differential equations with non-constant coefficients. We use the umbrella term "parametrized GPs" to refer to them collectively.
A linear system of differential equations can be written as
\begin{equation}
	{\cal L}^\phi(\vv x) \vec{u}(\vv x) = \vec{f}(\vv x) \quad \text{for all } \vv x \in \mathcal{X}.
    \label{eq:system1}
\end{equation}
Here, $\vec{u}(\vv x)$ is a vector containing all unknown functions (i.e. the solution), and $\vec{f}(\vv x)$ contains the inhomogeneities. ${\cal L}^\phi(\vv x)$ is a matrix containing a set of physical parameters $\phi$ and differential operators $\partial_{\vv x}$ with coefficient functions depending on $\vv x$. \\

Differential operators acting on functions of variables $ x_1,..., x_n$ can be organized into algebraic rings. Depending on the class of coefficient functions allowed, one obtains different rings -- for example rings of differential operators with meromorphic or rational coefficient functions. Of particular practical interest is the case of polynomial coefficients, leading to Weyl algebras, which are especially convenient for computer algebra. We will use them for our experiments in section \ref{sec:experiments}. For describing differential equations in $n$ variables, we define the $n$-th Weyl algebra $A_n(\mathbb{R})$ as consisting of elements that are finite linear combinations of terms $\Pi_{i = 1}^{n}x_i^{k_i}\partial_i^{l_i}$ with $k_i, l_i \in \mathbb{N}$ where the generators satisfy the commutation relations 
\begin{align}
[\partial_i, x_i]_{-} := \partial_i x_i-x_i\partial_i &= 1 \text{ for all } i = 1,....n\nonumber \\ 
[x_i, x_j]_{-} = [\partial_i, \partial_j]_{-} &= 0 \text{ for all } i, j\nonumber \\ 
[\partial_j, x_i]_{-} &= 0\text{ for }i \neq j,
\label{eq:commutator}
\end{align}
i.e. $\partial_i$ acts as differentiation with respect to $x_i$. A similar structure can be used for rational coefficients, giving the ring $B_n(\mathbb{R})$.

All results in this paper apply to a class of noncommutative rings that are are Noetherian domains, so-called "very simple domains", and have finite global dimension \cite{stafford}. Concretely, we state the results for Weyl algebras $A_n(\mathbb{R})$ and $B_n(\mathbb{R})$ (for polynomial and rational coefficients), rings with meromorphic coefficients, and, in the ODE case, rings with holomorphic (formal or locally convergent power series) coefficients. We will simply denote such a ring by $D$ throughout the paper unless stated otherwise.\\

The most prominent version to parametrize a Gaussian Process using operators in $D$ was introduced by \cite{raissi_machine_2017}. We will briefly summarize the method here: 
Let us assume the unknown functions $\vec{u}$ in Eq.~\ref{eq:system1} can be expressed through a "base GP" or "latent GP" $u(\vv x)\sim {\cal GP}(m_u(\vv x), k_{uu}(\vv x, \vv x))$. In general, this base GP has a base mean $m_u(\vv x)$ conventionally chosen to be zero and a sufficiently continuously differentiable base kernel $k_{uu}(\vv x, \vv x')$. A common choice for the latter is the squared-exponential kernel $k_{uu}(\vv x, \vv x') = A \exp \left(-\frac{||\vv x - \vv x'||^2}{2l}\right)$ that is infinitely often differentiable and has two hyperparameters $\vartheta$, the lengthscale $l$ and the scale factor denoted by $A$. If $\vec{u}(\vv x)$ has more than one output dimension, the base kernel can consist of multiple independent square-exponential kernels stacked in multitask fashion, with either shared or independent hyperparameter values. 
We can then exploit the linearity of ${\cal L}^{\phi}(\vv x)$ and model $\vec{u}(\vv x)$ and the inhomogeneities $\vec{f}(\vv x)$ together as a multitask GP of the form 
\begin{align}
	\begin{pmatrix}
	    \vec{u}(\vv x)\\ \vec{f}(\vv x)
	\end{pmatrix} \sim {\cal GP}\left(\vec{0}, \begin{pmatrix} k_{uu}(\vv x, \vv x') & {\cal L}^{\phi T}(\vv x') k_{uu}(\vv x, \vv x') \\
		{\cal L}^\phi(\vv x) k_{uu}(\vv x,\vv x') & {\cal L}^\phi(\vv x)k_{uu}(\vv x, \vv x'){\cal L}^{\phi T}(\vv x')\end{pmatrix}\right),
        \label{eq:PIGP}
\end{align}
where the transpose ${\cal L}^{\phi T}(\vv x')$ is applied to the second argument of the base kernel.
By conditioning on data points of $\vec{f}(\vv x)$ we can infer $\vec{u}(\vv x)$, which is called a \emph{forward problem}, and if we know data from both $\vec{f}(\vv x)$ and $\vec{u}(\vv x)$, we can learn the probability distribution of the physical parameters $\phi$ through hyperparameter optimization or sampling, which is referred to as the \emph{inverse problem}.
In analogy to physics-informed neural networks \cite{pinn}, we will refer to this method as "physics-informed Gaussian Processes" (PIGPs) and use them as comparison baseline for our method. \\

A downside of PIGPs is that they require a lot of tasks: for a system with $q$ unknown functions we need to model both the functions and their inhomogeneities, so we need $2q$ tasks, which is unfavorable in terms of the scaling of the computational effort with the number of dimensions. Accordingly, we also need $q$ independent base kernels, which makes the latent function space for $\vec{u}(\vv x)$ oftentimes unnecessarily large and thus requires more data to be sufficiently constrained. 
\cite{jidling2017linearly, lange-hegermann_algorithmic_2018} introduced an approach that is able to encode a system in a more compact way, making the latent function space smaller and rendering the approach more data efficient and more accurate in the inverse problem (\cite{mymaxent}). To explain this method, we will rewrite \eqref{eq:system1} in the form:\\
\begin{equation}
    R^{\phi}(\vv x)\vec{\xi}(\vv x) = 0 \quad \text{for all } \vv x \in \mathcal{X},
    \label{eq:system}
    \end{equation}
where $R^{\phi}(\vv x)\in D^{q \times p}$ is called the system matrix. As before, $\phi$ denotes physical parameters and the elements of $R^\phi(\vv x)$ lie in an operator ring $D$. Both the unknown functions and any inhomogeneities are components of the  $p$-dimensional vector-valued output function $\vec{\xi}(\vv x)$ with $p\geq q$. \\

The method of Lange-Hegermann \cite{lange-hegermann_algorithmic_2018} proposes to use the null space of $R^\phi(\vv x)$ in the operator space $D$ to solve the differential equation: If $Q^\phi(\vv x) \in D^{p \times n}$ denotes the null space of $R^\phi(\vv x)$, i.e. the matrix whose $n$ column vectors $Q^\phi_{\cdot,i}(\vv x)$ fulfill $R^\phi (\vv x)Q^\phi_{\cdot,i}(\vv x) = 0$, then for any set of (sufficiently differentiable) functions $g_i$, the function $\vec{\xi}(\vv x) := \sum_{i = 1}^nQ^\phi_{\cdot,i}(\vv x)g_i = Q^{\phi}(\vv x)\vec{g}(\vv x)$ automatically fulfills 
\[R^{\phi}(\vv x)\vec{\xi}(\vv x) = R^{\phi}(\vv x)\big(Q^{\phi}(\vv x)\vec{g}(\vv x)\big) = \big(R^{\phi}(\vv x)Q^{\phi}(\vv x)\big)\vec{g}(\vv x) = 0\vec{g}(\vv x) = 0,\]
and is therefore a solution to \eqref{eq:system}. 

Based on this observation, the author of \cite{lange-hegermann_algorithmic_2018} shows that we can build a Gaussian Process whose realizations $\xi(\vv x)$ are dense in the solution space of Eq.~\ref{eq:system} by transforming a base GP model $\vec{g}(\vv x)$ with this null space operator $Q^\phi(\vv x).$ We will refer to this method as "null space-based" (NSB) method. The NSB Gaussian Process has the form 
\begin{equation}
\vec{\xi}(\vv x) \sim \mathcal{GP}\left(Q^{\phi}(\vv x)m_0(\vv x), Q^{\phi}(\vv x)k_0(\vv x,\vv x')Q^{\phi T}(\vv x')\right).\end{equation}
As in PIGPs, the mean is frequently chosen to be zero. 
This construction requires one task per row of $Q^{\phi}(\vv x)$ and one latent base kernel per column of $Q^{\phi}(\vv x)$. For $q$ unknown functions,
$Q^{\phi}(\vv x)$ has at most $n\leq q$ columns and $q+1 \leq p \leq 2q$ rows, depending on how many components of $\vec{\xi}(\vv x)$ are needed to model the inhomogeneities. So, this method never needs more tasks or latent functions than PIGPs, and typically needs fewer of both. This results in general in a stronger coupling between the remaining tasks which in turn can help improve performance on the inverse problem of learning physical parameters as we have shown in \cite{mymaxent}. 

The NSB method \cite{lange-hegermann_algorithmic_2018} is applicable whenever a system is \textit{parametrizable}. 
Otherwise, finding a $Q^\phi$, such that \(R^\phi Q^\phi = 0\), only guarantees that every function generated through \(Q^\phi\) is a solution of the differential equation. It does not yet guarantee the converse, namely that every solution of the differential equation can be generated through \(Q^\phi\). In other words,
    $\operatorname{im}(Q^\phi) \subseteq \ker(R^\phi)$
always holds, whereas parametrizability ensures the stronger equality
    $\operatorname{im}(Q^\phi) = \ker(R^\phi)$.
Therefore, for a non-parametrizable system, the parametrized GP is only able to realize certain solutions, but not all solutions, or can, in the (frequent) extreme case, only produce the trivial solution.  We make this distinction precise in Sections~2.3--2.4.

\subsection{A short introduction to some algebraic concepts}
\label{subsec:algebraintro}
We will give a minimal introduction to necessary algebra in this section. Our work, however, remains focused on physics-informed machine learning, and not on algebra, so we intend to avoid overwhelmingly detailed definitions and provide just enough context to develop our method and for the reader to apply our results in a practical setting. For a detailed presentation of the rich topic of homological algebra, we refer interested readers for example to \cite{cours, dmodules, rotmancours}. We have already introduced the class of rings containing our differential operators that we will work over as $D$ in the last section. We  keep the notation for the system matrix $R$ as in Eq.~\ref{eq:system}, but drop parameter specifications $\phi$ and the argument $(\vv x)$ for better readability throughout the whole section. 

Firstly, we need to introduce multidimensional objects over rings. Just like vector spaces can be spanned over fields, modules can be thought of as the analogue over rings $D$. Because the rings we consider are in general non-commutative, we distinguish between left and right modules, depending on which side elements act from. Because we have fixed commutation relations (such as Eq.~\ref{eq:commutator} in the Weyl algebras), we can change between left and right modules easily via a mapping called involution (used frequently in computer algebra programs), so most arguments and definitions hold for both left and right modules, which is why we drop the specification at times. 

Given a left module $M$, we define some characteristics: 
\begin{definition}[\cite{cours, rotmancours}]
    We say 
    \begin{itemize}
        \item $M$ is \emph{stably free} if there exists $r, s \in \mathbb{N}_0$, such that $M\oplus D^{1\times s} \cong D^{1\times r}$, where $\oplus$ denotes the direct sum of modules. 
        \item $M$ is \emph{projective} if there exists $r \in \mathbb{N}_0$ and a left $D$-module $P$ such that $M\oplus P \cong D^{1\times r}.$
        \item $M$ is \emph{torsion-free} if the torsion left $D$-submodule defined by 
        \[t(M) := \{m \in M | \exists d \in D \setminus \{0\}: dm = 0 \}\]
        is reduced to 0, i.e. $t(M) = \{0\}$.
        \item $M$ is \emph{torsion} if $t(M) = M.$ 
    \end{itemize}
    Similar definitions hold for right modules.
    \label{def:torsionundco}
\end{definition}
If $M$ is stably free, it is automatically also projective, and if $M$ is projective, it is also torsion-free \cite{rotmancours}. If $M$ is not torsion-free, we say it \textit{has} torsion if $0 \neq t(M) \subset M$, or \textit{is} torsion, if $t(M) = M$. 

The following theorems hold for modules that are torsion:
\begin{theorem}[Stafford's theorem \cite{staffordog} \cite{stafford}, adapted from Thm.~9, Cor.~2; cyclicity part adapted from {\cite[Thm.~4--5]{serre}}]
\label{th:staffordserre}
Let $M$ be a finitely generated torsion left (resp.\ right) $D$-module, with $D$ defined as in chapter \ref{subsec:paramGP}.
\begin{enumerate}
    \item[(1)] $M$ can be generated by two elements. 
    \item[(2)] If, in addition, $D$ is a ring of \emph{ordinary} (i.e.\ single-variable) differential operators with polynomial, formal power series, or locally convergent power series coefficients, then $M$ is cyclic, i.e.\ can be generated by a single element.
\end{enumerate}
\end{theorem}

To study differential equations in operator space (i.e. over $D$-modules), we need to study the left quotient row-module of the differential equation characterized by the system matrix $R$, defined as
\begin{equation}
    M := D^{1\times p}/(D^{1 \times q}R)
\end{equation}
i.e. $D^{1\times p}$ modulo the relations defined in the rows of $R$. 
In order to characterize $M$, we furthermore need to define its \textit{Auslander transpose} $N := D^{q}/(RD^{p})$, and so called \textit{extension modules}, written ext$^i_D(M, D)$. These measure obstructions in chains of mappings from and to $M$, but only depend on the module $M$ itself, not on the particular choice of the mapping. A more rigorous introduction based on algebraic diagrams can be found for example in \cite{cours}. For full row-rank matrices $R$ (\cite{serre}), it holds that
\begin{equation}
    ext^1_D(M,D)\cong D^q/(RD^p) = N.
    \label{eq:ext}
\end{equation}
\begin{lemma}[{\cite[Prop.~2.2.1]{amc1311}}]
\label{lem:torsion}
Let $M$ be a finitely generated right (resp.\ left) $D$-module. Then for $i \ge 1$ the modules $\mathrm{ext}^i_D(M,D)$ are either zero or finitely generated torsion left (resp.\ right) $D$-modules.
\end{lemma}

With $M$, $N$, and extension modules introduced, we can check characteristics of $M$ in the following ways (for proofs we refer to cited literature):
\begin{lemma}[\cite{cours}]
Let $M = D^{1\times p}/(D^{1 \times q}R)$ be the left $D$-module finitely presented by the full row-rank matrix $R \in D^{q\times p}$, and $N$ its Auslander transpose. Then
\begin{enumerate}
   \item  $M$ is torsion-free $\Leftrightarrow$ ext$^1_D(N,D) = 0$
   \item $M$ is projective $\Leftrightarrow$ ext$^i_D(N, D) = 0$ for all $i = 1,2,3,4....$
   \item $M$ is projective $\Leftrightarrow$ $R$ has a right inverse $S$, such that $RS = \mathbb{I}_q$
   \item $M$ is stably free $\Leftrightarrow$ ext$^1_D(M, D) = 0$
\end{enumerate}
\label{lem:exteigenschaften}
\end{lemma}

To study solutions of differential equations - i.e. functions - we need not only the operator characterized above through modules, but also a function space ${\cal F}$ over which we want to find these solutions. Formally, ${\cal F}$ is also a left $D$-module, and depending on the chosen ring $D$, ${\cal F}$ can be \textit{injective}, or not. As a special case, it is well known that for a ring of differential ODE or PDE operators with \textit{constant} coefficients, $C^{\infty}$ is an injective function space \cite{oberst, gradefiltration}. 
However, for the more general case of $D$ being a Weyl algebra or having general meromorphic coefficients, considered here, standard $C$ function spaces including $C^{\infty}$ are not injective, which has consequences on parametrizability.

\subsection{Algebraic parametrization}
\label{subsec:params}
\begin{figure}[h!]
\centering
  \includegraphics[width=\textwidth]{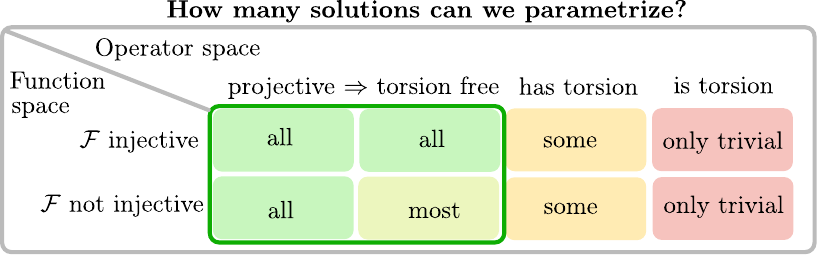}
\caption{Overview over parametrizability in operator space $D$ and function space ${\cal F}$. Every combination can be read as "Systems represented by modules in operator space that are [i.e. torsion-free] can parametrize [i.e. all] solutions over the function space ${\cal F}$ if [i.e. ${\cal F}$ is injective]"}
\label{fig:param_overview}
\end{figure}
Combining the definitions introduced in the previous section, we can now study solutions $\xi$ of our system \eqref{eq:system}. A graphical overview over the results of this section can be found in Fig. \ref{fig:param_overview}.
The following relations hold
\begin{lemma}[{\citep[Corollary 4.1]{cours}}]
Let $M = D^{1\times p}/(D^{1 \times q}R)$ be the left $D$-module finitely presented by the full row-rank matrix $R \in D^{q\times p}$. Let ${\cal F}$ be the function space over which we want to find solutions. Matrices always act on the right. 
    \begin{enumerate}
        \item If $M$ is torsion-free, there exists a non-trivial parametrization matrix $Q \in D^{p\times n}$ such that $RQ = 0$. Functions ${\cal F}^p \ni \vec{\xi} = Q\vec{f}$ with $\vec{f} \in {\cal F}^n$ are solutions of Eq.~\ref{eq:system}. Formally, we write $im_D(Q) = ker_D(R)$ and $im_{\cal F}(Q) \subseteq ker_{\cal F}(R)$.
        \item If $M$ is torsion-free and ${\cal F}$ is injective, \emph{all} solutions in ${\cal F}$ can be parametrized using $Q$, i.e. $im_{\cal F}(Q) = ker_{\cal F}(R)$.
        \item If $M$ is projective, \emph{all} solutions in ${\cal F}$ can be parametrized using $Q$, i.e. $im_{\cal F}(Q) = ker_{\cal F}(R)$, even if ${\cal F}$ is not injective.  
        \item If $M$ is stably free, then $M$ is in particular also projective (Definition~\ref{def:torsionundco}), so (3) applies as well.
    \end{enumerate}
    \label{lem:parametrizability}
\end{lemma}

\begin{definition}[Parametrizable systems]
If $im_{\cal F}(Q) = ker_{\cal F}(R)$, we call the system presented by $R$ parametrizable over the function space ${\cal F}$. 
\end{definition}
What happens if a system is not parametrizable over a certain function space? If a system over $D$ is torsion-free, but not projective, and ${\cal F}$ is not injective (as is the case for the function spaces considered here), some solutions may not be realizable, though a large portion of solution space is. If a system on the other hand has torsion ($t(M)\subset M$), we can in general only parametrize some solutions (for example for specific initial or boundary conditions). 
As a special case, if a system \emph{is} torsion ($t(M) = M$), the parametrization matrix $Q$ is trivial ($Q = 0$), and we can only parametrize the trivial, zero-in-every-task solution. 

\textit{Remark:} For linear systems governed by ordinary differential equations, the term parametrizability coincides with the notion of controllability in the sense of \cite{willems_controllability, fliess_controllability} in linear systems- and control theory. In contrast, for systems governed by partial differential equations, parametrizability is a stronger property, i.e. parametrizability always implies controllability \cite{zerz, pommaret1999algebraic} but not vice versa. Oftentimes in machine learning literature, parametrizability and controllability are used interchangeably \cite{lange-hegermann_algorithmic_2018, besginow_constraining_2022, harkonen_gaussian_2023}.

\begin{proposition}[{\cite[Definition 4.3, Theorem 4.2]{cours},\cite[Theorem 8]{chyzak}}]
If a system presented by the full row-rank matrix $R \in D^{q\times p}$ is torsion-free, there exist \textit{minimal parametrizations} $Q \in D^{p \times n}$, in the sense that the number of columns is minimal. This is the case, if $n = \mathrm{rank}(M) = p-q.$
\label{prop:minimalparam}
\end{proposition}

\section{Ghost tasking}
\subsection{Bridging PIGP and NSB}
\label{subsec:mlhraissi}
With the notion of parametrizability established in the previous sections, we are now in a position to characterize the parametrized Gaussian Processes defined in \ref{subsec:paramGP}. Lange-Hegermann \cite{lange-hegermann_algorithmic_2018} already observed that their method applies only to parametrizable systems, i.e., as established above, to projective systems, or at least torsion-free systems when working in an injective function space.

The method of Raissi et al. \cite{raissi_machine_2017} for parametrizing a Gaussian Process (Eq.~\ref{eq:PIGP}), by contrast, was not originally introduced within an algebraic context. We show here that it can equally be interpreted as a parametrization, obtained by applying the matrix
\[Q_{PIGP}(\vv x) = \begin{pmatrix} \mathbb{I}_{q} \\ {\cal L}^\phi(\vv x)\end{pmatrix}\]
to the base kernel $k_{uu}(\vv x, \vv x')$.  

Writing the system of Eq.~\ref{eq:PIGP} in homogeneous form,
\begin{equation}
    R^{\phi}_{PIGP}(\vv x)\vec{\xi}_{PIGP}(\vv x) = 0, \quad 
    R^{\phi}_{PIGP}(\vv x) = \begin{pmatrix} {\cal L}^{\phi}(\vv x) & -\mathbb{I}_{q} \end{pmatrix}, \quad
    \vec\xi_{PIGP}(\vv x) = \begin{pmatrix} \vec{u}(\vv x) \\ \vec{f}(\vv x) \end{pmatrix},
    \label{eq:mlhraissi}
\end{equation}
one sees that $Q^{\phi}_{PIGP}(\vv x)$ as defined above indeed satisfies
\begin{equation}R^{\phi}_{PIGP}(\vv x)Q^{\phi}_{PIGP}(\vv x) = \begin{pmatrix} {\cal L^{\phi}}(\vv x) & -\mathbb{I}_{q} \end{pmatrix} \begin{pmatrix} \mathbb{I}_{q} \\ {\cal L}^{\phi}(\vv x) \end{pmatrix} = 0.\end{equation}

Moreover, we can show that $R_{PIGP}(\vv x)$ as defined in Eq.~\ref{eq:mlhraissi} has a right inverse, choosing
\[S^{\phi}_{PIGP}(\vv x) = \begin{pmatrix} S_1 \\ {\cal L}^{\phi}(\vv x) S_1 - \mathbb{I}_q \end{pmatrix}\]
for any $S_1 \in D^{q \times q}$, since
\[R^{\phi}_{PIGP}(\vv x)S^{\phi}_{PIGP}(\vv x) = \begin{pmatrix}{\cal L}^{\phi}(\vv x) & -\mathbb{I}_q \end{pmatrix} \begin{pmatrix} S_1 \\ {\cal L}^{\phi}(\vv x) S_1 - \mathbb{I}_q \end{pmatrix} = \mathbb{I}_q. \]
By point (3) of Lemma \ref{lem:exteigenschaften}, the system presented by $R^{\phi}_{PIGP}(\vv x)$ is projective and hence parametrizable (Lemma \ref{lem:parametrizability}) over any (sufficiently smooth) function space. 

Comparing $R^{\phi}_{PIGP}(\vv x)$ as defined in Eq.~\ref{eq:mlhraissi} with the original system matrix $R^{\phi}(\vv x)$ in Eq.~\ref{eq:system}, we note that the latter is not necessarily projective, whereas $R^{\phi}_{PIGP}(\vv x)$ is projective by construction. The distinction between the two formulations lies entirely in how the inhomogeneities $f(\vv x)$ are incorporated. In the traditional physics-informed GP formulation $R^{\phi}_{PIGP}(\vv x)$, one inhomogeneity term is introduced per unknown function, even when $f(\vv x) = 0$, so that $R^{\phi}_{PIGP}(\vv x)$ must always have $2q$ tasks and thus also $2q$ columns. In the formulation of Eq.~\ref{eq:system} on the other hand, correlated inhomogeneities can be modeled using fewer tasks (so that $R^{\phi}(\vv x)$ has $p<2q$ columns), or, in the homogeneous case, omitted entirely, resulting in a system matrix $R^{\phi}(\vv x)$ that has exactly $p = q$ columns. In this latter case, with zero inhomogeneities, the module presented by $R^{\phi}(\vv x)$ is torsion (hence is neither torsion-free nor projective, see Def. \ref{def:torsionundco}), while the inclusion of $q$ inhomogeneities resulting in $R^{\phi}_{PIGP}(\vv x)$ always guarantees projectivity. This raises a question of direct practical consequence: are $q$ input tasks truly necessary, or can torsion-freeness or projectivity, and with it parametrizability, be achieved with fewer tasks? The answer is not merely of theoretical interest: it determines how much data and latent functions are actually required to obtain a well-posed, parametrized model. We show in the following that there exists a sweet spot reachable by our "ghost tasking" method, guaranteeing parametrizability while keeping tasks and latent functions at a minimum.

\subsection{Foundations of ghost tasking}
\begin{figure}[h!]
\centering
  \includegraphics[width=\textwidth]{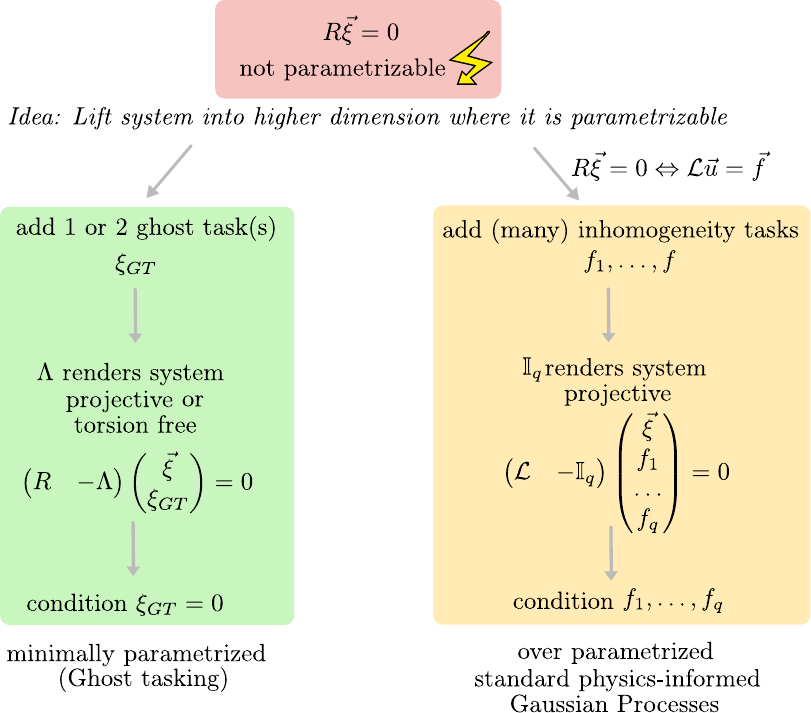}
\caption{Overview over the idea behind ghost tasking and its connection to physics-informed Gaussian Processes (PIGP) as in \cite{raissi_machine_2017}. A non parametrizable system can be rendered parametrizable by adding extra tasks, and subsequently conditioning them away again. The differences between ghost tasking and PIGPs lie in (a) how many tasks we add and (b) how they are coupled to the system ($\Lambda$ versus the identity matrix $\mathbb{I}_q$).}
\label{fig:ghosttasking_overview}
\end{figure}
\label{sec:theorem}
Ghost tasking locates the optimum between the PIGP method \cite{raissi_machine_2017} and the NSB method \cite{lange-hegermann_algorithmic_2018}, because it provides a method for rendering a physical system encoded in a Gaussian Process minimally parametrizable.
We propose a three-step procedure:
\begin{enumerate}
\item[(i)] introduce a maximum of 2 auxiliary tasks (i.e., auxiliary output variables) we call "ghost tasks" to embed the multi-dimensional GP in a higher-dimensional space, in which an otherwise unparametrizable system is reduced to a parametrizable system;
\item[(ii)] apply computer algebra methods to solve the resulting parametrizable system;
\item[(iii)] transform back to the original problem dimension by conditioning away the ghost tasks, i.e., conditioning them pointwise to vanish.
\end{enumerate}
We visualize this idea and show the analogy to the previously discussed connection to PIGP in Fig. \ref{fig:ghosttasking_overview}. For better readability, we omit the parameters $\phi$ and the argument $(\vv x)$ at times. To formalize this procedure and prove its applicability, we use a theorem developed for Serre reduction:

 \begin{theorem}[{Adapted from \cite[Thm.~17]{serre}, \cite[Thm.~3.6]{ogofserre}}]
\label{th:serre}
Let $D$ be a noetherian domain, $R \in D^{q\times p}$ a full row-rank matrix, $0\leq m \leq q$, $\Lambda \in D^{q\times m}$, $P = \begin{pmatrix}R & -\Lambda \end{pmatrix}$ and $M = D^{1\times p}/(D^{1\times q}R$) (resp., $E = D^{1\times(p+m)}/(D^{1\times q}P)$)) the left D-module finitely presented by $R$ (resp. $P$). Then, the following results are equivalent:
\begin{enumerate}
    \item The left D-module $E$ is stably free of rank $p-q+m$.
    \item $\{\tau(\Lambda_{\cdot i})\}_{i = 1,...,m}$ generates the right D-module ext$^1_D(M,D) = D^q/(RD^p)$, where the right D-homomorphism $\tau: D^q\rightarrow D^q/(RD^p)$ is the canonical projection and $\Lambda_{\cdot i}$ is the $i^{th}$ column of the matrix $\Lambda$. 
    \end{enumerate}

Moreover, the construction depends only on the residue class of $\Lambda$ i.e. $\overline{\Lambda} = \Lambda + RX$ yields the same result for any $X \in D^{p\times r}$.
\end{theorem}

\textit{Proof sketch:} For the full rigorous proof, we refer to \cite{ogofserre}. From Lemma \ref{lem:exteigenschaften} we recall that $E$ is stably free iff $\mathrm{ext}^1_D(E,D) = 0$. Using Eq.~\ref{eq:ext} and the third isomorphism theorem \footnote{For $A$ and $B$ submodules of $C$ such that $A\subseteq B \subseteq C$, then $(C/A)/(B/A)\cong C/B$) \cite{noether}} applied to $RD^p \subseteq PD^{(p+m)} \subseteq D^q$, we rewrite
\[\mathrm{ext}^1_D(E,D) \;\cong\; D^q/PD^{(p+m)} \;\cong\; \big(D^q/RD^p\big)\big/\big(PD^{(p+m)}/RD^p\big) \;\cong\; \mathrm{ext}^1_D(M,D)\big/\big(PD^{(p+m)}/RD^p\big).\]
Therefore,
\[\mathrm{ext}^1_D(E,D) = 0 \quad\Longleftrightarrow\quad \mathrm{ext}^1_D(M,D) = PD^{(p+m)}/RD^p = (RD^p + \Lambda D^{m})/RD^p.\]
Applying the canonical projection $\tau: D^q \to D^q/RD^p = \mathrm{ext}^1_D(M,D)$, the right-hand side is simply the image of $\Lambda D^{m}$ under $\tau$, i.e.
\[(RD^p + \Lambda D^{m})/RD^p = \tau(\Lambda D^{m}) = \tau(\Lambda)D^{m},\]
so that
\[\mathrm{ext}^1_D(E,D) = 0 \quad\Longleftrightarrow\quad \mathrm{ext}^1_D(M,D) = \tau(\Lambda)D^{m}.\]
That is, $E$ is stably free exactly when the projections $\tau(\Lambda_{\cdot i})$ of the columns of $\Lambda$ generate $\mathrm{ext}^1_D(M,D)$ as a right $D$-module.
\qed

Based on Theorem \ref{th:serre}, we can now formalize ghost tasking for parametrized Gaussian Processes:

\begin{theorem}[Ghost tasking]
Let $D$ denote a ring of linear differential operators whose coefficients are polynomial, rational, or general meromorphic functions. Let furthermore $R^{\phi}(\vv x) \in D^{q \times p}$ be a full row-rank matrix representing a linear system of ordinary or partial differential equations with coefficients in $D$, and $M$ be the left module presented by $R$. Assume the system is not parametrizable. 

Then, there exists a matrix $\Lambda(\vv x) \in D^{q \times m}$ with at most $m = 2$ columns
such that the embedded system $P^{\phi}(\vv x) = \begin{pmatrix}R^{\phi}(\vv x) & -\Lambda(\vv x) \end{pmatrix}$ is stably free, and hence parametrizable over any sufficiently smooth function space. The canonical projection of the columns of $\Lambda$ generates ext$^1_D(M, D)\cong D^q/(RD^p)$.  
Consequently, there exists a proper parametrization matrix $Q^{\phi}(\vv x)$, such that every realization of the parametrized GP
\[\eta \sim {\cal GP}(0, Q^{\phi}(\vv x)k_0(\vv x,\vv x')Q^{\phi T}(\vv x'))\]
solves the embedded system $P^{\phi}(\vv x)$. If we constrain the additional ghost task(s) to zero everywhere, the constrained solutions of the embedded system coincide with the solutions of the original system.
\label{th:GT}
\end{theorem}
   
\begin{proof}
Assume $R$ is not parametrizable, hence not projective. By Theorem~\ref{th:serre}, constructing a stably free module $E = D^{1\times(p+m)}/(D^{1\times q}P)$, $P = \begin{pmatrix} R & -\Lambda \end{pmatrix}$, amounts to choosing $\Lambda \in D^{q\times m}$ whose projected columns generate $\mathrm{ext}^1_D(M,D)$; the minimal such $m$ is thus bounded by the minimal number of generators of $\mathrm{ext}^1_D(M,D)$. By Lemma~\ref{lem:torsion}, $\mathrm{ext}^1_D(M,D)$ is a finitely generated torsion module, and by (1) of Theorem~\ref{th:staffordserre} every such module is generated by at most two elements. Hence there exists $\Lambda \in D^{q\times m}$ with $m \le 2$ such that $E = D^{1\times(p+m)}/(D^{1\times q}(R,-\Lambda))$ is stably free, and therefore also projective.

For $D$ a ring of ordinary differential operators, (2) of Theorem~\ref{th:staffordserre} strengthens this: every torsion module over $D$ is cyclic, so $\mathrm{ext}^1_D(M,D)$ is generated by a single element and $\Lambda$ may be chosen with $m = 1$.

Since $E$ is projective, $P\vec{\eta} = 0$ is parametrizable over any sufficiently smooth function space (Lemma~\ref{lem:parametrizability}), yielding a parametrized Gaussian Process with at most two ghost tasks $\xi_{GT}$. Writing
\begin{align}
\text{Original: } & R \vec{\xi} = 0, \\
\text{New: } & P \vec{\eta} = \begin{pmatrix} R & -\Lambda \end{pmatrix}\begin{pmatrix} \vec{\xi}\\\xi_{GT} \end{pmatrix} = R \vec{\xi} - \Lambda \xi_{GT} = 0,
\end{align}
with $\Lambda \neq 0$ (else $R$ is already parametrizable), setting $\xi_{GT} = 0$ recovers the original system, so the two are pointwise equivalent.
\end{proof}

\subsection{Properties of ghost tasking}
\label{subsec:ghostprops}
Having established the existence of a minimal ghost-tasking parametrization in Theorem~\ref{th:GT}, we now discuss several refinements, special cases, and practical considerations that follow from it.

\paragraph{Pointwise conditioning} In practice, a Gaussian Process cannot be constrained to vanish identically on a continuum, only at finitely many points. We therefore replace the pointwise condition $\xi_{GT}(\vv x) = 0$ for all $\vv x$ with a finite set of conditioning points $\{\vv x_1, \dots, \vv x_k\}$. Empirically, we find that a small to moderate number of conditioning points suffices in practice (Chapter~\ref{sec:experiments}). 

\paragraph{Torsion-freeness} Theorem~\ref{th:GT} guarantees the stronger characteristic of stable freeness, which implies torsion-freeness. For efficiency reasons (Chapter~\ref{sec:experiments} and \ref{sec:discussion}), it may suffice and be actually desirable to require only that the embedded module $E$ be torsion-free, rather than stably free. Theorem~\ref{th:GT} admits the following weaker, and thus more permissive, variant:

\begin{proposition}
\label{prop:torsionfreeghost}
Let $D, R, M, P,$ and $E$ be defined as in Theorem \ref{th:GT}. There exists a matrix $\Lambda \in D^{q\times m}$ with at most $m=2$ columns such that the embedded system  is torsion-free. Writing $\tau(\Lambda)D^{m}$ for the submodule generated by the canonical projections of the columns of $\Lambda$, this holds precisely when
\[
\mathrm{ext}^1_D\big(\mathrm{ext}^1_D(M,D)/(\tau(\Lambda)D^{m}),\ D\big) = 0.
\]
\label{prop:GT2}
\end{proposition}
This formulation subsumes the stably free case of Theorem~\ref{th:GT}: when $\tau(\Lambda)D^{m} \cong \mathrm{ext}^1_D(M,D)$ (i.e., $\Lambda$'s columns generate all of $\mathrm{ext}^1_D(M,D)$), the quotient $\mathrm{ext}^1_D(M,D)/\tau(\Lambda)D^{m}$ vanishes, and the condition holds trivially since $\mathrm{ext}^1_D(0,D)=0$.
\begin{proof}
    We denote the Auslander transpose of $E$ by $N_E$. By Lemma \ref{lem:exteigenschaften}, $E$ is torsion-free, iff $\mathrm{ext}^1_D(N_E, D) = 0$, and since analogously to the proof of Theorem \ref{th:serre}, we can write 
    \begin{align}
        N_E = D^q/PD^{(p+m)} \cong \big(D^q/RD^p\big)\big/\big(PD^{(p+m)}/RD^p\big) 
        \cong \mathrm{ext}^1_D(M,D)\big/\big(PD^{(p+m)}/RD^p\big) \nonumber \\
        \cong \mathrm{ext}^1_D(M,D)\big/\left((RD^p + \Lambda D^{m})/RD^p\right) = \mathrm{ext}^1_D(M,D)/\left(\tau(\Lambda)D^{m}\right), \nonumber
        \end{align}
    Therefore, iff $\mathrm{ext}^1_D\big(\mathrm{ext}^1_D(M,D)/(\tau(\Lambda)D^{m}),\ D\big) = 0$, $E$ is torsion-free. The number of columns is bounded by Theorem~\ref{th:GT}, and the rest of the proof follows the same arguments as Theorem~\ref{th:GT}. 
\end{proof}

\paragraph{Stronger conjecture} A long-standing conjecture of \cite{staffordog}[Conjecture 3.8] asks, if any finitely generated torsion right module is cyclic. If this is true, the result would extend the single-column bound of to all Weyl algebras $A_n(k)$, $n \ge 1$, i.e. also to PDEs with polynomial coefficients. An AI-assisted machine proof in Lean \cite{lean} has recently been published \cite{chris_stafford}, though has not been peer reviewed yet. 

\paragraph{Recovering PIGP} Choosing $\Lambda = \mathbb{I}_q$ recovers the PIGP construction of \cite{raissi_machine_2017}: since $\tau(\mathbb{I}_q)$ trivially generates $\mathrm{ext}^1_D(M,D)$, this is always a valid — if not minimal — choice of $\Lambda$, corresponding to $m = q$ ghost tasks (i.e. one ghost task per equation).

\paragraph{Finding $\Lambda$ in practice} Theorem~\ref{th:GT} guarantees the existence of $\Lambda$ but calculating a minimal set of generators is generally difficult. In the holonomic case, $\Lambda$ can be computed directly using \texttt{Macaulay2}~\cite{macaulay2}. Otherwise, we find the brute force method of systematically proposing $\Lambda$ candidates containing low-order terms in $D$ and verifying the validity of $\Lambda$ computationally to be surprisingly effective. Verification is comparatively cheap in both the stably free/projective case (amounting to checking if extension modules are zero, Lemma~\ref{lem:exteigenschaften}), using a computer algebra system such as \texttt{OreModules}~\cite{oremodules} or \texttt{Macaulay2}~\cite{macaulay2}. We demonstrate the appropriate syntax in \ref{ap:macaulay}.

\paragraph{Bound on the number of latent functions} Bounding the number of ghost tasks also bounds the number of latent functions needed to parametrize the solution space. By Proposition \ref{prop:minimalparam} the minimal parametrization matrix $Q$ of the embedded system $P \in D^{q \times (p+m)}$ has $n = (p+m) - q$ columns. 
Since $q \le p$ (full row-rank) and $0 \le m \le q$ (with $m\le 2$ the minimal choice and $m=q$ recovering PIGP), $n$ ranges between $p-q$ and $p-q+q = p$ latent functions. In the most extreme case of homogeneous differential equations ($R = {\cal L}$), $m \in \{1,2\}$ columns (corresponding to one or two ghost tasks) suffice. 

\section{Experiments}
\label{sec:experiments}
We illustrate ghost tasking in three experiments, discussing different aspects of the implementation in each of them: the first experiment is an ODE, where we are learning the dynamics of a pendulum system. We show the advantage of knowing the ghost task to be exactly zero: conditioning the ghost task to be zero on additional collocation points allows us to better identify the eigenfrequency of the system. The second experiment is a PDE with two variables. In this case, the previously discussed notions of torsion-freeness and projectivity do not coincide anymore, and we compare the performance between the two parametrization variants. Lastly, we show the performance on a slightly more real-world example of a PDE with 3 variables in magnetostatics. For each experiment, we first derive the augmented parametrization matrix $Q^\phi(\vv x)$, and then apply the resulting minimally parametrized Gaussian Process obtained by ghost tasking (in short referred to as "GT") to solve the forward and inverse problem. We implement the physics-informed Gaussian Process (PIGP) method of \cite{raissi_machine_2017} as a baseline comparison to our method in each experiment. For the second experiment, an explication of the commands used in the computer algebra programs \texttt{OreModules} or \texttt{Macaulay2} can be found in \ref{ap:macaulay}. 
 Scripts and data for all experiments can be found on our GitHub\footnote{\url{https://github.com/moserjo/GhostTask}}.\\

All numerical experiments use the python package \texttt{PCGP} \cite{mymaxent}, built on \texttt{GPyTorch} \cite{gpytorch1, gpytorch2} for optimization and \texttt{jax} \cite{jax} for grid evaluations. For each experiment we implement GT with the minimal number of ghost tasks, alongside PIGP as the comparison. In the forward problem, we only specify data of (i) the inhomogeneities, (ii) ghost task (or equivalent inhomogeneity tasks for PIGP), and (iii) boundary/initial conditions, and infer the solution on the unknown interior. Data points are noise-free, though we assume a Gaussian likelihood with small variance, $\sigma^2 = 10^{-6}$, as a nugget for numerical stability. We choose the squared-exponential kernel as introduced in Sec. \ref{subsec:paramGP} as our base kernel and the calibration parameters $\vartheta = (A, l)$ are optimized (ADAM, learning rate 0.1, 100 iterations) with a flat prior, $A$ constrained to be positive and $l$ constrained below by the minimal training-point distance. Accuracy is measured by RMSE on a test grid whose size depends on the experiment.\\

For the inverse problem, we systematically learn the maximum a posteriori (MAP) estimate of the physical parameters $\phi$ by optimization jointly with $\vartheta$
\begin{equation}
    (\phi_{MAP}, \vartheta_{MAP}) = \text{argmax } p(\phi, \vartheta \mid {\cal D}, \sigma),
\end{equation} for varying numbers of data points ${\cal D}$ and noise levels $\sigma$. Each setting is repeated 10 times with different random noises to account for stochasticity. We define a run as converged if the physical parameters $\phi$ have less than 0.01\% (relative to their mean) peak-to-peak variation within the last 50 iterations. Performance is measured by the MAP error relative to the true value of each physical parameter, $\frac{\phi_{i,MAP} - \phi_{i,true}}{\phi_{i,true}}$.\\

All computations were performed on a node of an x86\_64 Linux cluster (16 cores, 126 GB RAM) using 8 CPU cores per job.

\subsection{Experiment 1: Tripendulum}
\label{sec:Ex2}
The first experiment consists of a system of coupled ODEs. 

Inspired by \cite{besginow_constraining_2022} and \cite{mymaxent}, we consider three pendula of length $\ell$ that feel the gravitational acceleration $g$ and are connected with a rod that moves with the acceleration $f(t)$. 
The linearized system can be written in the form $R^{(\ell,g)}(\vv x)\vec{\xi}(\vv x) = \vec{0}$ as follows:
\begin{align}
    \begin{pmatrix}
        \partial_t^2+g/\ell&0&0&-1/\ell\\0&\partial_t^2+g/\ell&0&-1/\ell\\0&0&\partial_t^2+g/\ell&-1/\ell\\
    \end{pmatrix}\begin{pmatrix}
        \theta_1\\ \theta_2 \\\theta_3 \\ f
    \end{pmatrix} = \begin{pmatrix}
        0\\0\\0
    \end{pmatrix}.
\end{align}
If the lengths of the pendula are all the same, the system has torsion, and is therefore not parametrizable.
Since the system consists of ODEs, the notion of torsion-freeness and projectivity coincides. To render the system parametrizable, we will therefore need to introduce a column vector $\Lambda$, such that its projection generates $D^3/R^{(\ell, g)}D^4 \cong ext^1_D(M, D).$ In the ODE case, we can calculate this generator using \texttt{Macaulay2} (see \ref{ap:macaulay}). $\Lambda$ is not unique, in fact for any generator $\Lambda$, $\tilde{\Lambda} = \Lambda + R^{(\ell,g)}X$ for $X \in D^4$ are all equally valid choices.  We use $\Lambda = (-1, -\partial_t, -t)^T$, leading to the new system 
\begin{align}
    \begin{pmatrix}
        \partial_t^2+g/\ell&0&0&-1/\ell & 1\\
        0&\partial_t^2+g/\ell&0&-1/\ell & \partial_t\\
        0&0&\partial_t^2+g/\ell&-1/\ell & t\\
    \end{pmatrix}\begin{pmatrix}
        \theta_1\\ \theta_2 \\\theta_3 \\ f \\ \xi_{GT}
    \end{pmatrix} = \begin{pmatrix}
        0\\0\\0\\0
    \end{pmatrix}.
\label{eq:Ex2_2}
\end{align}
that is parametrizable with the parametrization matrix 
\begin{align} Q^{(\ell,g)}_{GT_1}(t) = \begin{pmatrix}
1 & -\partial_t^{2} \ell^{2}-g \ell \\
 1 & -\partial_t^{3} \ell^{2}-g \ell \partial_t \\
 1 & -\partial_t^{2} \ell^{2} t+2 \partial_t \,\ell^{2}-g l t \\
 \partial_t^{2} \ell+g & 0 \\
 0 & \partial_t^{4} \ell^{2}+2 \partial_t^{2} g \ell+g^{2} 
\end{pmatrix}.
\label{eq:Ex2_B}
\end{align}
The choice of using $t$ in $\Lambda$ means that we work in the non-commutative first Weyl algebra $D = A_1(\mathbb{R})$. Since for ODEs, torsion-freeness and projectivity coincide, the system is parametrizable over $C^{\infty}$. 

For the inverse problem setting we will use $Q_{GT_1}(t)$ to parametrize the GP kernel 
\[K(t, t') = Q^{(\ell,g)}_{GT_1}(t) k^{\vartheta}_0(t, t')Q^{(\ell,g) T}_{GT}(t').\]
For the forward problem, we would however typically use initial value conditions on $\theta_i(t = 0)$ and on $\partial_t\theta_i$ for $i = 1,2,3$ to uniquely identify the solution. To do so, we need to add additional tasks for $\partial_t\theta_i$. Since $\partial_t \theta_i$ correlates with $\theta_i$, we need to correctly couple the additional tasks to the parametrization of $\theta_i$. The parametrization of $\theta_i$ corresponds to the operators of the $i$-th row of $Q^{(\ell,g)}_{GT_1}(t)$ (Eq.~\ref{eq:Ex2_B}), so we parametrize $\partial_t \theta_i$ with $\partial_t$ applied to the operators of the $i$-th row of $Q^{(\ell,g)}_{GT_1}(t)$. Since we work in a Weyl algebra, we need to make sure to respect the commutation relation $[\partial_t, t] = 1$ (Eq.~\ref{eq:commutator}) when applying $\partial_t$. By adding the derivative tasks $\partial_t \theta_i$ to our parametrization matrix, we arrive at the parametrization matrix for the forward problem
\begin{align} Q^{(\ell,g)}_{GT_2}(t) = \begin{pmatrix}
1 & -\partial_t^{2} \ell^{2}-g \ell \\
 1 & -\partial_t^{3} \ell^{2}-g \ell \partial_t \\
 1 & -\partial_t^{2} \ell^{2} t+2 \partial_t \,\ell^{2}-g \ell t \\
 \partial_t^{2} \ell+g & 0 \\
 0 & \partial_t^{4} \ell^{2}+2 \partial_t^{2} g \ell+g^{2} \\
 \partial_t & -\partial_t^{3} \ell^{2}-g \ell\partial_t \\
 \partial_t & -\partial_t^{4} \ell^{2}-g \ell \partial_t^2 \\
 \partial_t & -\partial_t^{2} \ell^{2}-\partial_t^{3} \ell^{2} t+2 \partial_t^2 \,\ell^{2}-g \ell t \\
\end{pmatrix}.
\label{eq:Ex2_GT}
\end{align}

The baseline PIGP method interprets $f(t)$ as three separate inhomogeneities of the form 
\begin{align}
    {\cal L}^{(\ell,g)}(t)\vec{\theta}(t) &= \vec{f}(t)\Leftrightarrow
    \begin{pmatrix}
        \ell\partial_t^2+g&0&0\\0&\ell\partial_t^2+g&0\\0&0&\ell\partial_t^2+g
    \end{pmatrix}\begin{pmatrix}
        \theta_1\\ \theta_2 \\\theta_3
    \end{pmatrix} = \begin{pmatrix}
        f_1(t)\\f_2(t)\\f_3(t)
    \end{pmatrix}.
\end{align}

Though we theoretically know $f_i(t) = f(t)$ describes the movement of the same rod in all three tasks, i.e. the tasks are perfectly correlated, we still need to model them as independent (inter-task correlation = 0) tasks to account for the different movements of the pendula. 

This leads to the equivalent PIGP parametrization matrix (already including the derivative tasks $\partial_t\theta$)
\begin{align}
    Q^{(\ell,g)}_{PIGP}(t) = \begin{pmatrix}
        \mathbb{I}_q\\ {\cal L}^{(\ell, g)}(t) \\ \partial_t \mathbb{I}_q
    \end{pmatrix} = \begin{pmatrix}
    1&0&0\\0&1&0\\0&0&1\\
        \ell\partial_t^2+g&0&0\\0&\ell\partial_t^2+g&0\\0&0&\ell\partial_t^2+g\\
        \ell\partial_t&0&0\\
        0&\ell\partial_t&0\\
        0&0&\ell\partial_t\\
    \end{pmatrix}.
    \label{eq:Ex2_PIGP}
\end{align}
For the numerical experiments, we can directly use $Q^{(\ell,g)}_{GT}(t)$ resp. $Q^{(\ell,g)}_{PIGP}(t)$ as input for the package \texttt{PCGP} \cite{mymaxent} that automatically generates the corresponding kernels that can then be used either with \texttt{GPyTorch} \cite{gpytorch1, gpytorch2} or \texttt{jax} \cite{jax}. 

The ground truth solution of all three angles $\theta = \theta_i,$$i = 1,2,3$ given fixed initial conditions $\theta(0)=\theta_0,\ \dot\theta(0)=\dot\theta_0$ and a chosen forcing acceleration $f(t) = 5\sin(\omega_ft)$ is

\begin{equation}
\theta(t) = \theta_0\cos(\omega t)
+ \left(\dot\theta_0 - \frac{5 \omega_f}{\ell(\omega^2 - \omega_f^2)}\right)\frac{1}{\omega}\sin(\omega t)
+ \frac{5}{\ell(\omega^2 - \omega_f^2)}\sin(\omega_f t),
\qquad
\omega = \sqrt{\frac{g}{\ell}}
\label{eq:ex2_groundtruth}
\end{equation}
with chosen parameter values $g=9.81$, $\ell=2.5$, $\omega_f=1.6$.
The training data for the forward problem consists of a total of six initial condition points at $\theta_i(t =0)$ and $\partial_t\theta_i(t = 0), \; i = 1,2,3$ and 10 linearly spaced points for $f(t)$ in GT and the same 10 points in all $f_i(t),\; i = 1,2,3$ in the PIGP method. We compare two different datasets for the ghost task: because we know $\xi_{GT}(t) = 0$ everywhere without needing to measure it, we can use as many collocation points to enforce this pointwise as we want. We compare 10 to 20 linearly spaced collocation points and refer to the version with additional linearly spaced collocation points as "GT+". The RMSE was calculated over 201 linearly spaced test points in each task.
For the inverse problem we are comparing the MAP estimates of the parameters $\ell$ and $g$, as well as their ratio $g/\ell$ for different amounts of training data per task $n = 3, 5, 10$ and different noise levels $\sigma \in \{0.001, 0.01, 0.1, 0.2\}.$ While the ratio $g/\ell$ we compare is not an independent parameter, Eq.~\ref{eq:ex2_groundtruth} shows that $\omega = \sqrt{\frac{g}{\ell}}$ dominates the dynamics of the system and is therefore an important quantity to assess.

As in the forward problem, we compare a "standard" GT method where we condition the ghost task on $n$ collocation points, i.e. equally many points as the other tasks, and a "GT+" method, where we condition the ghost task on 25 collocation points to better constrain the ghost task. 

Since we have two parameters of interest here, it is also interesting to look at the joint posterior probability of $\ell$ and $g$ (Eq.~\ref{eq:posterior}), that is highly non-trivial. We calculate it on a $20\times 20\times 20 \times 20$ linearly spaced point grid over the parameters $\ell, g, A, l$ and subsequently marginalize over the calibration parameters amplitude $A$ and lengthscale $l$. We assume a flat prior over the considered ranges $g \in [1, 15],\, \ell \in [0.5, 5],\, A \in [0.5, 100],\, l \in [l_{min}, l_{max}]$ with $l_{min}, l_{max}$ being the closest and farthest distance between training points.
\begin{equation}
    p(\ell, g|{\cal D}) = \int p(\ell, g, A, l|{\cal D}) dA \,dl \approx \frac{1}{Z}\sum_{i = 1}^{20}\sum_{j = 1}^{20}p(\ell, g, A_i, l_j)
    \label{eq:posterior}
\end{equation} 
We will qualitatively assess the performance of GT, GT+, and PIGP compared to the optimal posterior resulting from a full Bayesian parameter estimation using a Gaussian likelihood centered around the true forward solution (Eq.~\ref{eq:ex2_groundtruth}) as a function of $g$ and $\ell$.

\begin{figure}[h!]
\centering
  \includegraphics[width=\textwidth]{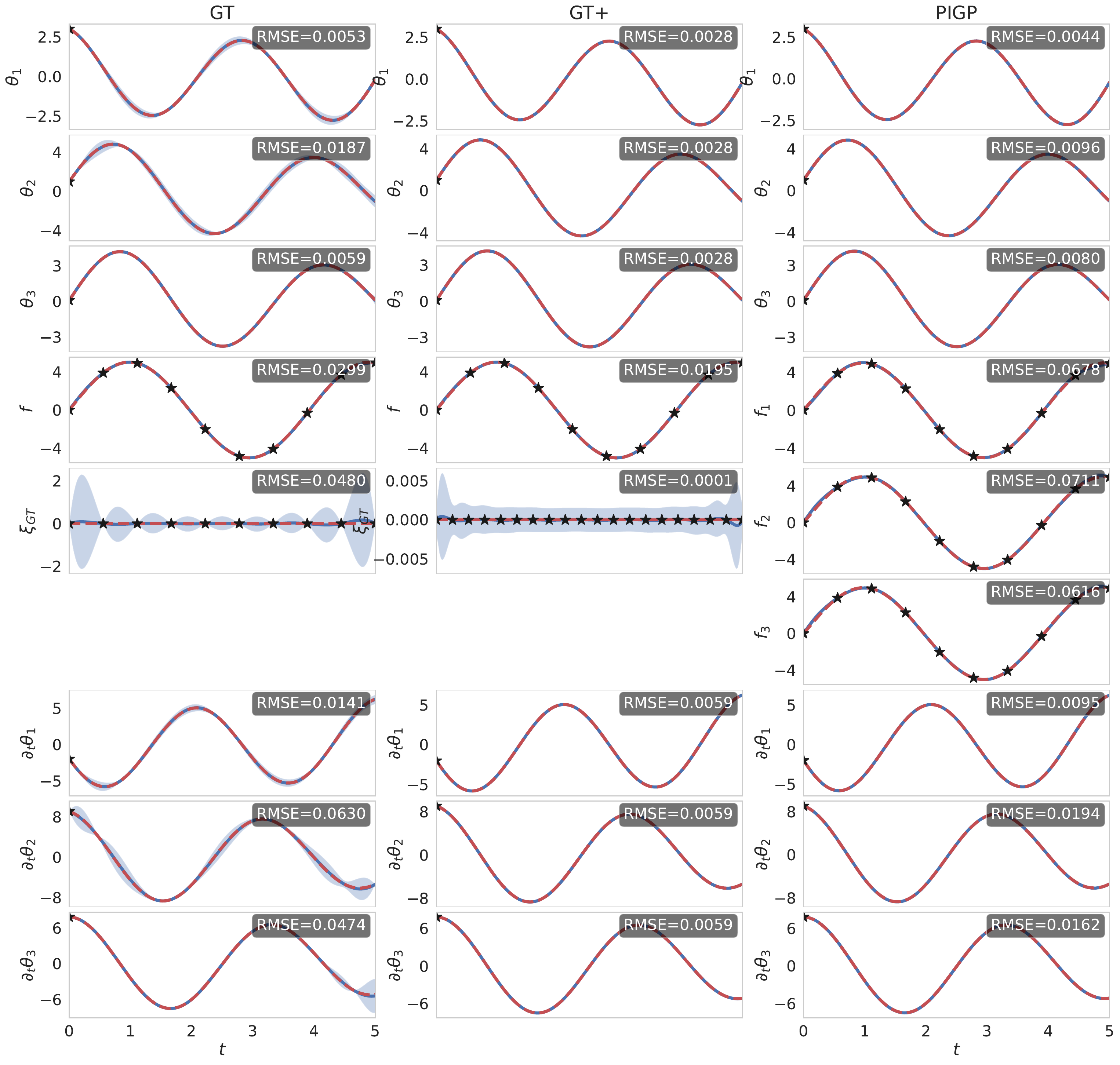} 
\caption{Inferred solutions of the forward problem of experiment 1 for the GT method (left), GT+ with extra collocation points (middle), and the PIGP method (right). Each subplot shows the solution of one task, the last three are the derivative of the angles in the first three tasks, necessary to constrain a second order initial condition problem. The true solution is marked in red, the training data with black stars, and the inferred solutions in blue with a lightblue $2\sigma$ uncertainty band. The RMSE on the test points is written in the corner of each task. }
\label{fig:Ex2_forward}
\end{figure}
In Fig. \ref{fig:Ex2_forward} we show the forward solution of each task, on the left side for GT, in the middle for GT+, and on the right side for PIGP. 
The RMSE of the mean solution calculated over the test points is written in the corner of each task and ranges between 0.0001 and 0.1. We see that overall, the RMSE of the GT method using a smaller amount of total datapoints than GT+ and PIGP are similar to PIGP, though the uncertainty is high in some places. The additional collocation points of GT+ reduce this uncertainty. In total, GT+ uses the same amount of training data as PIGP, but the RMSE for GT+ is lower than for PIGP. 

\begin{figure}[h!]
\centering
  \includegraphics[width=\textwidth]{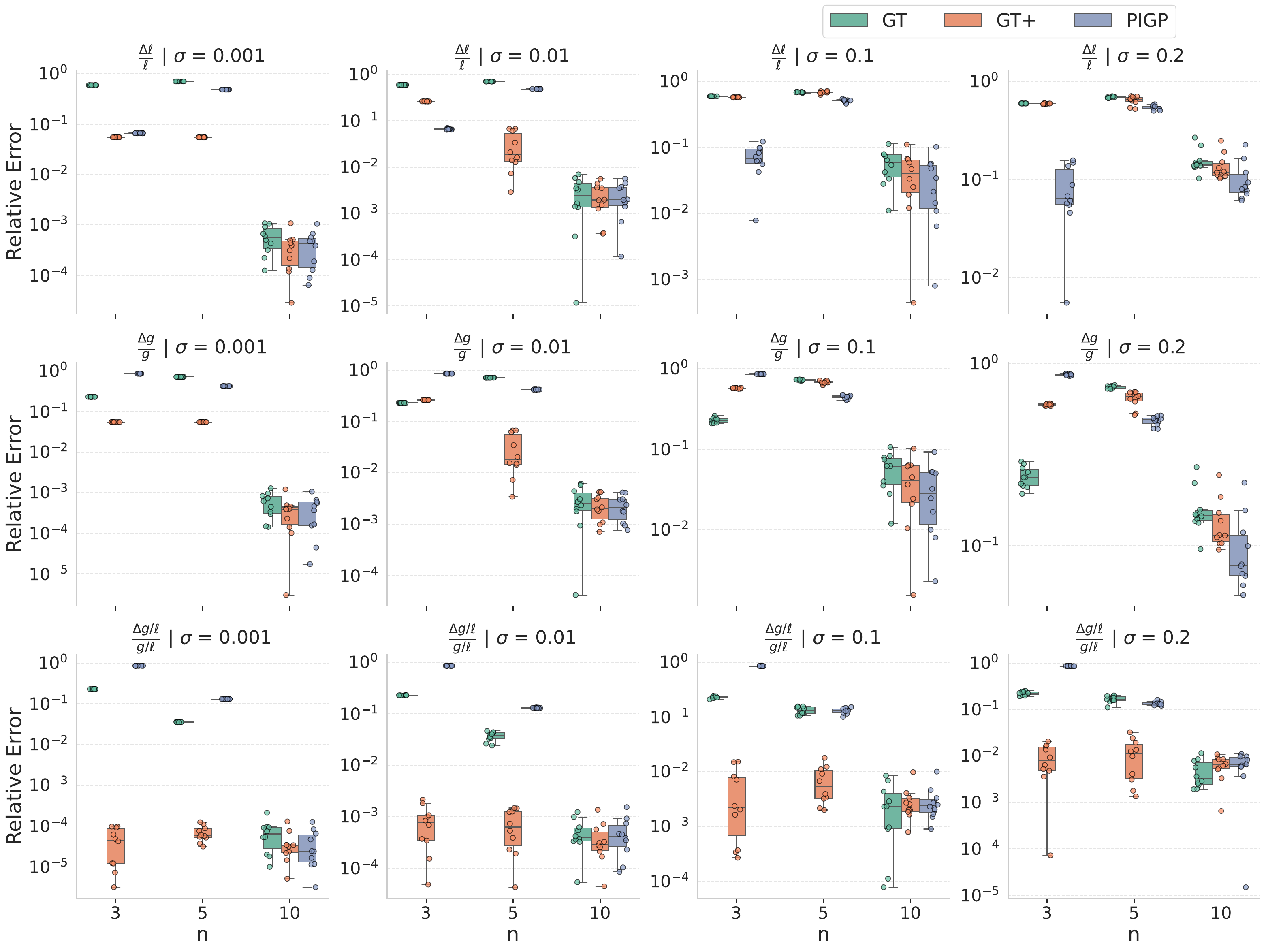}
\caption{Results of the systematic comparison in experiment 1 of GT (green), GT+ (orange), and PIGP (blue) using the mean absolute error of the learned physical parameters (maximum a posteriori, MAP value) $\ell$ (first row), $g$ (middle row), as well as their ratio $g/\ell$ (last row) relative to their true values $(\ell, g) = (2.5, 9.81)$ as metric. Note that $g/\ell$ is not an independent parameter, but does dictate much of the dynamics of the system and can as such be most precisely estimated by GT+. The experiments have been conducted for different numbers of training data $n \in \{3, 5, 10\}$ for the tasks, as well as different noise levels $\sigma \in \{0.001, 0.01, 0.1, 0.2\}$, shown in different subplots. We show the distinct runs as scatterplots, as well as their box-plots for each combination $(\sigma$, n).}
\label{fig:Ex2_inverse}
\end{figure}

We observe in Fig. \ref{fig:Ex2_inverse} that GT achieves a comparable, and in $g/\ell$ about one order of magnitude improved performance compared to PIGP in the low data regime. However, GT+ improves especially the $g/\ell$ results up to 5 orders of magnitude. 
We believe the performance of ghost tasking may be attributed to the inter-coupling between the tasks corresponding to the angles $\theta_i(t)$, which allows the system to share information on the joint dynamics \cite{mymaxent}. Constraining on additional collocation points in the ghost task forces the fitting error of the ghost task, that can be seen as a residue, closer to zero, progressively suppressing oscillatory components in the ghost task. Therefore, no dynamically important oscillations can be absorbed in the ghost task as a residue, considerably improving the learning of the frequency-related ratio. The additional collocation points cannot be provided to PIGP in this experiment, since the forcing term $f(t)$ is treated as an unknown quantity that must itself be inferred from measured data. PIGP recovers the performance of GT in the case where $f(t)$ is assumed to be fully known, as we show in the ablation studies in Tab. \ref{tab:Ex2}.

The full posterior distributions depicted in Fig. \ref{fig:Ex2_posteriors} corroborate this analysis. 
The figure displays the transformed log posterior landscapes\footnote{The reference standard Bayesian estimate is very sharp since we do not have any model uncertainty, so the underlying multimodal structure can only be resolved when plotting the logarithm of the log posterior. We therefore depict $-\log(max(\log(p))-\log(p)+1)$ for all three methods.} in the $(g,\ell)$-plane for all combinations of $n$ and $\sigma$, comparing the reference full Bayesian estimate to GT, GT+ and PIGP. The elongated structures in the posterior indicate a strong correlation between $g$ and $\ell$ that is explained by the system dynamics depending primarily on $\omega=\sqrt{g/\ell}$.
Furthermore, the reference standard Bayesian posterior exhibits multimodality, with several distinct frequency modes corresponding to solutions with different frequencies passing through the same training points. Although the true frequency mode has substantially higher posterior probability, the presence of additional modes reflects the inherent ambiguity of parameter inference from oscillatory observations. Importantly, the same multimodal structure is recovered by GT+, while multimodality is present in GT, but with a slightly shifted frequency, pointing to physical oscillations being absorbed in the ghost task as previously discussed. In general, this indicates that the ghost tasking method captures this aspect of the underlying posterior geometry. 
In contrast, the PIGP posterior is unimodal and is not able to reproduce these additional physically meaningful modes in the considered setting, which can lead to biased posteriors as can be seen for $n = 3$. 
When the forcing term $f(t)$ is assumed to be known (almost) exactly, PIGP is also able to recover the multimodal posterior structure, as shown in our ablation studies in Fig. \ref{fig:Ex2_posteriors_knownu} in \ref{ap:ablation}.

\begin{figure}[h]
\centering
\begin{minipage}{0.45\textwidth}
  \includegraphics[width=\textwidth]{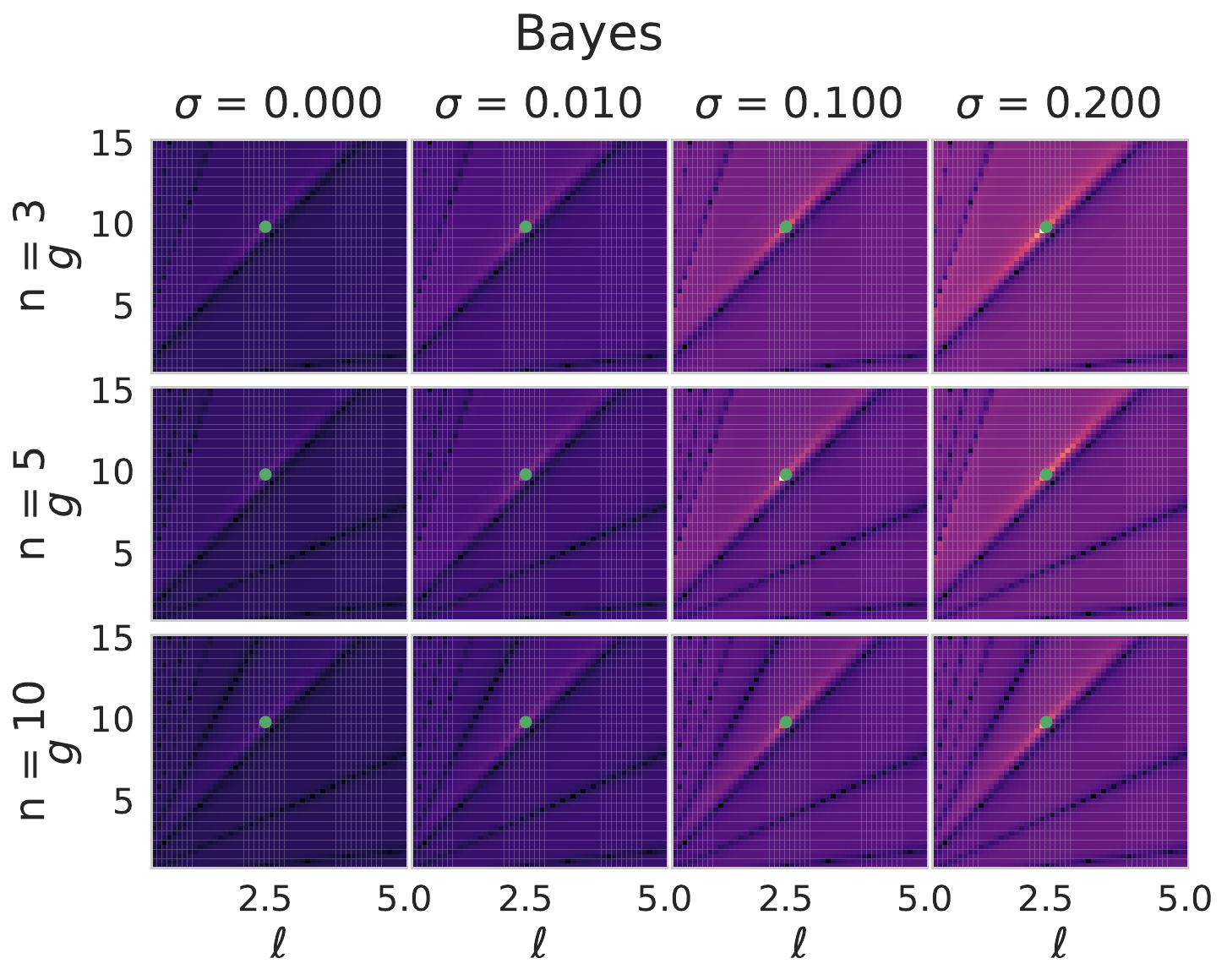}
  \end{minipage}
  \begin{minipage}{0.45\textwidth}
  \includegraphics[width=\textwidth]{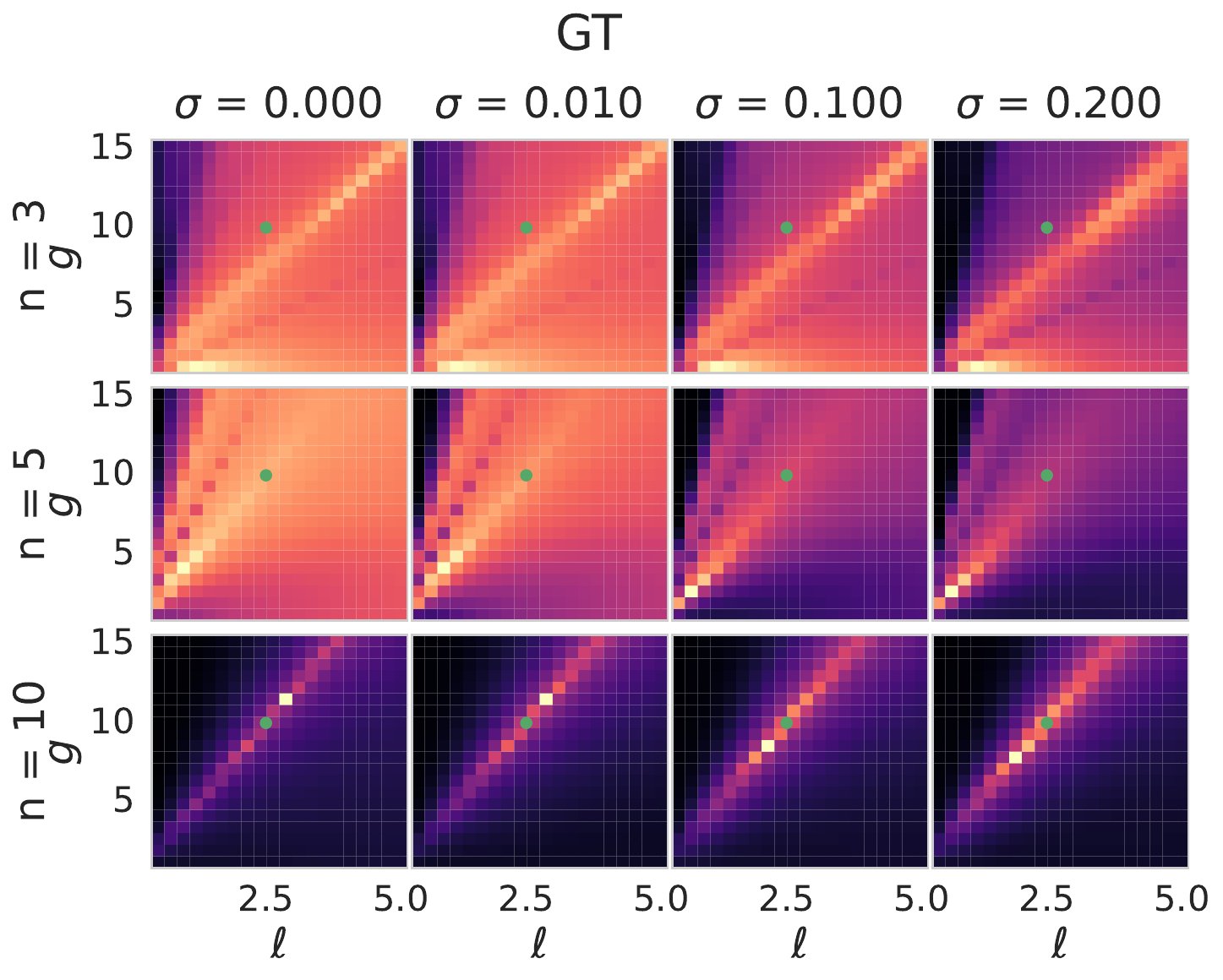}
  \end{minipage}
  \begin{minipage}{0.45\textwidth}
  \includegraphics[width=\textwidth]{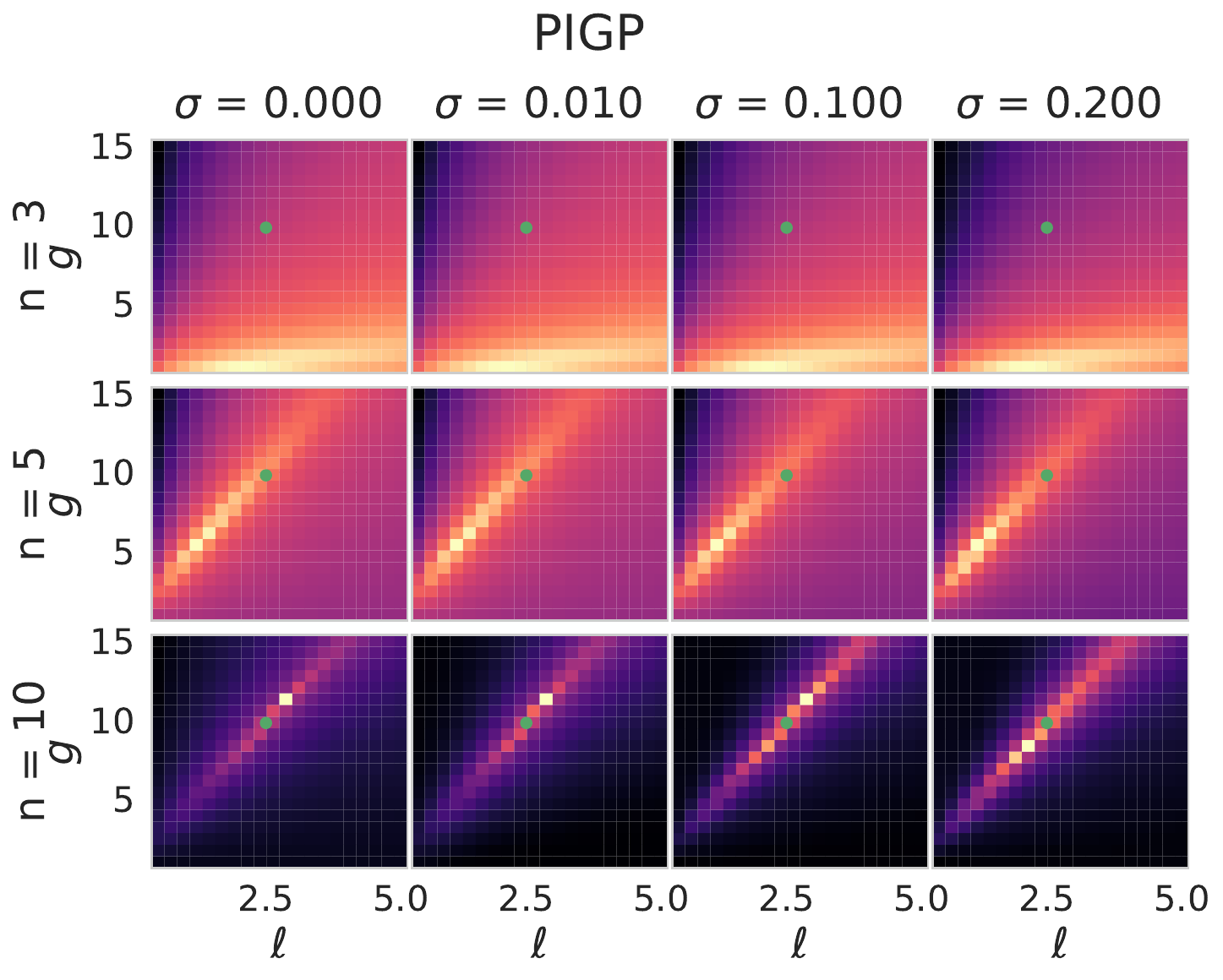}
  \end{minipage}
  \begin{minipage}{0.45\textwidth}
  \includegraphics[width=\textwidth]{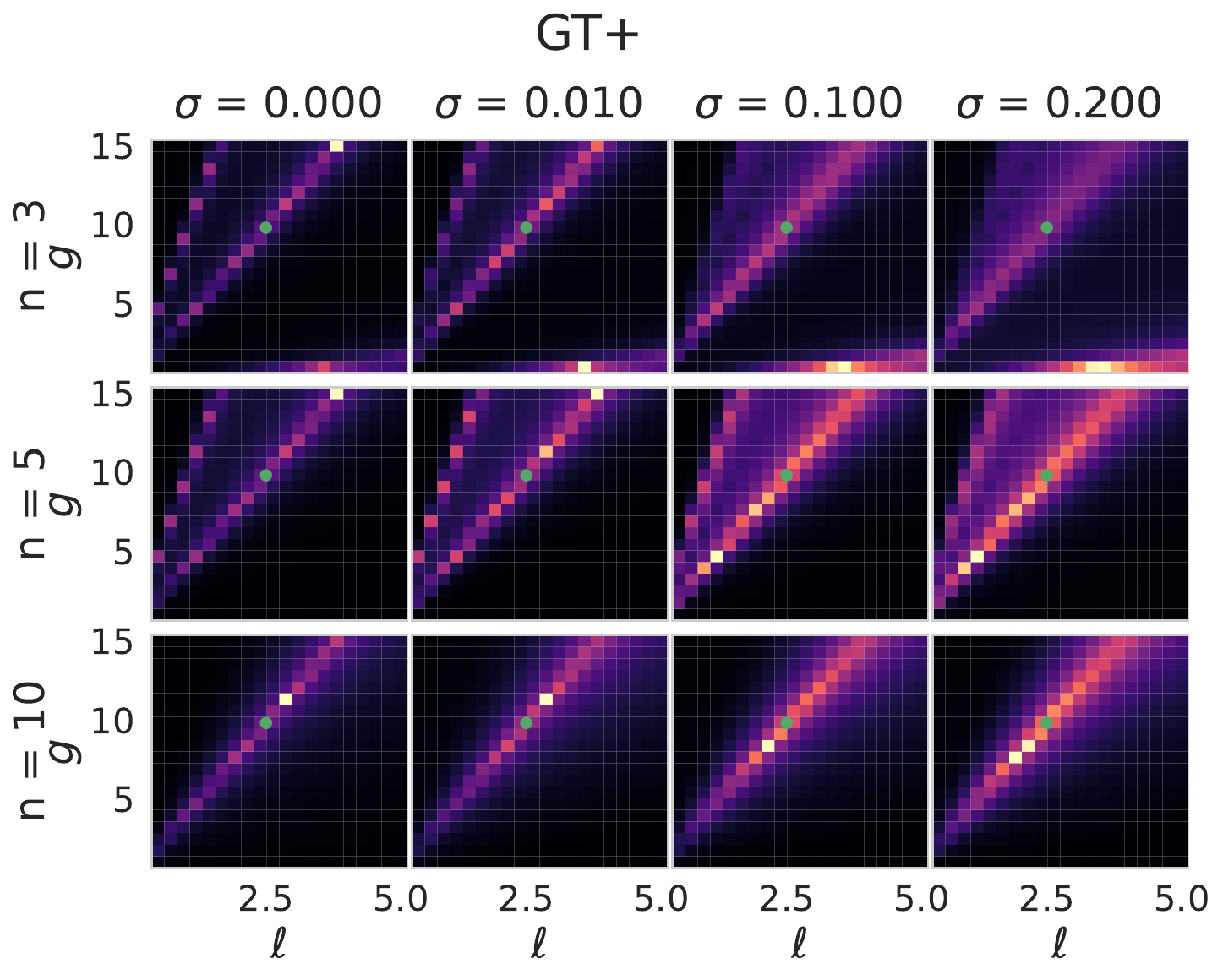}
  \end{minipage}
\caption{Re-scaled log posteriors  $-\log(\max(\log(p(\ell, g)))-\log(p(\ell,g))+1)$ in the $(g,\ell)$-plane marginalized over the calibration parameters $A$ and $l$ for all combinations of $n$ and $\sigma$. The reference Bayesian estimate for the known solution (upper left plot) is compared to GT (upper right), GT+ (lower right), and PIGP (lower left). The elongated structures show a strong correlation between $g$ and $\ell$, and we see multimodality in the reference and GT posteriors. We do not show colorbars because the likelihood scales vary substantially between panels; the comparison is qualitative and focuses on the relative structure of each likelihood rather than absolute values compared across panels.}
\label{fig:Ex2_posteriors}
\end{figure}

\FloatBarrier

\subsection{Experiment 2: Coupled transport}
\label{sec:Ex1}
We will now assess the performance of ghost tasking on a system of PDEs with two input variables, shown in Eq.~\ref{eq:ex1} that includes all possible cases and complications, but is still simple enough to be solved analytically as well. 
\begin{align}
\partial_t u + x \partial_x u + \alpha\partial_x v &= 0 \nonumber\\
\partial_t v + x\partial_x v &= 0
\label{eq:ex1}
\end{align}
A physical interpretation of this system could be a coupled transport in a shear flow field with some coupling constant $\alpha$ for example. Note that the coefficients are polynomials in $x$, so the second Weyl Algebra (over $\vv x = (x,t)$) is a suitable choice of operator ring $D$. 
We define our operator matrix $R^{\alpha}(\vv x)$ by rewriting Eq.~\ref{eq:ex1}
\begin{equation}
        R^{\alpha}(\vv x)\vec{\xi}(\vv x) =
    \begin{pmatrix}
 \partial_t+x\partial_x & \alpha\partial_x\\
        0&\partial_t+x\partial_x \\
    \end{pmatrix}\begin{pmatrix}
        u(x, t)\\v(x, t)
    \end{pmatrix} = \vec{0}.
    \label{eq:Ex1system}
\end{equation}
Alas, the original system is torsion and is therefore not parametrizable.
We will try to render the system parametrizable by introducing a ghost task and a corresponding column vector $\Lambda(\vv x)$. We have two options going forward: Using theorem \ref{th:GT}, we can construct an augmented system that is projective (in short "proj"), i.e. guaranteed to be able to realize all solutions in any sufficiently smooth function space, by choosing $\Lambda(\vv x)$, and the corresponding module $\tau(\Lambda)D^m$ its projection generates, such that $ext^1_D(M, D)/(\tau(\Lambda)D^m) = 0$, i.e. $\tau(\Lambda)D^m \cong ext^1_D(M, D)$. On the other hand, the condition $ext^1_D(ext^1_D(M,D)/(\tau(\Lambda)D^m), D) = 0$ of Proposition \ref{prop:GT2} to render the augmented system torsion-free (in short "tf") is weaker, and finding a valid $\Lambda(\vv x)$ in this sense by brute force search can be significantly faster, and heuristically results in shorter, lower order term structures in the resulting parametrization matrix $Q^\alpha(\vv x)$. On the downside, since we work over the second Weyl algebra, torsion-freeness does not guarantee that our parametrization captures all solutions over $C^{\infty}$ (see chapter \ref{subsec:params}). Although this may be mathematically unsatisfying, from an engineering and efficiency standpoint we nonetheless recommend using the torsion-free version whenever it is applicable, as we demonstrate in this experiment and the next. For this experiment, we will explore and compare both versions. 
For the projective system, a suitable $\Lambda(\vv x)$ is $\Lambda_{proj}(\vv x) = (1,t)$, resulting in the parametrization matrix
\begin{equation}
\resizebox{\textwidth}{!}{$\displaystyle
Q^{\alpha}_{\mathrm{GT\,proj}}(\vv x) :=
\begin{pmatrix}
 \alpha t^{2} x\partial_x^{2} + \alpha t^{2}\partial_t\partial_x - 2\alpha t x\partial_x^{2} - t x^{2}\partial_x^{2} - 2\alpha t\partial_t\partial_x + \alpha t^{2}\partial_x - 2tx\partial_x\partial_t + 2x^{2}\partial_x^{2} - 2\alpha t\partial_x - t\partial_t^{2} + 4x\partial_x\partial_t + 4\alpha\partial_x + 2\partial_t^{2} + t\partial_t - 3x\partial_x - 5\partial_t + 1
\\
 -t^{2}x^{2}\partial_x^{2} - 2t^{2}x\partial_t\partial_x + 2tx^{2}\partial_x^{2} - t^{2}\partial_t^{2} + 4tx\partial_x\partial_t + 2t\partial_t^{2} + t^{2}\partial_t - 2tx\partial_x - 4t\partial_t - 2x\partial_x - 2\partial_t + 4
\\
 -tx^{3}\partial_x^{3} - 3tx^{2}\partial_t\partial_x^{2} + 2x^{3}\partial_x^{3} - 3tx\partial_t^{2}\partial_x + 6x^{2}\partial_t\partial_x^{2} - 2tx^{2}\partial_x^{2} - t\partial_t^{3} + 6x\partial_x\partial_t^{2} - tx\partial_x\partial_t + 2\partial_t^{3} + t\partial_t^{2} - 6x\partial_x\partial_t - 6\partial_t^{2} - 2x\partial_x + 2\partial_t
\end{pmatrix}.
$}
\label{eq:Ex1_GT1}
\end{equation}
On the other hand, choosing the $\Lambda_{tf} = (1,\partial_t)$ renders the system torsion-free and leads to the parametrization matrix
\begin{align}
Q^{\alpha}_{GTtf}(\vv x) = \begin{pmatrix}
- \alpha \partial_t \partial_x+x \partial_x+\partial_t-1 \\
 \partial_t \partial_x x+\partial_t^{2}-\partial_t \\
 \partial_x^{2} x^{2}+2 \partial_t \partial_x x+\partial_t^{2}-\partial_t
    \end{pmatrix}.
    \label{eq:Ex1_GT2}
\end{align}

The corresponding (projective) PIGP parametrization, equivalent to choosing $\Lambda = \begin{pmatrix}
    1&0\\0&1
\end{pmatrix}$, is
\begin{align}
    Q^{\alpha}_{PIGP}(\vv x) = \begin{pmatrix}
        1 & 0\\ 0 & 1\\
        \partial_t+x\partial_x & \alpha\partial_x\\
        0&\partial_t+x\partial_x \\
    \end{pmatrix}.
    \label{eq:Ex1_PIGP}
\end{align} 
Comparing Eq.~\ref{eq:Ex1_GT1} and Eq.~\ref{eq:Ex1_GT2} to Eq.~\ref{eq:Ex1_PIGP}, we note that both $Q^\alpha_{GTproj}(\vv x)$ and $Q^\alpha_{GTtf}(\vv x)$ have three rows, and therefore needs three tasks, while $Q^\alpha_{PIGP}(\vv x)$ needs four tasks. Furthermore, $Q^\alpha_{GTproj}(\vv x)$ and $Q^\alpha_{GTtf}(\vv x)$ have one column, meaning that the latent space consists of one function, while $Q^\alpha_{PIGP}(\vv x)$ has two columns and therefore needs two independent functions spanning the latent space in order to express all solutions of our system.\\

To generate training data, we use the analytical solution obtained by the method of characteristics,
\begin{align}
    u(t,x) &= u_0(x\exp({-t})) - \alpha(1-\exp({-t}))v_0'(x\exp({-t})) \ , \nonumber \\
    v(t,x) &= v_0(x\exp({-t})) \ .
\end{align}
Providing initial conditions $u(x, 0) = 0$ and $v(x,0) = 3\sin(2x)$ and choosing the true parameter to be $\alpha = 2$ renders the solution unique. 
For the forward problem, both GT versions require conditioning on the full grid in 1 ghost task, while PIGP requires conditioning on the full grid in 2 tasks. We provide initial conditions in both cases. The domain of interest is $\Omega = \{(x,t)^T: x,t \in [0,1]\}$, but we place collocation points on a wider region: the fully conditioned tasks are conditioned on a $11\times10$ grid over $x \in [-0.1, 1], t \in [0,1]$, and the initial conditions of $u$ and $v$ are enforced on 11 linearly spaced points $(x,0)$ for $x \in [-0.1, 1]$. This extension outside $\Omega$ is possible because the parametrization is valid on all of $\mathbb{R}^2$, and we use it for stability reasons: In this example, the characteristics transport information along $x = x_0 e^t$, so every value near the corner $(x,t) = (0,1)$ traces back to a small neighborhood of $(0,0)$, making that region especially sensitive to variations near $(0,0)$. By also constraining on $x \in (-0.1, 0)$, we make sure that the initial condition is well represented near $(0,0)$. In total, this means we condition on $1\cdot 110 + 2\cdot 11 = 132$ training points for our ghost tasking methods and $2\cdot110 + 2\cdot11 = 242$ training points for the PIGP method. This data inequality stems from the different amount of tasks we need to condition on and ensures the same data density is provided to both methods. 
The test points lie on a linear $51 \times 51$ grid for each task, and only points that lie in the domain $\Omega$ are considered for the RMSE calculations. 

For the inverse problem, analogously to experiment 1, we systematically generated training points on $n\times n$ grids with $n = 3, 5, 10$ for all tasks. In our ablation studies \ref{ap:ablation}, we also considered the previously introduced "GT+" method using additional collocation points. Concretely, the collocation grid for the ghost task is augmented to a $10 \times 10$ instead of a $n\times n$ grid. Because we know the tasks corresponding to inhomogeneities in the PIGP method also to be exactly zero, we can also consider an augmented $10 \times 10$ collocation grid for the two inhomogeneity tasks in the PIGP method. We also study this augmented method in our ablation studies in \ref{ap:ablation} and analogously call this variant "PIGP+". 
While the additional collocation points improved the results slightly in both methods, we did not find the same level of improvement relative to the increased computation time as in experiment 1, so we only discuss the standard cases "GT" and "PIGP", with the same amount of training points in each task, in the main text. \\

Fig. \ref{fig:Ex1_forward} shows an example solution of the forward problem given the ghost tasks and the initial conditions for both the GT (upper two rows) and the PIGP method (lower two rows). 
We see that both methods perform comparably well with an RMSE below 0.006 on every task. Overall, the projective version of GT outperforms the other two slightly, while using less training data than PIGP. It shall be noted however, that inference for projective GT took about 25 seconds, while for torsion-free GT and PIGP, it took about 3 seconds.

\begin{figure}[h!]
\centering
  \includegraphics[width=\textwidth]{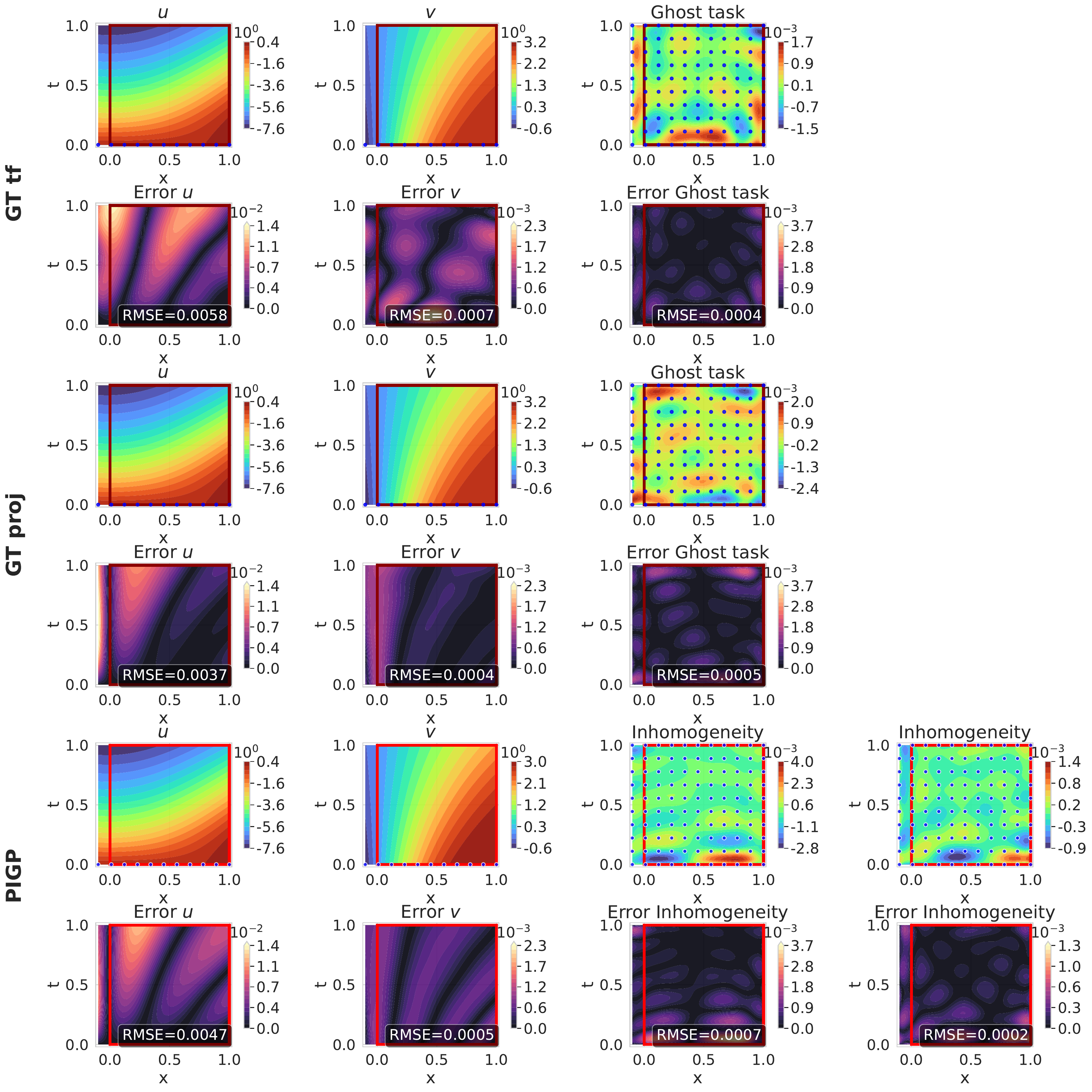}
\caption{Inferred solution (uneven rows) and absolute error (even rows) of the forward problem given training data for the ghost task/input tasks and the initial conditions. The torsion-free GT version is depicted in the first two rows, the projective GT version in the middle two rows, and PIGP in the lowest two rows. The left column shows the results for $u$, the middle column for $v$ and on the right we show the ghost task resp. input tasks for the inhomogeneity necessary. Conditioned points are marked with blue. The RMSE is computed within the region of interest, marked by a red box.}
\label{fig:Ex1_forward}
\end{figure}

The results of the inverse problem are shown in Fig. \ref{fig:Ex1_inverse}. 
All methods perform comparably across settings. This matches our expectations, since the underlying model is only loosely coupled (with ${\cal L}^{\alpha}(x,t)$ being a triangular matrix) -- a setting in which the loosely-coupled PIGP approach is already well suited. We especially note that the performances of the torsion-free and the projective GT versions are similar, however, since the parametrization matrix for projective GT (eq. \ref{eq:Ex1_GT1}) contains more terms, the kernel evaluation takes much more time, resulting in convergence times that were approximately 10 times longer than for the torsion-free GT and the PIGP case. In summary, both torsion-free ghost tasking  and projective ghost tasking recover the performance of PIGP in this experiment while using less training data.

\begin{figure}[h!]
\centering
  \includegraphics[width=\textwidth]{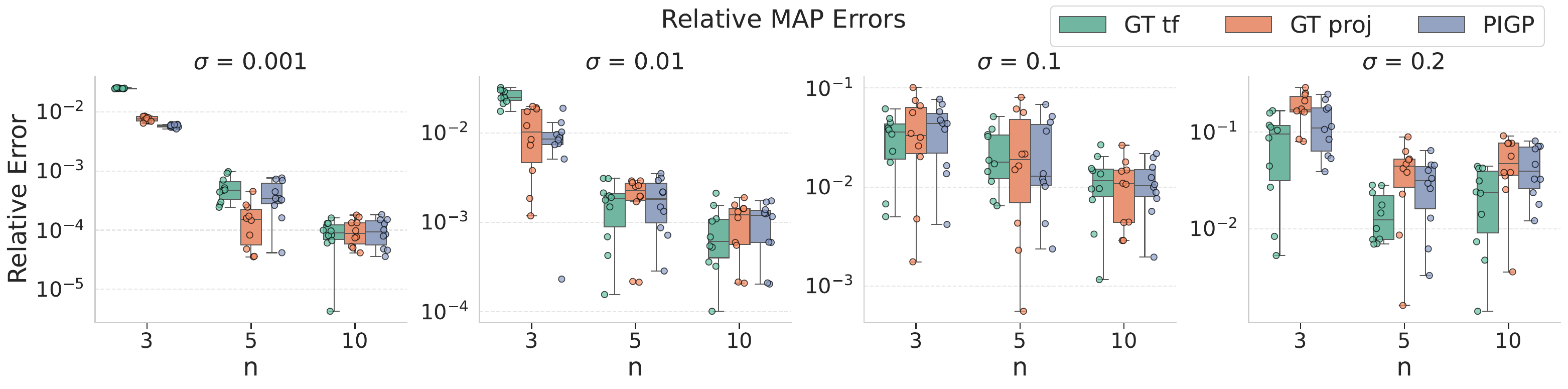}
\caption{Relative error as performance metric of the inverse problem of experiment 2, analogously to Fig. \ref{fig:Ex2_inverse} of experiment 1. The results for the torsion-free GT system are shown in green, the results for the projective GT version are shown in orange and the results for PIGP are shown in blue.}
\label{fig:Ex1_inverse}
\end{figure}

\subsection{Experiment 3: Anisotropic Magnetostatics}
Next we discuss a magnetostatics example as in \cite{chris_und_patrick}: For a simply connected domain $\Omega$, we can write the magnetostatic equations for the magnetic flux density $\vec{B}(x,y,z)$ in terms of a vector potential $\vec{A}(r)$ with $\nabla \times \vec{A} = \vec{B}$ through the curl-curl equation with Dirichlet boundary conditions
\begin{align}
    \nabla \times (\nu\nabla \times \vec{A}) = \vec{J}, &&\text{in } \Omega\\
    \vec{n} \times \vec{A} = 0 &&\text{on }\partial \Omega
\end{align} 
$\nu(x,y,z)$ refers to the reluctivity tensor, i.e. the inverse of the permeability, and $\vec{J}(x,y,z)$ is the current density.
We can write the curl operator $\nabla \times$ explicitly as an operator matrix
\begin{equation}(\nabla \times) = \begin{pmatrix}
    0&-\partial_z & \partial_y\\ \partial_z & 0 & -\partial_x\\ -\partial_y & \partial_x & 0
\end{pmatrix}.\end{equation}

For this example, we choose $\vec{J}(x,y,z)$ to be aligned with the $x$-axis, i.e. $\vec{J}(x,y,z) = \begin{pmatrix}
    J_x(x,y,z)&0&0\end{pmatrix}^T$.
We consider an inhomogeneous material in a box $\Omega = \{(x,y,z)^T : x,y \in (0, 1), z \in (1,2) \}$ with a local reluctivity $\nu(x,y,z) = \text{diag}(\nu_0, \nu_0, \nu_0 z)$. It behaves symmetrically in the $x-y$ plane but its reluctivity increases in $z$ direction with its position. This could be e.g. an engineered meta-material, or an approximation for a layered composite material. Then, we can formulate our problem in terms of operator matrices
\begin{align}
(\nabla \times)\nu(\nabla \times) \vec{A} = \vec{J} \Leftrightarrow
\left(\begin{array}{cccc}
-\partial_y^{2} z\nu_0-\partial_z^{2} \nu_0  & \partial_x \nu_0\partial_y z & \partial_z \nu_0  \partial_x & 0 
\\
 \partial_x\nu_0 \partial_y \nu_0z & -\partial_x^{2} z-\partial_z^{2} \nu_0  & \partial_z \nu_0  \partial_y & 0 
\\
 \partial_z \nu_0  \partial_x & \partial_z \nu_0  \partial_y & -\partial_x^{2} \nu_0 -\partial_y^{2} \nu_0  & -1 
\end{array}\right)
\begin{pmatrix}
        A_1\\A_2\\A_3\\J_x
    \end{pmatrix} = \vec{0}.
\end{align}
This system is not parametrizable. Since it is a system of PDEs, we can distinguish between a torsion-free and a projective augmented system. We were able to calculate a projective $\Lambda_{proj} = (0, 1, xz)$. However, calculating the corresponding parametrization matrix was not feasible on a laptop in a reasonable time. Instead, we are opting to build a torsion-free system by adding the column vector $\Lambda = (0, 1, 0)^T$, and show that already the torsion-free system is indeed sufficient to vastly outperform PIGP in this experiment. The torsion-free parametrization matrix is
\begin{align}
Q^{\nu_0}_{GT}(\vv x) = \begin{pmatrix} 
- \partial_x & - \partial_y 
\\
 - \partial_y &  \partial_x 
\\
 - \partial_z & 0 
\\
 0 &  \partial_x^{2}  \partial_y z \nu_{0}+ \partial_y^{3} z \nu_{0}+ \partial_y \, \partial_z^{2} \nu_{0} 
\\
 0 & - \partial_x^{3} z \nu_{0}- \partial_x \, \partial_y^{2} z \nu_{0}- \partial_x \, \partial_z^{2} \nu_{0} 
\end{pmatrix}.
\end{align}
It is interesting to note that the first column of $Q^{\nu_0}_{GT_1}(\vv x)$ is $-\nabla \vec{A}$, reflecting the gauge invariance of the vector potential. 
In this experiment, conditioning on $J_x$ and on the boundaries of $\vec{A}$ is sufficient for solving the forward problem, but the vector potential $\vec{A}$ is a latent variable that we cannot observe. What we are actually interested in is the magnetic flux density $\vec{B}(x,y,z)$, that we can calculate by taking the curl of the vector potential $\vec{B} = \nabla \times \vec{A}$. Similar to the first experiment, we can add tasks to infer $\vec{B}$ by taking the curl of the first three rows of $Q^{\phi}_{GT_1}(\vv x)$, resulting in a total parametrization matrix 
\begin{align}
Q^{\nu_0}_{GT_{+B}}(\vv x) = \begin{pmatrix} 
- \partial_x & - \partial_y 
\\
 - \partial_y &  \partial_x 
\\
 - \partial_z & 0 
\\
 0 &  \partial_x^{2}  \partial_y z \nu_{0}+ \partial_y^{3} z \nu_{0}+ \partial_y \, \partial_z^{2} \nu_{0} 
\\
 0 & - \partial_x^{3} z \nu_{0}- \partial_x \, \partial_y^{2} z \nu_{0}- \partial_x \, \partial_z^{2} \nu_{0} \\
 0 & -\partial_x\partial_z\\
 0 & -\partial_y\partial_z\\
 0 & \partial_x^2 + \partial_y^2
\end{pmatrix}.
\label{eq:Ex3_B_GT}
\end{align}

The PIGP method uses $J_y$ and $J_z$ as additional input tasks:
\begin{align}
    \begin{pmatrix}
        -\partial_y^{2} \nu_0z-\partial_z^{2} \nu_0  & \partial_x \nu_0\partial_y z & \partial_z \nu_0  \partial_x &-1&0&0\\ 
 \partial_x \nu_0\partial_y z & -\partial_x^{2} \nu_0z-\partial_z^{2} \nu_0  & \partial_z \nu_0  \partial_y &0&-1&0\\
 \partial_z \nu_0  \partial_x & \partial_z \nu_0  \partial_y & -\partial_x^{2} \nu_0 -\partial_y^{2} \nu_0  &0&0&-1&\end{pmatrix}\begin{pmatrix}
        A_x\\A_y\\A_z\\J_x\\J_y\\J_z
    \end{pmatrix} = \begin{pmatrix}
    0\\0\\0
    \end{pmatrix}
\end{align}
and the corresponding parametrization matrix (already including the curl operator to calculate $\vec{B}$) is
\begin{align}
    Q^{\nu_0}_{PIGP}(\vv x) = \begin{pmatrix}
        1 & 0 & 0\\
        0 & 1 & 0\\
        0&0&1\\
         -(\partial_y^{2} z+\partial_z^{2}) \nu_0  & \partial_x \partial_y z\nu_0 & \nu_0\partial_z   \partial_x \\
 \partial_x \partial_y z\nu_0 & -(\partial_x^{2} z+\partial_z^{2}) \nu_0  & \partial_z \nu_0  \partial_y \\
 \partial_z \nu_0  \partial_x & \partial_z \nu_0  \partial_y & -(\partial_x^{2} +\partial_y^{2} )\nu_0\\
 0&-\partial_z & \partial_y\\ \partial_z & 0 & -\partial_x\\ -\partial_y & \partial_x & 0
    \end{pmatrix}.
    \label{eq:Ex3_B_PIGP}
\end{align}

We choose the true parameter value in the simulations to be $\nu_0 = 1$. Given a fixed current density $\vec{J}$, the corresponding solution for the vector potential $\vec{A}$ and the magnetic flux density $\vec{B}$ is
\begin{align}
J = 2\nu_0 \begin{pmatrix} 
z^3 - 3z^2 + 2z - y + y^2 \\ 0 \\ 0 
\end{pmatrix} &
    \vec{A} = \begin{pmatrix} 
y(1-y)(z-1)(z-2) \\  0 \\ 0 
\end{pmatrix} & 
\vec{B} = \begin{pmatrix} 
0 \\ 
(y - y^2)(2z - 3) \\ 
-(1 - 2y)(z^2 - 3z + 2)
\end{pmatrix}.\end{align}

For the forward problem, we want to infer $\vec{B}$ provided $\vec{J}$ and the boundary constraint $\vec{n} \times A = 0 \text{ on } \partial \Omega$. In the GT method, we encode our knowledge by conditioning on a $6\times 6 \times 6$ grid on $J_x$ and the ghost task, while in the PIGP method, we condition the tasks $J_x, J_y,$ and $J_z$ on these gridpoints. Per construction, $J_y$ and $J_z$ are enforced to be exactly zero in the GT method. 
For both methods, we enforced the boundary condition on the necessary surface points of the point grid, amounting to 
\begin{align}
    A_x(x_i, y = 0, z_j) = A_x(x_i, y = 1, z_j) = A_x(x_i, y_j, z = 1) = A_x(x_i, y_j, z = 2) &= 0 \nonumber\\
    A_y(x = 0, y_i, z_j) = A_y(x = 1, y_i, z_j) = A_y(x_i, y_j, z = 1) = A_y(x_i, y_j, z = 2) &= 0 \nonumber\\
    A_z(x = 0, y_i, z_j) = A_z(x = 1, y_i, z_j) = A_z(x_i, y = 0, z_j) = A_z(x_i, y = 1, z_j) &= 0 \quad \forall i,j = 1,...,n.
\end{align} 
This yields a total of $360$ boundary points for both methods and $2\cdot6^3 = 432$ inner points ($J_x$ and $\xi_{GT}$) for GT resp. $3\cdot6^3 = 648$ ($J_x, J_y, \text{ and } J_z$) for PIGP. 
The test set consists of a $18\times 18 \times 18$ grid across all tasks. 

The inverse problem is structured analogously to the previous experiments on $n^3$ training points on every task on a grid with $n = 3,5,7$. As in the second experiment, we show the effect of adding extra collocation points to the ghost/input tasks in the ablation studies in \ref{ap:ablation}. Since the vector potential $\vec{A}$ is not a measurable quantity, we did not use it for the inverse problem.  

\begin{figure}[h!]
\centering
  \includegraphics[width=\textwidth]{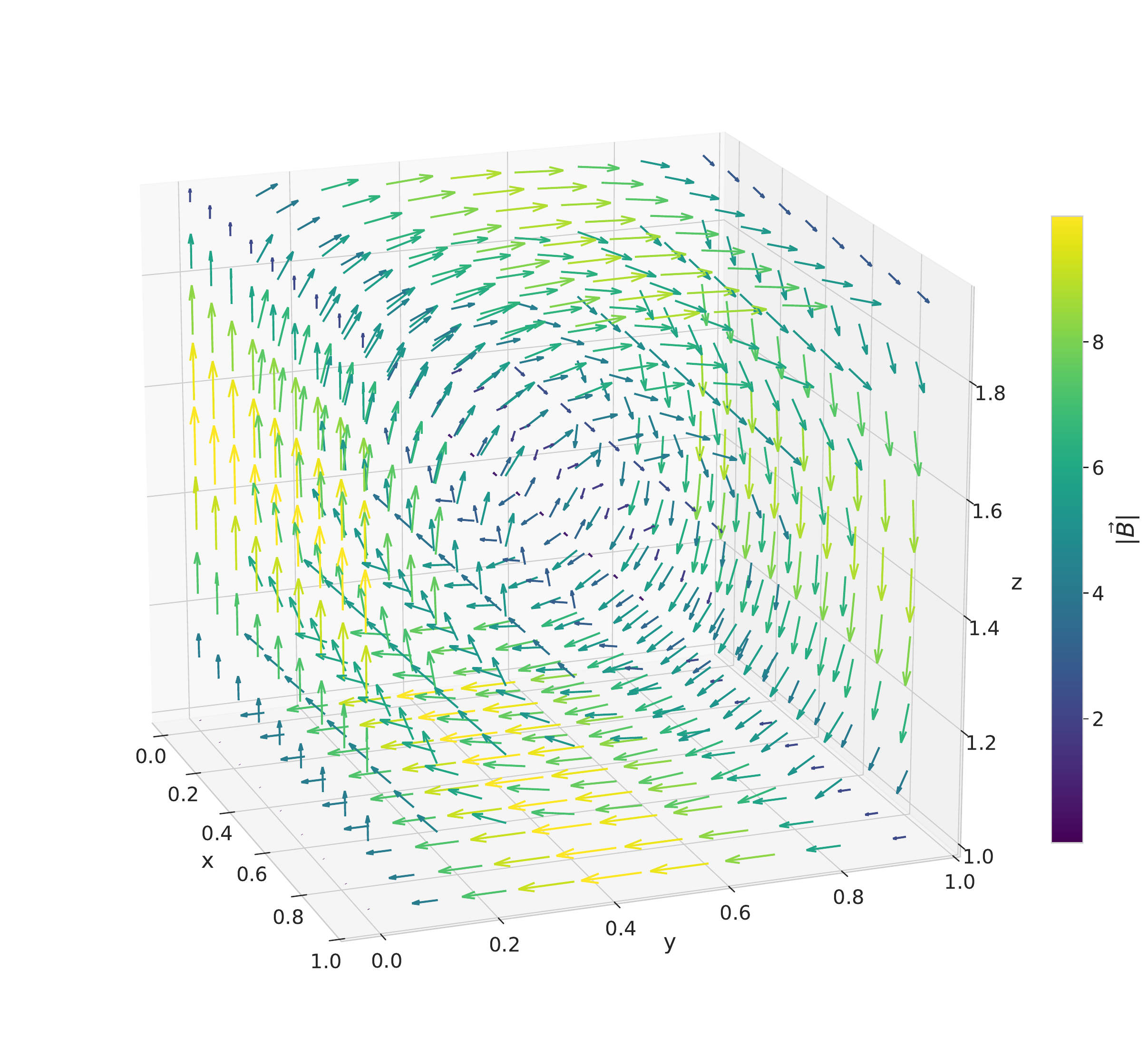}
\caption{Quiver plot of the magnetic flux density $\vec{B}$ inferred with GT in the forward problem of experiment 3 given Dirichlet boundary conditions for $\vec{A}$ and $6\times 6\times 6$ training points for $J_x$ and the ghost task $\xi_{GT}$. Direction of $\vec{B}$ is indicated with arrows, the strength is represented with the length of the arrows as well as their color.}
\label{fig:Ex3_quiver_GT}
\end{figure}

Since this is a three-dimensional problem, Fig. \ref{fig:Ex3_quiver_GT} shows a quiver plot of the inferred magnetic flux density $\vec{B}$ of the GT method as a visualization. 
Fig. \ref{fig:Ex3_forward_auxiliary} resp. Fig. \ref{fig:Ex3_forward_B} show a 2d slice of the conditioned tasks resp. the inferred $\vec{B}$ fields through the $y-z$ plane at $x = 0.6$. In the corner of each subplot we show the RMSE calculated over the whole volume. 

\begin{figure}[h!]
\centering
  \includegraphics[width=\textwidth]{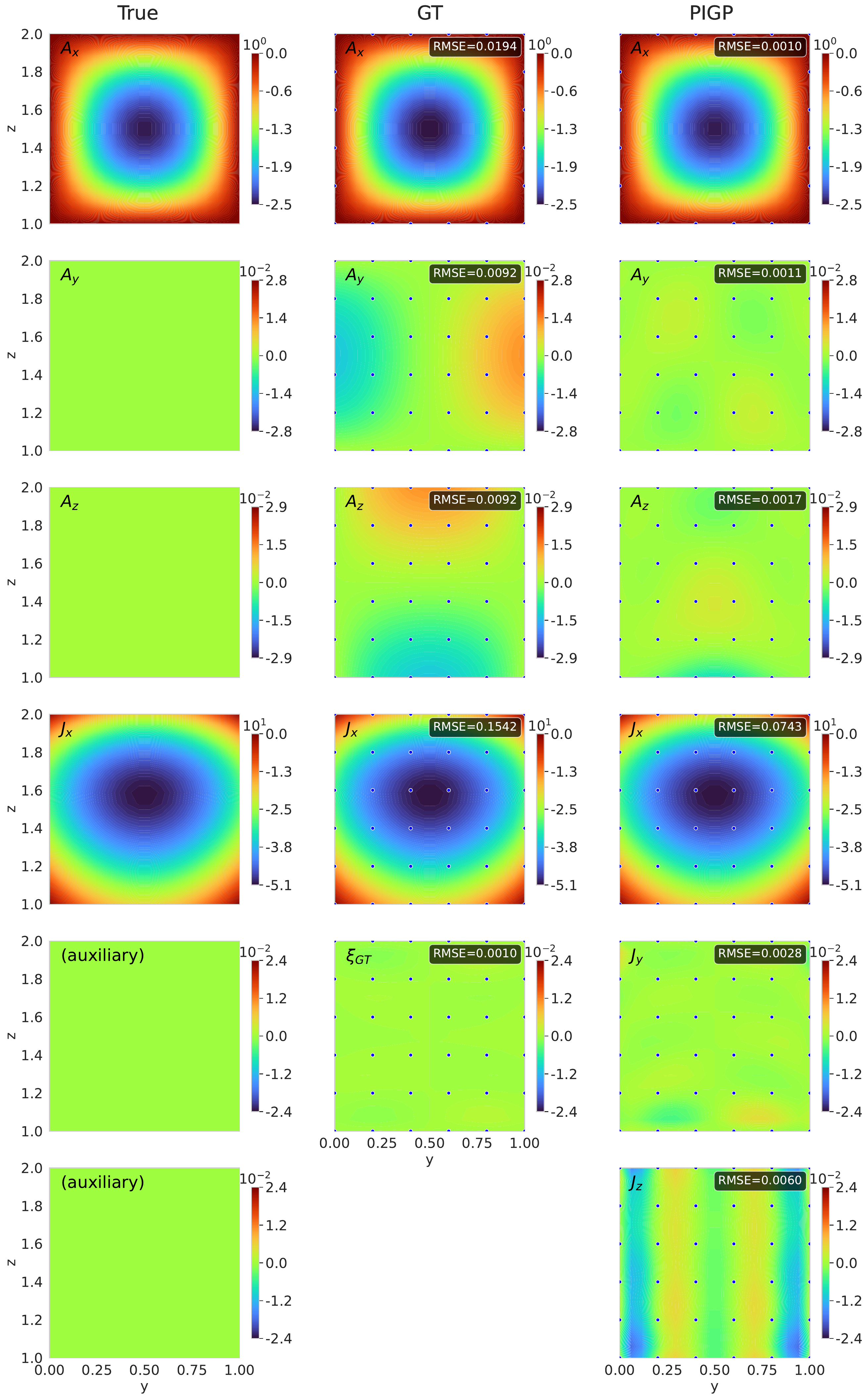}
\caption{Visualization of the conditioned tasks in the forward problem of Experiment 3. The left column shows two dimensional slices though the $y-z$ plane at $x = 0.6$ of the analytic ground truth, the middle column shows the ghost tasking method, and the right column shows PIGP. The rows depict the (x,y,z) components of the vector potential $\vec{A}$, the $x$ component of the current density $J_x$, and the ghost task $\xi_{GT}$ for GT, and $J_y$, $J_z$ respectively for PIGP. Training data is marked with blue points and the RMSE over the full volume is written in the corner of each task.}
\label{fig:Ex3_forward_auxiliary}
\end{figure}
\begin{figure}[h!]
\centering
  \includegraphics[width=\textwidth]{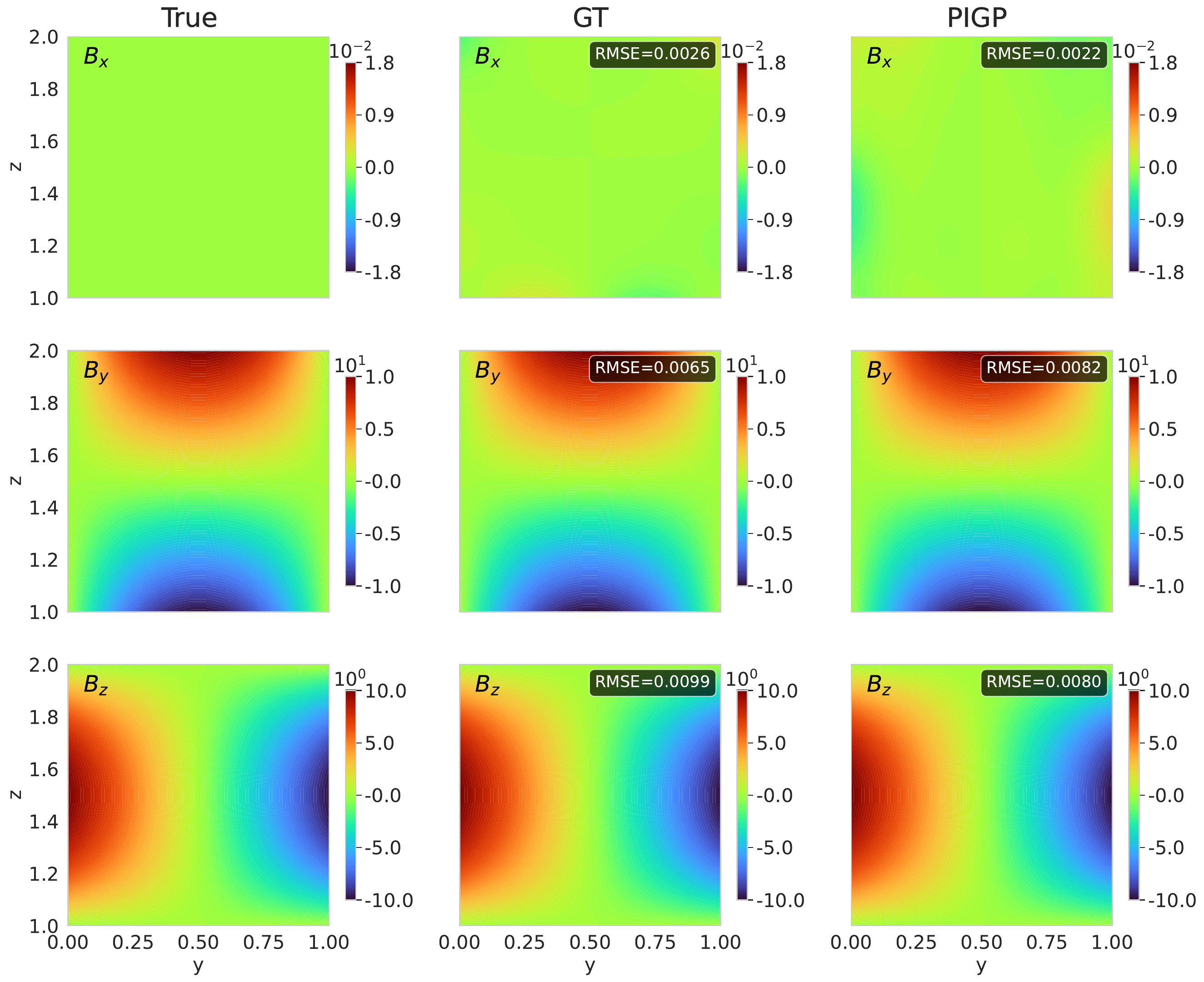}
\caption{Visualization of the solution of the forward problem of Experiment 3 for the magnetic flux density components. The left column shows two dimensional slices though the $y-z$ plane at $x = 0.6$ of the analytic ground truth, the middle column shows the ghost tasking method, and the right column shows PIGP. The rows depict the (x,y,z) components of the inferred quantity of interest, the magnetic flux density $\vec{B}$. No training data on $\vec{B}$ is used. The RMSE over the full volume is written in the corner of each task.}
\label{fig:Ex3_forward_B}
\end{figure}

As in the previous experiments we see good performances across both methods with RMSEs between 0.002 and 0.01 in the inferred magnetic flux density components. 
In the inverse problem setting shown in Fig. \ref{fig:Ex3_inverse} however we observe a significantly better performance of GT than PIGP, which can be understood when looking at the posterior distribution of $\nu_0$: Similar to experiment 1 (Fig. \ref{fig:Ex2_posteriors}), PIGP has again a rather biased posterior (shown in the lowest panel of Fig. \ref{fig:Ex3_posterior}). This can be attributed to the lack of informativeness of the data-fit term of the log marginal likelihood (Eq.~\ref{eq:mll}) shown in the middle panel of Fig. \ref{fig:Ex3_posterior}. The complexity penalty is similarly peaked for both methods, favouring small values of $\nu_0$. GT's data-fit term is able to compensate this bias by peaking sharply around the correct value $\nu_0\approx 1.$ On the contrary, PIGPs data-fit term is flatter, i.e. less informative. The reason for this may lie in the column rank and therefore the coupling strength of the parametrization matrices, discussed in more detail in chapter \ref{sec:discussion}. 

\begin{figure}[h!]
\centering
  \includegraphics[width=\textwidth]{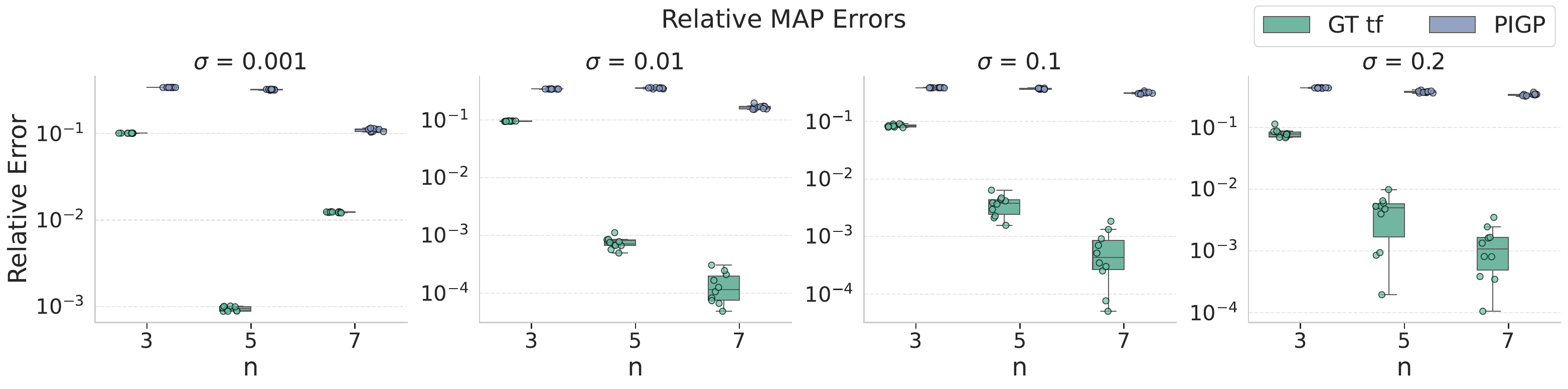}
\caption{Relative error of the MAP estimate of the parameter $\nu_0$ as performance metric of the inverse problem of experiment 3, analogously to Fig. \ref{fig:Ex2_inverse} and \ref{fig:Ex1_inverse}. We see a clear advantage of the GT method (green) as compared to PIGP (blue) in all experiments.}
\label{fig:Ex3_inverse}
\end{figure}

\begin{figure}[h]
\centering
  \includegraphics[width=0.9\textwidth]{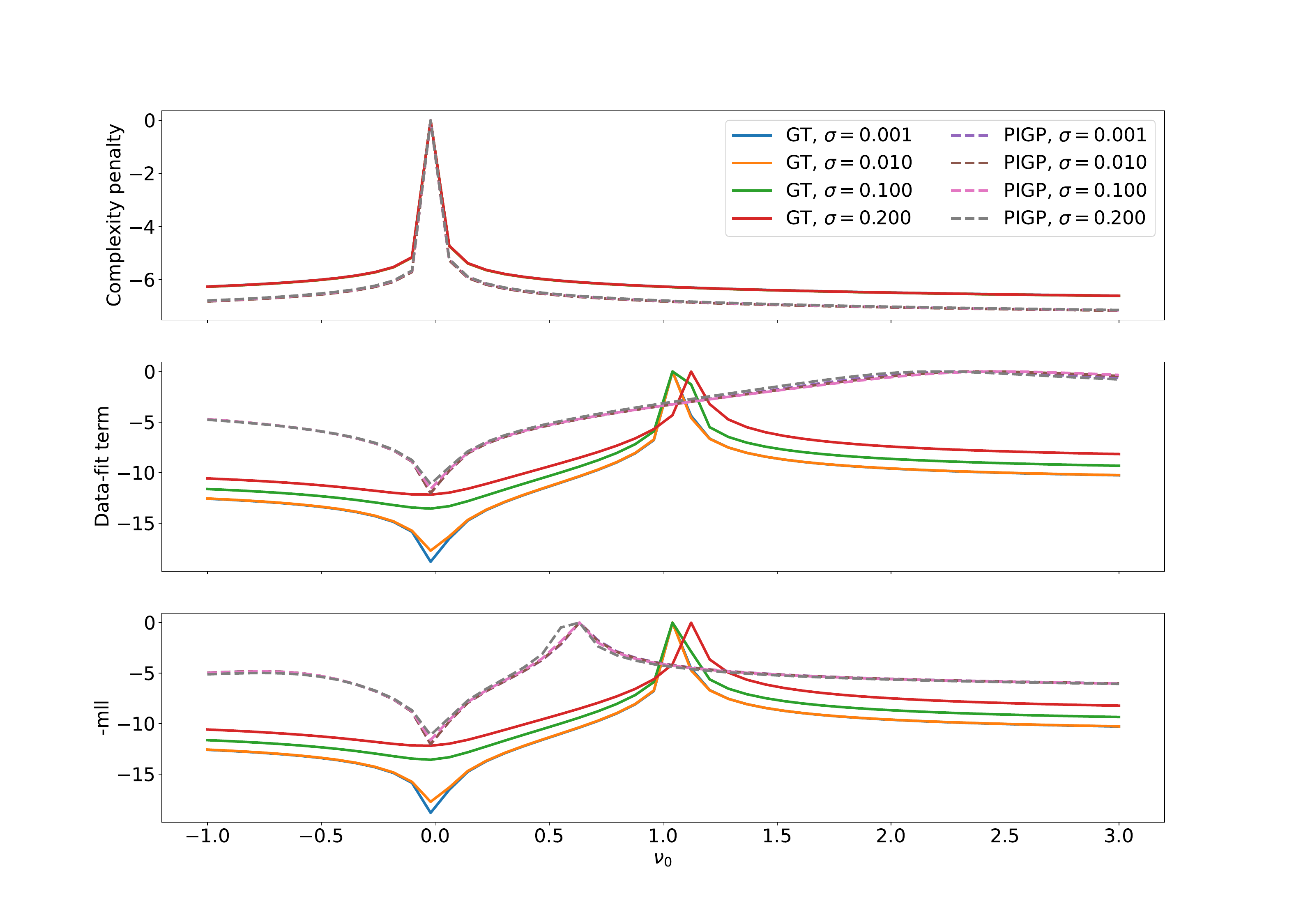}
\caption{Re-scaled components $p$  $-\log(\max(p)-p+1)$ of the log marginal likelihood (Eq.~\ref{eq:mll}) for GT (marked with x) and PIGP (marked with circles) as a function of $\nu_0$, conditioned on the respective MAP estimates of the calibration parameters $A$ and $l$ for the $n = 3$ training data set and all values $\sigma$ evaluated on 50 points $\nu_0 \in [-1, 3]$. The first row shows the complexity term, the second row shows the data-fit term, and the last row shows the total log marginal likelihood. Curves are not normalized due to re-scaling; the relevant information is conveyed by the position and width of the peaks.}
\label{fig:Ex3_posterior}
\end{figure}

\FloatBarrier
\section{Discussion and Qualitative Comparison}
\label{sec:discussion}
\subsection{Discussion of experimental results}
\label{subsec:discussionresults}
A central structural difference between ghost tasking and PIGP lies in the number of latent functions required to parametrize the solution space, and this difference is a direct consequence of the task-counting results established in chapter~\ref{sec:theorem}. Recall that the PIGP construction of Raissi et al. \cite{raissi_machine_2017} always requires $q$ additional inhomogeneity tasks, and the resulting minimal parametrization matrix $Q^{\phi}_{PIGP}$ correspondingly has $q$ columns, i.e., $q$ latent functions. Ghost tasking, by contrast, requires at most $m$ ($m \in \{1,2\}$) tasks additional to the $q\leq p\leq 2q-m$ tasks from the system matrix $R$, so that the number of latent functions $n = p+m-q$ ranges between $m$ and $q$, with the upper bound recovering PIGP exactly (chapter~\ref{subsec:ghostprops}). Fewer latent functions mean that a fixed amount of training data constrains each latent function more strongly, resulting in a more sharply peaked, informative data-fit term.

We observe this effect across our experiments. In Experiment 3, the most extreme case, the parametrization matrices $Q^{\nu_0}_{GT_{+B}}$ and $Q^{\nu_0}_{PIGP}$ (Eq.~\eqref{eq:Ex3_B_GT} and \eqref{eq:Ex3_B_PIGP}) illustrate this directly. Since the first three tasks, corresponding to $\vec{A}$, are not observed in the inverse problem, $Q^{\nu_0}_{GT_{+B}}$ effectively has column rank $n=1$: a single latent function is sufficient to fit all training data across five tasks. This constrains the latent function strongly, yielding a data-fit term sharply peaked around the true value of $\nu_0$ (Fig. \ref{fig:Ex3_posterior}). In contrast, $Q^{\nu_0}_{PIGP}$ has column rank 3, so three latent functions must be fit from training data across six tasks. Despite the larger total amount of training data, each latent function is effectively less constrained, allowing misfits in $\nu_0$ to be compensated for and producing a flatter, less informative data-fit term that fails to fully counteract the complexity penalty.

The same effect, though weaker, appears in Experiments 1 and 2: in Experiment 1, $Q^\alpha_{GT}$ (Eq.~\eqref{eq:Ex2_GT}) has column rank 2, while $Q^\alpha_{PIGP}$ (Eq.~\eqref{eq:Ex2_PIGP}) has column rank 3, and the broader posteriors of PIGP in Fig. \ref{fig:Ex2_posteriors} reflect this effect.  In Experiment 2, $Q^{(\ell,g)}_{GT}$ (Eq.~\eqref{eq:Ex1_GT2}) has column rank 1, while $Q^{(\ell,g)}_{PIGP}$ (Eq.~\eqref{eq:Ex1_PIGP}) has column rank 2.

\subsection{Practical considerations and limitations}
\label{subsec:limitations}
Beyond the structural aspects discussed above, some practical considerations arise when applying ghost tasking.

First, finding a suitable generator $\Lambda$ is algorithmic for holonomic systems (meaning the dimension of the system equals the number of variables), where it can be computed directly using \texttt{Macaulay2}~\cite{macaulay2} (\ref{ap:macaulay}). More generally, Stafford's two-generator theorem (Th. \ref{th:staffordserre}) admits a constructive proof \cite{leykin2004, hillebrand, quadrat2006} and while they are in principle implemented in \texttt{OreModules}'s \texttt{Stafford} package and \texttt{Macaulay2}'s \texttt{Dmodules} package, to our knowledge, they can currently not be applied to extension modules in those algebra programs. In practice, we therefore rely on a heuristic ansatz for $\Lambda(\vv x)$, that we verify computationally (see chapter~\ref{subsec:ghostprops} and \ref{ap:macaulay}).

Second, ghost tasking has the advantage that the ghost tasks are known to be exactly zero, so collocation points can be added freely without requiring measurements or real data. However, adding collocation points trades off accuracy against computation time, which scales cubically with the total amount of training data. For oscillatory systems, we observed that additional collocation points yield substantial improvements in accuracy (Fig. \ref{fig:Ex2_inverse}); for the remaining experiments, the improvement is comparatively modest (Tab. \ref{tab:Ex1} and \ref{tab:Ex3} in \ref{ap:ablation}).

Third, as with any parametrization method applied to systems with non-constant coefficients, the resulting parametrized kernels can become numerically ill-conditioned, which may affect stability in sharply peaked likelihood regimes \cite{mypsdpaper}. We refer to the extensive literature on stabilization and regularization techniques, e.g.\ \cite{burger}, for approaches to mitigating this issue.

Lastly, a practical trade-off concerns the choice between the projective and the torsion-free variant of ghost tasking introduced in chapter~\ref{subsec:ghostprops}. Although the two variants are largely comparable in predictive performance, requiring projectivity incurs greater computational overhead both in the computer algebra required to determine $\Lambda(\vv x)$ and the parametrization matrix, and in evaluation time, owing to the heuristically more complex terms arising in the projective parametrization. We therefore recommend using the weaker, torsion-free variant in practice. This choice carries the theoretical caveat that parametrizability of the full solution space is no longer guaranteed, but at least in our experiments, we did not encounter a case in which this limitation was practically relevant.

\subsection{Qualitative Method Comparison}
To contextualize ghost tasking within existing hybrid GP methods, we need to compare the methods with respect to several metrics: (i) reliance on analytical prior knowledge versus algorithmic construction, (ii) structural restrictions on the admissible systems, notably parametrizability and classes of coefficients, and (iii) the trade-off between model complexity and performance, particularly in inverse problems.

Mercer-type kernels, such as those discussed in  \cite{albert_gaussian_2019, dong_kriging, dong_diss, boogaart}, can achieve strong performance both in inverse and in forward problems, when applicable. However, they rely on prior analytical knowledge of the system in that the kernel is constructed from the form of the fundamental solutions. This is also the case for latent-force models \cite{alvarez_linear_2013}, which use the system's Green's function. These approaches are therefore non-algorithmic and tailored to specific problem classes where advanced prior knowledge beyond the general form of the equation is available or obtainable. While highly effective in settings where such knowledge is available, their applicability is inherently limited by the availability of such fundamental solutions; moreover, latent force models, as one of the earliest approaches, have since been outperformed by more recent methods such as LODE-GPs ("Linear Ordinary Differential Equations") \cite{besginow_constraining_2022} that construct kernels based on the Smith normal form of the system matrix and are applicable to ordinary systems with constant coefficients. 

Parametrization-based methods such as PIGPs \cite{raissi_machine_2017} and the method of \cite{lange-hegermann_algorithmic_2018}, as well as our ghost tasking, enforce constraints through operator structure on multiple correlated tasks (outputs, and inhomogeneities or ghost tasks) and find solutions through subsequent pointwise conditioning, rather than requiring known functions. Their key advantage is that they are algorithmic, i.e., applicable without requiring analytical knowledge of the solution.
A central trade-off in these methods lies in the number of auxiliary tasks and with it latent functions as discussed previously. Increasing the number of tasks improves generality by introducing more latent functions, but weakens the coupling between the tasks due to the larger latent function space. This can degrade performance in inverse problems (as seen for example in \cite{mymaxent}), and requires more training points to constrain the additional tasks and latent functions. 
Conversely, reducing the number of tasks and latent functions strengthens coupling and improves data efficiency, but risks restricting admissible systems and applicability of the method to parametrizable systems.
Ghost tasking can be understood as the sweet spot in this trade-off: it preserves strong coupling while relaxing parametrizability constraints for general systems. 

 Ghost tasking is theoretically applicable to a broad class of systems, including those with meromorphic coefficients. In practice, however, this generality is, as in \cite{lange-hegermann_algorithmic_2018}, mediated by the need to compute parametrizations via computer algebra programs. As a consequence, it inherits the limitations of currently available computer algebra implementations, which are largely restricted to polynomial or rational coefficients (e.g., via Weyl algebras). Other approaches that seek to circumvent or eliminate the parametrizability constraints are "linear ordinary differential equation Gaussian processes" LODE-GPs \cite{besginow_constraining_2022} for ODEs and "Ehrenpreis-Palamodov Gaussian Processes" EPGPs \cite{harkonen_gaussian_2023} for PDEs, but they are only applicable to systems with constant coefficients. 

In certain cases, the limitation to rational coefficients due to the algorithms available in computer algebra programs that we use in ghost tasking can be mitigated in practice in so-called commuting systems, where coefficient functions commute with differential operators. This allows the use of simpler algorithms over commutative polynomial rings as used in \cite{lange-hegermann_algorithmic_2018}, yielding tractable parametrizations when such assumptions are justified. However, this approach introduces implicit structural constraints through the need to verify $R^{\phi}(\vv x)Q^{\phi}(\vv x) \overset{!}{=} 0$ considering the specific coefficient functions explicitly and is therefore only applicable in restricted settings (see \ref{app:commutative} for an example).

In future work it may be interesting to investigate parametrizations based on the Jacobson Normal Form (JNF) \cite{jacobson}, analogous to the Smith Normal Form employed in \cite{besginow_constraining_2022}. While JNF applies only to ODEs, it may be able to work with general coefficient classes and reduce the latent function space of ODE systems even further. Another promising direction is the integration of ghost tasking with approaches that incorporate boundary conditions directly into the kernel. Since we have shown that ghost tasking shares the structure of the parametrized GPs of \cite{lange-hegermann_algorithmic_2018} and \cite{raissi_machine_2017}, boundary-condition constructions that have been applied to either of these methods, such as \cite{pmlr-v130-lange-hegermann21a} or \cite{gulian2022gaussian}, could likely be transferred with only minor modifications. More broadly, extending parametrization-based methods to the wider class of weak formulations of differential operators remains an important direction for future research.
\section{Conclusion}
We have introduced ghost tasking as a method to extend parametrizability to physical systems that on their own are not parametrizable, while retaining a minimal latent function space. We do so by adding  one, but at most two, additional ("ghost") tasks. This results in a minimal latent function space that is able to fit all solutions of the system. We show that the physics-informed Gaussian Process (PIGP) method \cite{raissi_machine_2017} can itself be understood as a form of parametrization, alas one that introduces an often unnecessarily large number of auxiliary tasks and latent functions, rather than the minimal number for both that ghost tasking identifies. As such, ghost tasking can be understood as the optimally data-efficient parametrization on the spectrum between PIGP as in \cite{raissi_machine_2017} and the null space-based parametrization as in \cite{lange-hegermann_algorithmic_2018}.
We showed that ghost tasking achieves performance comparable to PIGP on forward problems while requiring fewer training data and outperforms PIGP when using the same amount of training data. On inverse problems, ghost tasking frequently attains several orders of magnitude higher accuracy, particularly in the low-data regime.
Our results demonstrate that a substantially broader class of differential equation systems can be encoded more efficiently than previously possible as priors in GPs, requiring fewer tasks and latent functions and less training data without sacrificing predictive performance.
As advances in constructive algebraic analysis and automated generator construction continue, we expect the range of systems that can be encoded optimally by ghost tasking to grow further, making the framework increasingly practical for more complex physical models.

\section{Acknowledgements}
J.M. would like to thank A. Quadrat for helpful discussions and A. Leykin for providing the foundation of the \texttt{Macaulay2} code snippets. J.M. and C.A. were financially supported by the Austrian Research Promotion Agency (FFG) project VENTUS within the AI for Green funding program, Grant no. 910263. S.R. was financially supported by the Austrian Science Fund (FWF), Grant DOI: 10.55776/J4774. This work has been carried out within the framework of the EUROfusion Consortium, funded by the European Union via the Euratom Research and Training Programme (Grant Agreement No 101052200 — EUROfusion). Views and opinions expressed are however those of the author(s) only and do not necessarily reflect those of the European Union or the European Commission. Neither the European Union nor the European Commission can be held responsible for them. 

\section{Declaration of generative AI and AI-assisted technologies in the manuscript preparation process}
During the preparation of this work the authors used Claude (Anthropic) and GPT (OpenAI) in order to help with code development, brainstorming, and text polishing. After using these tools, the authors reviewed and edited the content as needed and take full responsibility for the content of the published article.

  \bibliographystyle{elsarticle-num} 
  \bibliography{literature}

@article{serre,
  title={Serre’s reduction of linear partial differential systems with holonomic adjoints},
  author={Cluzeau, Thomas and Quadrat, Alban},
  journal={Journal of symbolic computation},
  volume={47},
  number={10},
  pages={1192--1213},
  year={2012},
  publisher={Elsevier}
}

@article{ogofserre,
  title={Serre’s reduction of linear functional systems},
  author={Boudellioua, Mohamed Salah and Quadrat, Alban},
  journal={Mathematics in Computer Science},
  volume={4},
  number={2},
  pages={289--312},
  year={2010},
  publisher={Springer}
}

@article{amc1311,
  title={A functorial approach to the behaviour of multidimensional control systems},
  author={Pommaret, Jean-Fran{\c{c}}ois and Quadrat, Alban},
  journal={International Journal of Applied Mathematics and Computer Science},
  volume={13},
  number={1},
  pages={7--13},
  year={2003},
  publisher={Uniwersytet Zielonog{\'o}rski. Oficyna Wydawnicza}
}

@article{stafford,
  title={A constructive study of the module structure of rings of partial differential operators},
  author={Quadrat, Alban and Robertz, Daniel},
  journal={Acta Applicandae Mathematicae},
  volume={133},
  number={1},
  pages={187--234},
  year={2014},
  publisher={Springer}
}

@book{rasmussen_gaussian_2006,
	title = {Gaussian {Processes} for {Machine} {Learning} Vol. 2},
	url = {https://mitpress.mit.edu/9780262182539/gaussian-processes-for-machine-learning/},
	language = {en-US},
	urldate = {2024-11-26},
	publisher = {MIT press Cambridge, MA},
	author = {Rasmussen},
	year = {2006},
}

@inproceedings{lange-hegermann_algorithmic_2018,
	title = {Algorithmic {Linearly} {Constrained} {Gaussian} {Processes}},
	volume = {31},
	url = {https://proceedings.neurips.cc/paper/2018/hash/68b1fbe7f16e4ae3024973f12f3cb313-Abstract.html},
	urldate = {2024-12-03},
	booktitle = {Advances in {Neural} {Information} {Processing} {Systems}},
	publisher = {Curran Associates, Inc.},
	author = {Lange-Hegermann, Markus},
	year = {2018},
}

@article{albert_gaussian_2019,
	title = {Gaussian {Processes} for {Data} {Fulfilling} {Linear} {Differential} {Equations}},
	volume = {33},
	copyright = {http://creativecommons.org/licenses/by/3.0/},
	issn = {2504-3900},
	url = {https://www.mdpi.com/2504-3900/33/1/5},
	doi = {10.3390/proceedings2019033005},
	language = {en},
	number = {1},
	urldate = {2024-11-28},
	journal = {Proceedings},
	author = {Albert, Christopher G.},
	year = {2019},
	note = {Number: 1
Publisher: Multidisciplinary Digital Publishing Institute},
	pages = {5},
}

@inproceedings{dong_kriging,
	address = {Dordrecht},
	title = {Kriging {Variables} that {Satisfy} the {Partial} {Differential} {Equation} $\Delta${Z} = {Y}},
	isbn = {978-94-015-6844-9},
	doi = {10.1007/978-94-015-6844-9_17},
	
	language = {en},
	booktitle = {Geostatistics},
	publisher = {Springer Netherlands},
	author = {Dong, Anne},
	editor = {Armstrong, Margaret},
	year = {1989},
	pages = {237--248},
}

@article{alvarez_linear_2013,
	title = {Linear {Latent} {Force} {Models} {Using} {Gaussian} {Processes}},
	volume = {35},
	issn = {1939-3539},
	url = {https://ieeexplore.ieee.org/abstract/document/6514873},
	doi = {10.1109/TPAMI.2013.86},
	number = {11},
	urldate = {2025-07-01},
	journal = {IEEE Transactions on Pattern Analysis and Machine Intelligence},
	author = {Álvarez, Mauricio A. and Luengo, David and Lawrence, Neil D.},
	month = nov,
	year = {2013},
	pages = {2693--2705},
}

@article{raissi_machine_2017,
	title = {Machine learning of linear differential equations using {Gaussian} processes},
	volume = {348},
	issn = {0021-9991},
	url = {https://www.sciencedirect.com/science/article/pii/S0021999117305582},
	doi = {10.1016/j.jcp.2017.07.050},
	urldate = {2025-07-01},
	journal = {Journal of Computational Physics},
	author = {Raissi, Maziar and Perdikaris, Paris and Karniadakis, George Em},
	month = nov,
	year = {2017},
	pages = {683--693},
}

@article{besginow_constraining_2022,
	title = {Constraining {Gaussian} processes to systems of linear ordinary differential equations},
	volume = {35},
	url = {https://proceedings.neurips.cc/paper_files/paper/2022/hash/bcef27c5825d1ed8757290f237b2d851-Abstract-Conference.html},
	urldate = {2025-01-21},
	journal = {Advances in Neural Information Processing Systems},
	author = {Besginow, Andreas and Lange-Hegermann, Markus},
	year = {2022},
	pages = {29386--29399},
}

@inproceedings{harkonen_gaussian_2023,
	title = {Gaussian process priors for systems of linear partial differential equations with constant coefficients},
	url = {https://proceedings.mlr.press/v202/harkonen23a.html},
	urldate = {2025-01-21},
	booktitle = {International conference on machine learning},
	publisher = {PMLR},
	author = {Harkonen, Marc and Lange-Hegermann, Markus and Raita, Bogdan},
	year = {2023},
	pages = {12587--12615},
}

@article{oberst,
  title={Multidimensional constant linear systems},
  author={Oberst, Ulrich},
  journal={Acta Applicandae Mathematica},
  volume={20},
  number={1},
  pages={1--175},
  year={1990},
  publisher={Springer}
}

@article{hillebrand,
  title={Towards an effective version of a theorem of Stafford},
  author={Hillebrand, Andre and Schmale, Wiland},
  journal={Journal of Symbolic Computation},
  volume={32},
  number={6},
  pages={699--716},
  year={2001},
  publisher={Elsevier}
}

@article{leykin2004,
  title={Algorithmic proofs of two theorems of Stafford},
  author={Leykin, Anton},
  journal={Journal of Symbolic Computation},
  volume={38},
  number={6},
  pages={1535--1550},
  year={2004},
  publisher={Elsevier}
}

@inproceedings{quadrat2006,
  title={Constructive computation of flat outputs of a class of multidimensional linear systems with variable coefficients},
  author={Quadrat, Alban and Robertz, Daniel},
  booktitle={Proceedings of MTNS},
  pages={24--28},
  year={2006}
}

@article{chyzak,
  title={Effective algorithms for parametrizing linear control systems over Ore algebras},
  author={Chyzak, Fr{\'e}d{\'e}ric and Quadrat, Alban and Robertz, Daniel},
  journal={Applicable Algebra in Engineering, Communication and Computing},
  volume={16},
  number={5},
  pages={319--376},
  year={2005},
  publisher={Springer}
}

@article{staffordog,
  author = {Stafford, J. T.},
  title = {Module structure of {W}eyl algebras},
  journal = {Journal of the London Mathematical Society. Second Series},
  volume = {18},
  year = {1978},
  number = {3},
  pages = {429--442},
  doi = {10.1112/jlms/s2-18.3.429}
}

@book{rotmancours,
  title={An introduction to homological algebra},
  author={Rotman, Joseph J},
  number={2},
  year={2009},
  publisher={Springer}
}

@Misc{macaulay2,
 title = {Macaulay2, a software system for research in algebraic geometry},
          author = {Grayson, Daniel R. and Stillman, Michael E.},
year = {},
         
          howpublished = {Available at \url{http://www2.macaulay2.com}}
        }

@inproceedings{gpytorch1,
  title={GPyTorch: Blackbox Matrix-Matrix Gaussian Process Inference with GPU Acceleration},
  author={Gardner, Jacob R and Pleiss, Geoff and Bindel, David and Weinberger, Kilian Q and Wilson, Andrew Gordon},
  booktitle={Advances in Neural Information Processing Systems},
  year={2018}
}

@article{noether,
author = {Noether, E.},
journal = {Mathematische Annalen},
pages = {26-61},
title = {Abstrakter Aufbau der Idealtheorie in algebraischen Zahl- und Funktionenkörpern},
url = {http://eudml.org/doc/159155},
volume = {96},
year = {1927},
}

@article{gradefiltration,
  title={Grade filtration of linear functional systems},
  author={Quadrat, Alban},
  journal={Acta Applicandae Mathematicae},
  volume={127},
  number={1},
  pages={27--86},
  year={2013},
  publisher={Springer}
}

@misc{gpytorch2,
	title = {{GPyTorch}: {Blackbox} {Matrix}-{Matrix} {Gaussian} {Process} {Inference} with {GPU} {Acceleration}},
	shorttitle = {{GPyTorch}},
	url = {http://arxiv.org/abs/1809.11165},
	doi = {10.48550/arXiv.1809.11165},
	urldate = {2025-07-01},
	publisher = {arXiv},
	author = {Gardner, Jacob R. and Pleiss, Geoff and Bindel, David and Weinberger, Kilian Q. and Wilson, Andrew Gordon},
	month = jun,
	year = {2021},
	note = {arXiv:1809.11165 [cs]},
}

@incollection{oremodules,
  title={OreModules: A symbolic package for the study of multidimensional linear systems},
  author={Chyzak, Fr{\'e}d{\'e}ric and Quadrat, Alban and Robertz, Daniel},
  booktitle={Applications of time delay systems},
  pages={233--264},
  year={2007},
  publisher={Springer}
}

@article{narcowich,
  title={Generalized Hermite interpolation via matrix-valued conditionally positive definite functions},
  author={Narcowich, Francis J and Ward, Joseph D},
  journal={Mathematics of Computation},
  volume={63},
  number={208},
  pages={661--688},
  year={1994},
  publisher={Washington, DC: National Academy of Sciences-National Research Council,[1960?-}
}

@INPROCEEDINGS{wahlstrom0,
  author={Wahlström, Niklas and Kok, Manon and Schön, Thomas B. and Gustafsson, Fredrik},
  booktitle={2013 IEEE International Conference on Acoustics, Speech and Signal Processing}, 
  title={Modeling magnetic fields using Gaussian processes}, 
  year={2013},
  volume={},
  number={},
  pages={3522-3526},
  doi={10.1109/ICASSP.2013.6638313}}

@inproceedings{mymaxent,
  author    = {Moser, Johanna and Albert, Christopher and Ranftl, Sascha},
  title     = {Parameter Learning With Physics-Consistent Gaussian Processes},
  booktitle = {44th International Workshop on Bayesian Inference and
               Maximum Entropy Methods in Science and Engineering},
  year      = {2026},
  month     = may,
  address   = {Auckland, New Zealand},
  doi       = {10.5281/zenodo.20392038},
  url       = {https://doi.org/10.5281/zenodo.20392038},
}

@misc{jax,
  author = {James Bradbury and Roy Frostig and Peter Hawkins and Matthew James Johnson and Chris Leary and Dougal Maclaurin and George Necula and Adam Paszke and Jake Vander{P}las and Skye Wanderman-{M}ilne and Qiao Zhang},
  title = {{JAX}: composable transformations of {P}ython+{N}um{P}y programs},
  howpublished = {\url{http://github.com/jax-ml/jax}},
  note = {Version 0.3.13},
  year = {2018},
}

@article{chris_und_patrick,
  title={2D Fourier finite element formulation for magnetostatics in curvilinear coordinates with a symmetry direction},
  author={Albert, Christopher G and Lainer, Patrick and Biro, Oszkar},
  journal={Computer Physics Communications},
  volume={277},
  pages={108401},
  year={2022},
  publisher={Elsevier}
}

@article{cours,
  title={An introduction to constructive algebraic analysis and its applications},
  author={Quadrat, Alban},
  journal={Les cours du CIRM},
  volume={1},
  number={2},
  pages={281--471},
  year={2010}
}

@misc{mypsdpaper,
      title={Resolving positive semi-definiteness in physics-informed kernels for scientific machine learning}, 
      author={J. Moser and C. Albert and S. Ranftl},
      year={2026},
      eprint={2608.20985},
      archivePrefix={arXiv},
      primaryClass={physics.comp-ph},
      url={https://arxiv.org/abs/2608.20985}, 
}

@article{fliess_controllability,
title = {Some basic structural properties of generalized linear systems},
journal = {Systems \& Control Letters},
volume = {15},
number = {5},
pages = {391-396},
year = {1990},
issn = {0167-6911},
doi = {https://doi.org/10.1016/0167-6911(90)90062-Y},
url = {https://www.sciencedirect.com/science/article/pii/016769119090062Y},
author = {Michel Fliess}
}

@book{zerz,
  title={Topics in multidimensional linear systems theory},
  author={Zerz, Eva},
  year={2000},
  publisher={Springer}
}

@article{pommaret1999algebraic,
  title={Algebraic analysis of linear multidimensional control systems},
  author={Pommaret, JF and Quadrat, A},
  journal={IMA Journal of Mathematical control and Information},
  volume={16},
  number={3},
  pages={275--297},
  year={1999},
  publisher={Oxford University Press}
}

@article{pinn,
  author       = {Raissi, Maziar and Perdikaris, Paris and Karniadakis, George Em},
  title        = {Physics-informed neural networks: A deep learning framework for solving forward and inverse problems involving nonlinear partial differential equations},
  doi          = {10.1016/j.jcp.2018.10.045},
  url          = {https://www.osti.gov/biblio/1595805},
  journal      = {Journal of Computational Physics},
  issn         = {ISSN 0021-9991},
  number       = {C},
  volume       = {378},
  place        = {United States},
  publisher    = {Elsevier},
  year         = {2018},
  month        = {11}}

@ARTICLE{willems_controllability,
  author={Willems, J. and Mitter, S.},
  journal={IEEE Transactions on Automatic Control}, 
  title={Controllability, observability, pole allocation, and state reconstruction}, 
  year={1971},
  volume={16},
  number={6},
  pages={582-595},
  doi={10.1109/TAC.1971.1099819}}

@article{burger,
  title={Modern regularization methods for inverse problems},
  author={Benning, Martin and Burger, Martin},
  journal={Acta numerica},
  volume={27},
  pages={1--111},
  year={2018},
  publisher={Cambridge University Press}
}

@book{jacobson,
  title={Basic algebra I},
  author={Jacobson, Nathan},
  year={2012},
  publisher={Courier Corporation}
}

@article{raissi_nonlinear,
author = {Raissi, Maziar and Perdikaris, Paris and Karniadakis, George Em},
title = {Numerical Gaussian Processes for Time-Dependent and Nonlinear Partial Differential Equations},
journal = {SIAM Journal on Scientific Computing},
volume = {40},
number = {1},
pages = {A172-A198},
year = {2018},
doi = {10.1137/17M1120762},
URL = {     
        https://doi.org/10.1137/17M1120762
},
eprint = { 
        https://doi.org/10.1137/17M1120762
},
}

@article{chen_nonlinear,
title = {Solving and learning nonlinear PDEs with Gaussian processes},
journal = {Journal of Computational Physics},
volume = {447},
pages = {110668},
year = {2021},
issn = {0021-9991},
doi = {https://doi.org/10.1016/j.jcp.2021.110668},
url = {https://www.sciencedirect.com/science/article/pii/S0021999121005635},
author = {Yifan Chen and Bamdad Hosseini and Houman Owhadi and Andrew M. Stuart}
}

@article{moss2026conditioning,
  title={Conditioning Gaussian Processes on Almost Anything},
  author={Moss, Henry and Astfalck, Lachlan and Cowperthwaite, Thomas and Doumont, Colin and Willis, Sam and Hennig, Philipp and Nemeth, Christopher and Zammit-Mangion, Andrew},
  journal={arXiv preprint arXiv:2605.21041},
  year={2026}
}

@article{boogaart,
  title={Kriging for processes solving partial differential equations},
  author={van den Boogaart, Karl Gerald},
  journal={Proceedings of the IAMG2001, Cancun, Mexiko},
  pages={10--12},
  year={2001}
}

@phdthesis{dong_diss,
  title={Estimation g{\'e}ostatistique des ph{\'e}nom{\`e}nes r{\'e}gis par des {\'e}quations aux d{\'e}riv{\'e}es partielles},
  author={Dong, Anne},
  year={1990},
  school={ENSMP}
}

@book{wahlstrom2015,
  title={Modeling of magnetic fields and extended objects for localization applications},
  author={Wahlstr{\"o}m, Niklas},
  year={2015},
  publisher={Linkopings Universitet}
}

@article{jidling2017linearly,
  title={Linearly constrained Gaussian processes},
  author={Jidling, Carl and Wahlstr{\"o}m, Niklas and Wills, Adrian and Sch{\"o}n, Thomas B},
  journal={Advances in neural information processing systems},
  volume={30},
  year={2017}
}

@article{berlinghieri2023gaussian,
  title={Gaussian processes at the Helm (holtz): A more fluid model for ocean currents},
  author={Berlinghieri, Renato and Trippe, Brian L and Burt, David R and Giordano, Ryan and Srinivasan, Kaushik and {\"O}zg{\"o}kmen, Tamay and Xia, Junfei and Broderick, Tamara},
  journal={arXiv preprint arXiv:2302.10364},
  year={2023}
}

@article{ding2019bdrygp,
  title={BdryGP: a new Gaussian process model for incorporating boundary information},
  author={Ding, Liang and Mak, Simon and Wu, CF},
  journal={arXiv preprint arXiv:1908.08868},
  year={2019}
}

@article{tan2018gaussian,
  title={Gaussian process modeling with boundary information},
  author={Tan, Matthias Hwai Yong},
  journal={Statistica Sinica},
  pages={621--648},
  year={2018},
  publisher={JSTOR}
}

@InProceedings{pmlr-v130-lange-hegermann21a,
  title = 	 { Linearly Constrained Gaussian Processes with Boundary Conditions },
  author =       {Lange-Hegermann, Markus},
  booktitle = 	 {Proceedings of The 24th International Conference on Artificial Intelligence and Statistics},
  pages = 	 {1090--1098},
  year = 	 {2021},
  editor = 	 {Banerjee, Arindam and Fukumizu, Kenji},
  volume = 	 {130},
  series = 	 {Proceedings of Machine Learning Research},
  month = 	 {13--15 Apr},
  publisher =    {PMLR},
  url = 	 {https://proceedings.mlr.press/v130/lange-hegermann21a.html}
}

@InProceedings{pmlr-v89-solin19a,
  title = 	 {Know Your Boundaries: Constraining Gaussian Processes by Variational Harmonic Features},
  author =       {Solin, Arno and Kok, Manon},
  booktitle = 	 {Proceedings of the Twenty-Second International Conference on Artificial Intelligence and Statistics},
  pages = 	 {2193--2202},
  year = 	 {2019},
  editor = 	 {Chaudhuri, Kamalika and Sugiyama, Masashi},
  volume = 	 {89},
  series = 	 {Proceedings of Machine Learning Research},
  month = 	 {16--18 Apr},
  publisher =    {PMLR},
  url = 	 {https://proceedings.mlr.press/v89/solin19a.html}
}

@article{gulian2022gaussian,
  title={Gaussian process regression constrained by boundary value problems},
  author={Gulian, Mamikon and Frankel, Ari and Swiler, Laura},
  journal={Computer Methods in Applied Mechanics and Engineering},
  volume={388},
  pages={114117},
  year={2022},
  publisher={Elsevier}
}

@article{swiler2020survey,
  title={A survey of constrained Gaussian process regression: Approaches and implementation challenges},
  author={Swiler, Laura P and Gulian, Mamikon and Frankel, Ari L and Safta, Cosmin and Jakeman, John D},
  journal={Journal of Machine Learning for Modeling and Computing},
  volume={1},
  number={2},
  year={2020},
  publisher={Begel House Inc.}
}

@article{hamelijnck2024physics,
  title={Physics-informed variational state-space gaussian processes},
  author={Hamelijnck, Oliver and Solin, Arno and Damoulas, Theodoros},
  journal={Advances in Neural Information Processing Systems},
  volume={37},
  pages={98505--98536},
  year={2024}
}

@inproceedings{lean,
  title={The Lean 4 Theorem Prover and Programming Language},
  author={Leonardo Mendonça de Moura and Sebastian Ullrich},
  booktitle={CADE},
  year={2021},
  url={https://api.semanticscholar.org/CorpusID:235800962}
}

@misc{dmodules,
  author = {Schnell, Christian},
  title = {{$D$}-modules},
  year = {2019},
  note = {Lecture notes},
  url = {https://www.math.stonybrook.edu/~cschnell/pdf/notes/d-modules.pdf}
}

@misc{chris_stafford,
  author = {{Christopher Albert}},
  title = {{Stafford 3.8}},
  year = {2026},
  howpublished = {Palomar, PALOMAR-2026-09-05-000007 v2},
  url = {https://palomar-registry.org/entry?id=PALOMAR-2026-09-05-000007&version=2},
}
\appendix
\section{Computer algebra syntax}
\label{ap:macaulay}
Here we provide the \texttt{Macaulay2} and \texttt{OreModules} syntax used to construct a ghost column $\Lambda(\vv x)$ and to check the validity of a proposed one. For the full scripts, we refer to our github page. As before, we drop the index $\phi$ for physical parameters and the argument $(\vv x)$ for readability. Full tutorials and detailed discussions on \texttt{OreModules} can be found on their website\footnote{\url{https://who.rocq.inria.fr/Alban.Quadrat/OreModules/index.html} Accessed on 30.7.2026}.

\newenvironment{macaulay}
{\verbbox}
{\endverbbox\par\noindent\colorbox{Thistle}{\parbox{\textwidth}{\theverbbox}}\par}

\newenvironment{ore}
{\verbbox}
{\endverbbox\par\noindent\colorbox{Apricot}{\parbox{\textwidth}{\theverbbox}}\par}

\paragraph{Defining the algebra and the system}
\begin{macaulay}
--Macaulay
needsPackage "Dmodules"
D = frac(QQ[a])[t, x, dt, dx, WeylAlgebra => {x=>dx, t=> dt}]
R = matrix{{x*dx + dt, a*dx}, {0, x*dx + dt}} 
\end{macaulay}
\begin{ore}
--OreModules
with(OreModules):
with(Ore_algebra):
D := DefineOreAlgebra(diff = [dt, t], diff = [dx, x], 
                      polynom = {t, x}, comm = {a})
R := evalm([[dt + x*dx, a*dx], [0, dt + x*dx]])
\end{ore}
Here hides a common pitfall: a system may be parametrizable for generic values of commuting constants (such as $a$ here), but the parametrization fails for a special combination of them (such as in \cite{besginow_constraining_2022}). Both packages here certify only the generic case, so if the constants are known numerically beforehand, it is safer to substitute them before certifying, precisely because a special value can break a generically parametrizable system.

\paragraph{Finding $\Lambda$}
The columns of $\Lambda$ must be chosen so that their residues generate this module (cf.\ Theorem \ref{th:GT}). If $\mathrm{ext}^1_D(M,D)$ is \emph{holonomic}, it is automatically cyclic \cite{serre}, and \texttt{Macaulay2} can produce a cyclic generator algebraically:
\begin{macaulay}
ext1M = coker Dtransposition R;
isHolonomic ext1M --if true, the next line works
Lambda = Dtransposition (makeCyclic Dtransposition R).Generator
\end{macaulay}
\texttt{Dtransposition} implements the involution map needed to map from left to right modules and inversely, needed to use  \texttt{Macaulay2}'s native Gr\"obner machinery. 

For a non-holonomic $\mathrm{ext}^1_D(M,D)$ (as in experiment 1 and 3), we do not know of a constructive method that is implemented for extension modules. The fallback is to loop through a set of candidate operators and test each candidate with the projectivity or torsion-freeness check given below, stopping at the first success; this has worked reasonably well in practice. The script running the loop can be found on our github page. 

Once we have a candidate, we define $P = \begin{pmatrix}R & {-}\Lambda\end{pmatrix}.$
\paragraph{Checking torsion-freeness} $M$ is torsion-free iff $\mathrm{ext}^1_D(N,D) = 0$ ((1) of Lemma \ref{lem:exteigenschaften}). In \texttt{Macaulay2}, we explicitly calculate $\mathrm{ext}^1_D(N,D)$: 
\begin{macaulay}
strip = X -> map((ring X)^(numRows X), (ring X)^(numColumns X), entries X)
K1 := strip syz P;
Q1 := strip transpose Dtransposition P;
Q2 := strip transpose Dtransposition K1;
ext1N := prune (kernel Q2 / image Q1);
ext1N == 0 --if true, then M is torsion-free
\end{macaulay}

In \texttt{OreModules} the analogous call is
\begin{ore}
Exti(Involution(P, D), D, 1);
\end{ore}
if the first of its three returned matrices is the identity, $M$ is torsion-free. 

\paragraph{Checking projectivity} To test projectivity, it is actually easiest to test if $M$ is stably free, and thus by implication also projective. Lemma \ref{lem:exteigenschaften} tells us, that $M$ is stably free iff $\mathrm{ext}^1_D(M,D) = 0$. In \texttt{Macaulay2} we can test this by checking whether $D^q \subseteq PD^{p+m}$, i.e.\ whether $\mathrm{ext}^1_D(M,D) = D^q/PD^{p+m} = 0$
\begin{macaulay}
isSubset(image id_(target Dtransposition P), image Dtransposition P)
\end{macaulay}
Alternatively, we can test if $R$ has a right inverse $S$ with $PS = \mathrm{id}_q$ ((3) of Lemma \ref{lem:exteigenschaften}). 
\begin{macaulay}
W = id_(target Dtransposition P) // Dtransposition P 
W =!= null and (Dtransposition P)*W == id_(target Dtransposition P) 
S = Dtransposition W
\end{macaulay}
In \texttt{OreModules}, the analogous call is to calculate
\begin{ore}
RightInverse(P, D)
\end{ore}
Alternatively, we can check (2) of Lemma \ref{lem:exteigenschaften} by calculating the first $n$ extension modules $\mathrm{ext}^i_D(N,D)$ (with $n$ being the number of variables in the chosen Weyl algebra $A_n(\mathbb{R})$).
\begin{ore}
Exti(Involution(R, D), D, n);
\end{ore}
the first matrix returned by each of the $n$ calls is the presentation matrix $R_{\mathrm{ext}^i}$ of $\mathrm{ext}^i_D(M,D)$; if every first matrix is the identity, all $n$ extension modules vanish and the system is projective.

\paragraph{Minimal parametrization}
Once we have checked torsion-freeness or projectivity, we need to calculate the corresponding parametrization matrix $Q$. 
In \texttt{Macaulay2}, we can build a parametrization matrix by calculating the generators of the right null space of $P$
\begin{macaulay}
    Q = Dtransposition mingens image syz Dtransposition P 
    transpose Q*(transpose P) -- test: should be 0
\end{macaulay}
We note that \texttt{Macaulay2} has a reversed product convention, so to check $PQ = 0$, we actually need to check $Q^TP^T = 0$. 
Based on this parametrization, we can also check torsion-freeness analogous to what \cite{lange-hegermann_algorithmic_2018} has proposed for commutative algebras by calculating the left null space $P_{new}$ of $Q$ an checking if it presents the same module as $P$ does:
\begin{macaulay}
    Pnew = transpose mingens image syz transpose Q
\end{macaulay}
We note that the parametrization that we can calculate with \texttt{Macaulay2} does not necessarily have the minimal number of columns. 
In \texttt{OreModules}, we can find a minimal parametrization in \texttt{OreModules} by calling
\begin{ore}
MinimalParametrization(P, D)
\end{ore}

\section{Ablation studies}
\label{ap:ablation}
In ghost tasking, but also in PIGP for homogeneous equations, we need to condition on tasks that we know are exactly zero everywhere. We want to know if adding extra collocation points to condition on for the zero-valued tasks improve the results compared to the "standard" variant of having the same amount of points in each task. Additionally, since we know the collocation points to be exactly zero, we choose to constrain collocation points in exactly known tasks to a small noise nugget of $\sigma = 10^{-3}$. We refer to this method with extra, exactly known collocation points with a "+" (e.g. "GT+" and "PIGP+"). 
We show the median of the errors of the learned MAP estimates of the physical parameters over the 10 runs relative to their true values in the inverse settings of all experiments in Tab. \ref{tab:Ex1}-\ref{tab:Ex3} for the standard and the "+" variants. The variant that performed best across methods is marked with bold, and the best variant within a method is marked in italics. In the last rows, we show the overall performance calculated by the geometric mean of the table entries. 

In our first experiment (chapter \ref{sec:Ex2}), "GT+" shows strong improvements especially for the ratio of the parameters $g/\ell$, which is strongly related to the eigenfrequency of the system. As discussed, this is likely due to the original GT method absorbing physically relevant oscillations in the ghost task, which can be suppressed with "GT+". "PIGP+" cannot be defined, since per construction, we measure its inputs $f_i(t)$ alongside the other tasks and can therefore not provide arbitrarily many collocation points as in GT. 
This is an inherent advantage of GT, but for the sake of the ablation studies, we consider a new variant which we will call "known $f(t)$": We assume the inhomogeneity $f(t)$ is fully known and can therefore be conditioned on as many extra collocation points as the ghost task. In this case, we condition both $\xi_{GT}$ and $f(t)$ for GT, and $f_1(t)$, $f_2(t)$, and $f_3(t)$ for PIGP on 25 noiseless input points. In this limit, GT and PIGP perform similarly and we recover the multimodal posterior distribution of $\ell$ and $g$ also with PIGP (Fig. \ref{fig:Ex2_posteriors_knownu}. For the other two experiments, adding extra collocation points does not lead to similarly strong improvements. 

\begin{table}[]
\sisetup{mode=match, round-mode=figures, round-precision=2, scientific-notation=true, detect-weight=true, detect-shape=true, detect-family=true}
\caption{Ablation study experiment 1. Median absolute error of the physical parameters (each cell: $\frac{\Delta \ell}{\ell}$ / $\frac{\Delta g}{g}$ / $\frac{\Delta g/\ell}{g/\ell}$) relative to their true values $(\ell, g) = (2.5, 9.81)$ across 10 runs. We compare the variants "standard", where every task has the same amount of training data $n$, and "+", where the ghost task is constrained on $25$ points and the standard deviation of the likelihood in these tasks is fixed to be $\sigma = 0.001$. We compare two settings "measured $f(t)$", where $f(t)$ is provided with $n$ measurements, and "known $f(t)$", where $f(t)$ is provided with 25 noise-free input points for all settings. For PIGP, we cannot provide a "+" variant for "measured $f(t)$" per definition. Overall performance is calculated using the geometric mean of the above results and shown in the last row. We see similar performances between GT and PIGP, while "GT+" outperforms both drastically especially for $g/\ell$ in the measured case. For known $f(t)$, performances are similar.} 
\label{tab:Ex2}
\resizebox{\textwidth}{!}{%
\begin{tabular}{llccccc}
\toprule
$\frac{\Delta \ell}{\ell}/\frac{\Delta g}{g}/\frac{\Delta g/\ell}{g/\ell}$ & & \multicolumn{3}{c}{measured $f(t)$} & \multicolumn{2}{c}{known $f(t)$} \\
\cmidrule(lr){3-5} \cmidrule(lr){6-7}
 & & \multicolumn{2}{c}{\textbf{GT}} & PIGP & \textbf{GT} & PIGP \\
\cmidrule(lr){3-4} \cmidrule(lr){5-5} \cmidrule(lr){6-6} \cmidrule(lr){7-7}
n & $\sigma$ & standard & + & standard &  &   \\
\midrule
3 & 0.001 & \num{0.596268}/\num{0.231035}/\num{0.228804} & {\bfseries \num{0.0551756}}/{\bfseries \num{0.0551655}}/{\bfseries \num{4.49098e-05}} & \num{0.0663496}/\num{0.862723}/\num{0.852972} & {\bfseries \num{0.000308705}}/{\bfseries \num{0.000353493}}/\num{3.52798e-05} & \num{0.000328083}/\num{0.000372496}/{\bfseries \num{3.3207e-05}} \\
3 & 0.01 & \num{0.596267}/{\bfseries \num{0.231478}}/\num{0.228526} & {\itshape \num{0.263984}}/\num{0.263757}/{\bfseries \num{0.000765588}} & {\bfseries \num{0.0658151}}/\num{0.862947}/\num{0.853053} & {\bfseries \num{0.00561553}}/\num{0.00513256}/\num{0.000447494} & \num{0.00562678}/{\bfseries \num{0.00510016}}/{\bfseries \num{0.000404498}} \\
3 & 0.1 & \num{0.596229}/{\bfseries \num{0.231431}}/\num{0.228536} & {\itshape \num{0.57447}}/\num{0.575996}/{\bfseries \num{0.00219645}} & {\bfseries \num{0.0675862}}/\num{0.8603}/\num{0.851068} & \num{0.111937}/\num{0.102848}/\num{0.00696708} & {\bfseries \num{0.0964749}}/{\bfseries \num{0.0892782}}/{\bfseries \num{0.00434157}} \\
3 & 0.2 & \num{0.596577}/{\bfseries \num{0.236963}}/\num{0.225217} & {\itshape \num{0.594656}}/\num{0.596028}/{\bfseries \num{0.00783146}} & {\bfseries \num{0.0643401}}/\num{0.870774}/\num{0.860306} & \num{0.231128}/\num{0.21009}/\num{0.0183854} & {\bfseries \num{0.219964}}/{\bfseries \num{0.200549}}/{\bfseries \num{0.0153837}} \\
\addlinespace
5 & 0.001 & \num{0.708796}/\num{0.719209}/\num{0.0355823} & {\bfseries \num{0.055105}}/{\bfseries \num{0.0551411}}/{\bfseries \num{5.89572e-05}} & \num{0.488758}/\num{0.422138}/\num{0.130266} & \num{0.000514962}/{\bfseries \num{0.000500515}}/\num{5.70447e-05} & {\bfseries \num{0.000513011}}/\num{0.000501817}/{\bfseries \num{5.63325e-05}} \\
5 & 0.01 & \num{0.709609}/\num{0.719732}/\num{0.0376355} & {\bfseries \num{0.0186755}}/{\bfseries \num{0.0180059}}/{\bfseries \num{0.000630084}} & \num{0.489115}/\num{0.422694}/\num{0.130524} & {\bfseries \num{0.00190377}}/\num{0.00127412}/\num{0.000652329} & \num{0.00201696}/{\bfseries \num{0.00123779}}/{\bfseries \num{0.000652296}} \\
5 & 0.1 & \num{0.682977}/\num{0.728831}/\num{0.130991} & {\itshape \num{0.680969}}/{\itshape \num{0.678849}}/{\bfseries \num{0.00529511}} & {\bfseries \num{0.51666}}/{\bfseries \num{0.450707}}/\num{0.133448} & \num{0.073338}/\num{0.0668064}/{\bfseries \num{0.00734228}} & {\bfseries \num{0.0645341}}/{\bfseries \num{0.0584159}}/\num{0.00758037} \\
5 & 0.2 & \num{0.688832}/\num{0.744617}/\num{0.161509} & {\itshape \num{0.659868}}/{\itshape \num{0.65802}}/{\bfseries \num{0.0111512}} & {\bfseries \num{0.554321}}/{\bfseries \num{0.494583}}/\num{0.131474} & \num{0.218723}/\num{0.202936}/\num{0.0125428} & {\bfseries \num{0.199967}}/{\bfseries \num{0.184742}}/{\bfseries \num{0.0117897}} \\
\addlinespace
10 & 0.001 & \num{0.0005541}/\num{0.00053243}/\num{6.40572e-05} & {\bfseries \num{0.000353224}}/{\bfseries \num{0.000393609}}/{\itshape \num{3.26283e-05}} & \num{0.00043175}/\num{0.000418039}/{\bfseries \num{2.4065e-05}} & \num{0.000304575}/{\bfseries \num{0.000330387}}/\num{2.71255e-05} & {\bfseries \num{0.000303512}}/\num{0.000340104}/{\bfseries \num{2.03579e-05}} \\
10 & 0.01 & \num{0.00245328}/\num{0.00252543}/\num{0.000393967} & {\bfseries \num{0.0019588}}/{\bfseries \num{0.00204219}}/{\bfseries \num{0.000291407}} & \num{0.00198631}/\num{0.00210454}/\num{0.000418856} & {\bfseries \num{0.00241212}}/\num{0.00254801}/{\bfseries \num{0.000349406}} & \num{0.00245741}/{\bfseries \num{0.00251442}}/\num{0.000350032} \\
10 & 0.1 & \num{0.0593031}/\num{0.0619563}/\num{0.0023064} & {\itshape \num{0.0401823}}/{\itshape \num{0.0408355}}/{\bfseries \num{0.00227542}} & {\bfseries \num{0.0278957}}/{\bfseries \num{0.0285181}}/\num{0.00239935} & \num{0.0380358}/\num{0.0386261}/{\bfseries \num{0.00227954}} & {\bfseries \num{0.0344958}}/{\bfseries \num{0.0355603}}/\num{0.00228054} \\
10 & 0.2 & \num{0.142673}/\num{0.145875}/{\bfseries \num{0.00323054}} & {\itshape \num{0.119392}}/{\itshape \num{0.113933}}/\num{0.00638651} & {\bfseries \num{0.0821377}}/{\bfseries \num{0.0779711}}/\num{0.00638667} & \num{0.104084}/\num{0.101338}/\num{0.00460427} & {\bfseries \num{0.0943625}}/{\bfseries \num{0.0933007}}/{\bfseries \num{0.00444603}} \\
\midrule
\multicolumn{2}{l}{\textbf{overall}} & \num{0.162704}/\num{0.121165}/\num{0.022201} & {\itshape \num{0.0666162}}/{\bfseries \num{0.0670696}}/{\bfseries \num{0.000866538}} & {\bfseries \num{0.0608057}}/\num{0.136916}/\num{0.041289} & \num{0.010617}/\num{0.0100925}/\num{0.000969479} & {\bfseries \num{0.0102012}}/{\bfseries \num{0.00964358}}/{\bfseries \num{0.000879039}} \\
\bottomrule
\end{tabular}%
}
\end{table}
\begin{figure}[h!]
\centering
  \begin{minipage}{0.45\textwidth}
  \includegraphics[width=\textwidth]{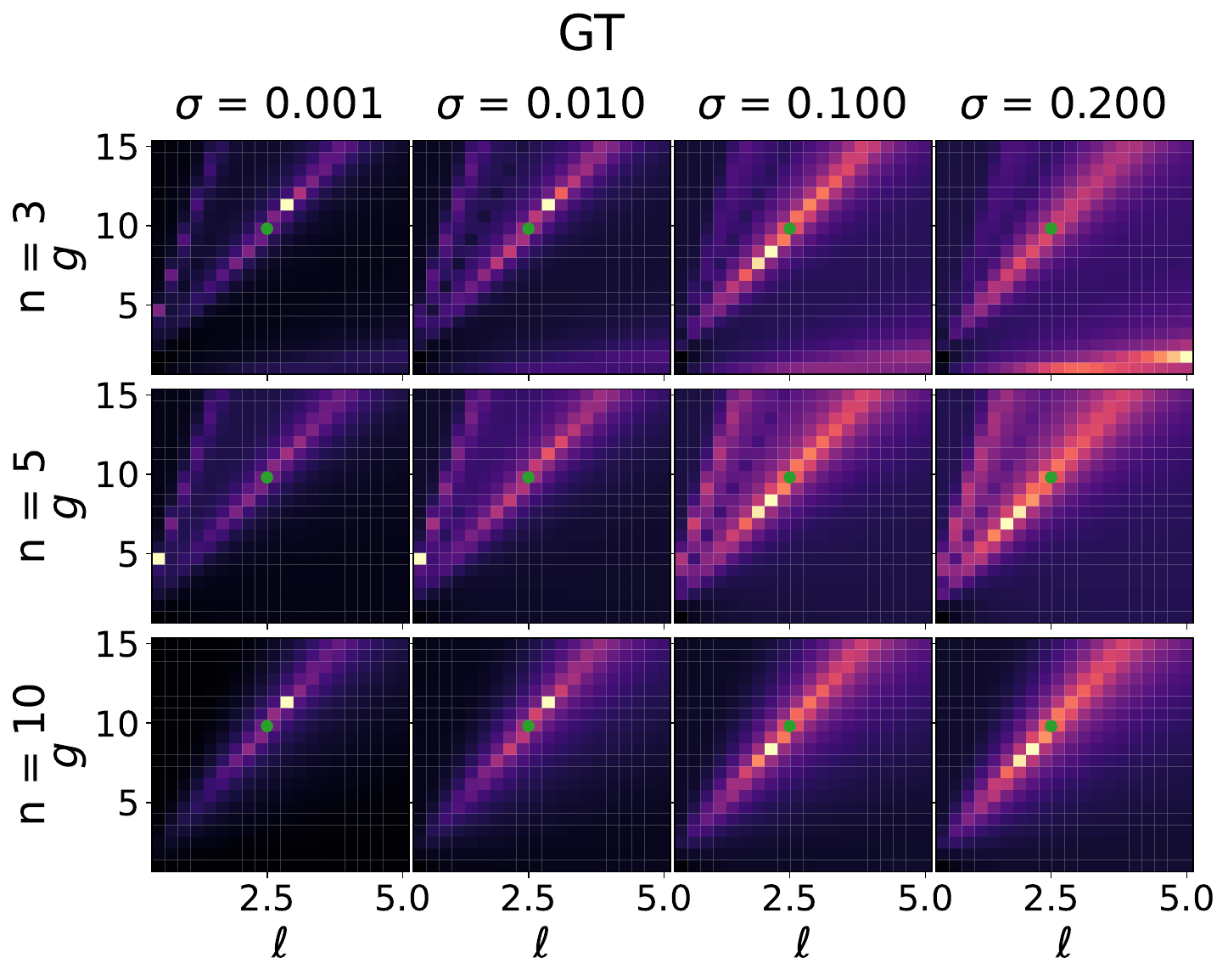}
  \end{minipage}
  \begin{minipage}{0.45\textwidth}
  \includegraphics[width=\textwidth]{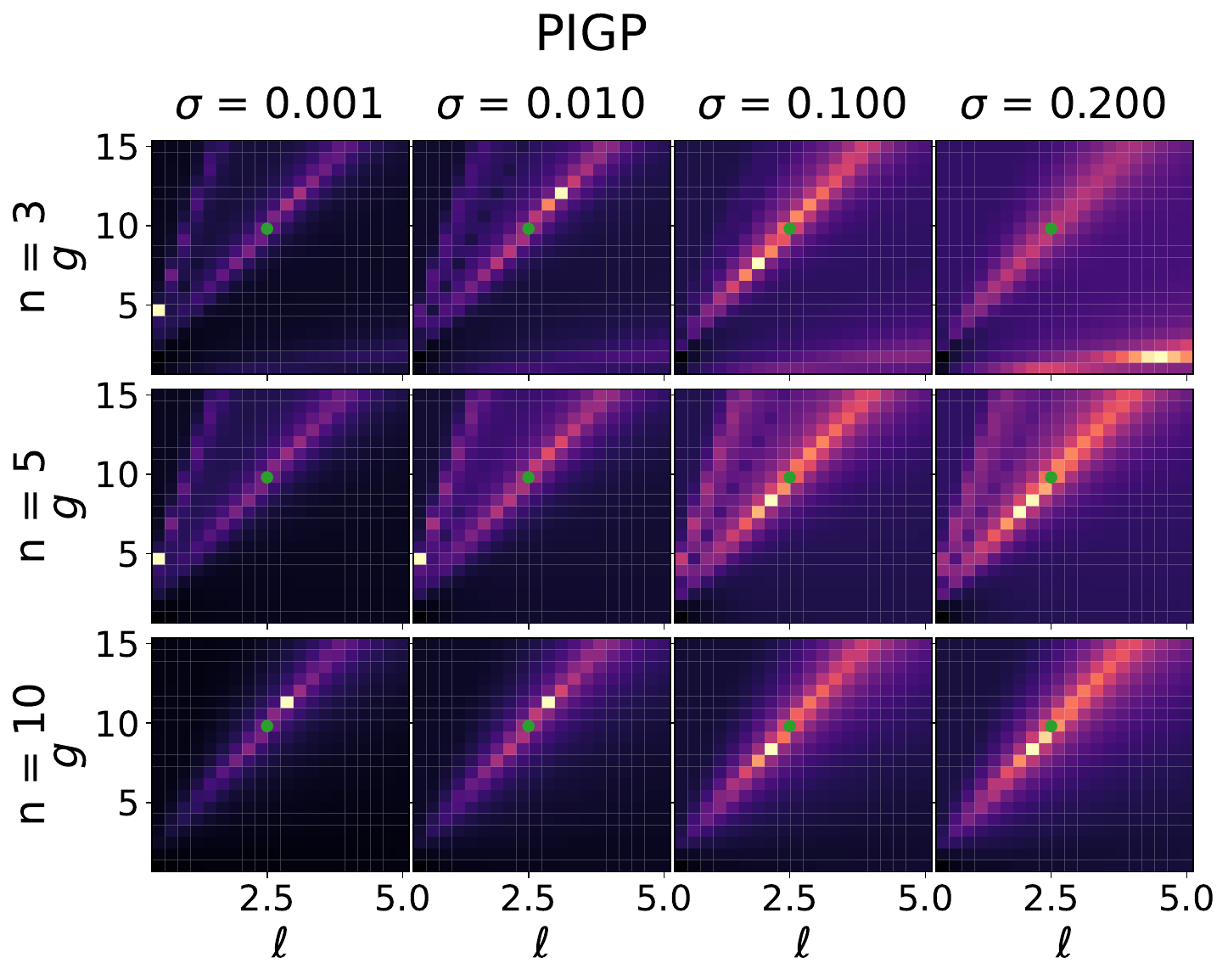}
  \end{minipage}
\caption{Re-scaled log posteriors of $g, \ell$ analogously to Fig. \ref{fig:Ex2_posteriors} assuming $f(t)$ is known. We recover similar multimodal structures as we did in the original experiment, but this time for both methods.}
\label{fig:Ex2_posteriors_knownu}
\end{figure}

\begin{table}[]
\sisetup{mode=match, round-mode=figures, round-precision=2, scientific-notation=true, table-format=1.2e-1, detect-weight=true, detect-shape=true, detect-family=true}
\caption{Ablation study experiment 2. Median absolute error of the physical parameter $\alpha$ relative to its true value $\alpha_{true} = 2$, analogously to \ref{tab:Ex2}. For the "+" variant, the ghost task resp. the inhomogeneities (for GT resp. PIGP) are constrained on $15\times 15$ points and the standard deviation of the likelihood in these tasks is fixed to be $\sigma = 0.001$. Contrary to experiment 1 (\ref{tab:Ex2}), we do not see significant differences between variants in the results.}
\label{tab:Ex1}
\resizebox{\textwidth}{!}{%
\begin{tabular}{llSSSSSS}
\toprule
{$\frac{\Delta \alpha}{\alpha}$} & & \multicolumn{2}{c}{\textbf{GT (torsion-free)}} & \multicolumn{2}{c}{\textbf{GT (projective)}} & \multicolumn{2}{c}{PIGP} \\
\cmidrule(lr){3-4} \cmidrule(lr){5-6} \cmidrule(lr){7-8}
n & {$\sigma$} & {standard} & {+} & {standard} & {+} & {standard} & {+} \\
\midrule
3 & 0.001 & 0.0250089 & \itshape 0.000663676 & 0.00762962 & \itshape 0.000510805 & 0.00582949 & \bfseries 0.000355886 \\
3 & 0.01 & 0.0251172 & \bfseries 0.00308351 & \itshape 0.0103383 & 0.0125361 & 0.00865 & \itshape 0.00339176 \\
3 & 0.1 & \itshape 0.0360237 & 0.0435996 & \bfseries 0.0333033 & 0.0830999 & \itshape 0.044 & 0.109045 \\
3 & 0.2 & \bfseries 0.0953646 & 0.255076 & \itshape 0.168597 & 0.725525 & \itshape 0.109904 & 0.641046 \\
\addlinespace
5 & 0.001 & 0.000476128 & \itshape 0.000189513 & 0.000152691 & \bfseries 0.000141994 & 0.000345346 & \itshape 0.000185464 \\
5 & 0.01 & 0.00184069 & \bfseries 0.0016835 & 0.00225342 & \itshape 0.00194024 & \itshape 0.00182672 & 0.00183743 \\
5 & 0.1 & \itshape 0.0179435 & 0.0189678 & \itshape 0.0189525 & 0.0323223 & \bfseries 0.012887 & 0.0264649 \\
5 & 0.2 & \bfseries 0.012287 & 0.0570769 & \itshape 0.044273 & 0.178954 & \itshape 0.0312465 & 0.302208 \\
\addlinespace
10 & 0.001 & \itshape 8.99094e-05 & 0.000101184 & 8.7346e-05 & \bfseries 8.6918e-05 & 9.30985e-05 & \itshape 9.02278e-05 \\
10 & 0.01 & 0.00061429 & \bfseries 0.000471394 & 0.00122068 & \itshape 0.00117363 & 0.0012042 & \itshape 0.00120184 \\
10 & 0.1 & 0.0116274 & \bfseries 0.00634416 & 0.0108223 & \itshape 0.00901841 & 0.0102954 & \itshape 0.0101971 \\
10 & 0.2 & \bfseries 0.0236151 & 0.0246172 & \itshape 0.0470936 & 0.0701019 & \itshape 0.0393625 & 0.0776159 \\
\midrule
\multicolumn{2}{l}{\textbf{overall}} & 0.00655762 & \bfseries 0.00443914 & \itshape 0.00660932 & 0.00764208 & \itshape 0.0061298 & 0.0072101 \\
\bottomrule
\end{tabular}}
\end{table}

\begin{table}[]
\sisetup{mode=match, round-mode=figures, round-precision=2, scientific-notation=true, table-format=1.2e-1, detect-weight=true, detect-shape=true, detect-family=true}
\caption{Ablation study experiment 3. Median absolute error of the physical parameter $\nu_0$ relative to its true value $\nu_{0true} = 1$, analogously to \ref{tab:Ex2}. For the "+" variant, the ghost task resp. the inhomogeneities (for GT resp. PIGP) are constrained on $10\times 10\times 10$ points and the standard deviation of the likelihood in these tasks is fixed to be $\sigma = 0.001$. GT outperforms PIGP in every setting. Contrary to experiment 1 (\ref{tab:Ex2}), we do not see significant differences between variants in the results.}
\label{tab:Ex3}
\begin{tabular}{llSSSS}
\toprule
{$\frac{\Delta \nu_0}{\nu_0}$} & & \multicolumn{2}{c}{\textbf{GT}} & \multicolumn{2}{c}{PIGP} \\
\cmidrule(lr){3-4} \cmidrule(lr){5-6}
n & {$\sigma$} & {standard} & {+} & {standard} & {+} \\
\midrule
3 & 0.001 & 0.10125 & \bfseries 0.0809063 & \itshape 0.342445 & 0.483184 \\
3 & 0.01 & 0.0960358 & \bfseries 0.0731214 & \itshape 0.345774 & 0.494206 \\
3 & 0.1 & 0.0834565 & \bfseries 0.0539069 & \itshape 0.386825 & 0.525768 \\
3 & 0.2 & 0.0778331 & \bfseries 0.0469981 & \itshape 0.434896 & 0.551056 \\
\addlinespace
5 & 0.001 & 0.000944797 & \bfseries 0.000884615 & \itshape 0.32439 & 0.377715 \\
5 & 0.01 & \bfseries 0.000722241 & 0.000837045 & 0.359478 & \itshape 0.279961 \\
5 & 0.1 & 0.00377455 & \bfseries 0.00332554 & 0.373856 & \itshape 0.193678 \\
5 & 0.2 & 0.00500214 & \bfseries 0.0046573 & 0.372423 & \itshape 0.1768 \\
\addlinespace
7 & 0.001 & 0.0123503 & \bfseries 0.0120364 & 0.11204 & \itshape 0.110139 \\
7 & 0.01 & 0.000116718 & \bfseries 9.69767e-05 & 0.168336 & \itshape 0.100939 \\
7 & 0.1 & 0.000431414 & \bfseries 0.000415067 & 0.311 & \itshape 0.205141 \\
7 & 0.2 & \bfseries 0.00106767 & 0.00117856 & 0.340092 & \itshape 0.157765 \\
\midrule
\multicolumn{2}{l}{\textbf{overall}} & 0.0053422 & \bfseries 0.00463437 & 0.304887 & \itshape 0.258767 \\
\bottomrule
\end{tabular}

\end{table}

\section{Application to systems with general non-constant coefficients}
\label{app:commutative}
We can find parametrizations even for certain DEs with non-constant coefficents via commutative polynomial rings. Such \emph{commuting cases} can be found for instance in Example~4.2 of \cite{lange-hegermann_algorithmic_2018} and Example~1 of \cite{mymaxent}, the latter of which we shall use as an illustrative example.
If we try to find a parametrization for the system represented by 
\begin{align}
R(\vv x) = \begin{pmatrix}
        \nabla^2+g_1(\vv x) & g_2(\vv x)  & -1\\[0.3em]
        g_3(\vv x)  & \nabla^2 + g_4(\vv x)  & 0
    \end{pmatrix}.
\end{align} with general functions $g_i(\vv x) = g_i(x_1, x_2)$, we could do so by assuming $g_i(\vv x)$ as (implicitly constant) parameters and using the Macaulay2 algorithm over commutative rings as in \cite{lange-hegermann_algorithmic_2018}, and obtain a parametrization of the form
\begin{align}
Q(\vv x) = \begin{pmatrix}
        g_4(\vv x) + \nabla^2\\[0.3em]
        -g_3(\vv x)\\[0.3em]
        -g_2(\vv x)g_3(\vv x) + g_1(\vv x)g_4 (\vv x)
        + \nabla^2\!\bigl(g_1(\vv x)+g_4(\vv x)\bigr) 
        + (\nabla^2)^2
\end{pmatrix}.
\end{align}

Evaluating $R(\vv x)Q(\vv x)$ yields
\begin{align}
R(\vv x)Q(\vv x) = \begin{pmatrix}
        \nabla^2 g_4(\vv x) - g_4(\vv x)\nabla^2 \\[0.3em]
        -\nabla^2 g_3(\vv x) + g_3(\vv x)\nabla^2
    \end{pmatrix} 
    \overset{!}{=} \begin{pmatrix} 0 \\ 0 \end{pmatrix},
\end{align}
from which we deduce that the parametrization is valid precisely when $g_3(\vv x)$ and $g_4(\vv x)$ commute with $\nabla^2$, i.e., they must be constants. Notably, no such commutativity constraints apply to $g_1(\vv x)$ and $g_2(\vv x)$. 

This commutativity requirement is a strong restriction: in general applications, $g_3(\vv x)$ and $g_4(\vv x)$ will depend on their arguments, and thus the commutative parametrization is not available. Nonetheless, when the elements of $R(\vv x)$ commute with the elements of $Q(\vv x)$, this approach provides a convenient and computationally inexpensive way to obtain a parametrization using standard tools for commutative polynomial rings.

\end{document}